\documentclass[a4paper,10pt]{article}

\usepackage[utf8]{inputenc}
\usepackage[colorlinks=true]{hyperref}
\usepackage[margin=3cm]{geometry}
\usepackage{graphicx}
\graphicspath{{figures/}}
\usepackage{array}
\usepackage{multirow}
\usepackage{mathtools,amssymb}
\usepackage{amsmath,amsthm}
\usepackage{cite}
\usepackage{subfigure}
\numberwithin{equation}{section}

\usepackage[capitalise,noabbrev]{cleveref}
\usepackage[title]{appendix}

\newtheorem{theorem}{Theorem}[section]
\newtheorem{lemma}[theorem]{Lemma}
\newtheorem{remark}{Remark}

\def\R{\mathbb{R}}
\def\N{\mathbb{N}}
\def\Z{\mathbb{Z}}
\def\C{\mathbb{C}}

\def\Hspace{H}
\def\Lspace{L}

\def\vv{\mathbf{v}}

\def\k{\mathbf{k}}
\def\x{\mathbf{x}}
\def\F{\mathbf{F}}
\def\B{\mathbf{B}}

\def\tvv{\tilde{\mathbf{v}}}
\def\teta{\tilde\eta}

\def\e{{e}}
\def\ee{\mathbf{e}}
\def\E{\mathbf{E}}
\def\tee{\widetilde{\mathbf{e}}}
\def\tE{\widetilde{\mathbf{E}}}

\def\calL{\mathcal L}
\def\calK{\mathcal K}
\def\calS{\mathcal S}

\def\calO{\mathcal O}
\def\calN{\mathcal N}
\def\calR{\mathcal R}

\def\zz{\chi}

\def\lb{{C}}

\def\kx{{k_x}}
\def\ky{{k_y}}
\def\kcx{k'_x}
\def\kcy{k'_y}

\def\ampA{A}

\def\setM{\mathcal M}
\def\crit{{\rm c}}
\def\per{{\rm per}}
\def\eps{\varepsilon}
\def\s{{\rm s}}
\def\tf{{\tilde f}}
\def\ns{{\rm ns}}
\def\cc{{\rm c.c.}}

\def\Id{\mathrm{Id}}
\def\Re{\mathrm{Re}}
\def\Im{\mathrm{Im}}
\def\rmi{\mathrm{i}}
\def\rme{\mathrm{e}}
\def\rmr{{\mathrm r}}
\def\dif{\mathrm{d}}
\def\range{\mathrm{range}}

\def\sgn{{\rm sgn}}

\def\td{\tilde{d}}

\title{Rotating shallow water equations with bottom drag:\\ bifurcations and growth due to kinetic energy backscatter
\thanks{This paper is a contribution to the project M2 (Systematic multi-scale modelling and analysis for geophysical flow) of the Collaborative Research Centre TRR 181 ``Energy Transfers in Atmosphere and Ocean" funded by the Deutsche Forschungsgemeinschaft (DFG, German Research Foundation) under project number 274762653. J.Y. is partially supported by the Fundamental Research Funds for the Central Universities through Sun Yat-sen University (grant no.\ 22qntd2901).}}

\author{Artur Prugger\thanks{University of Bremen, Department 3 - Mathematics, 28359 Bremen, Germany
  ({\tt a.prugger@uni-bremen.de}).}
\and Jens D. M. Rademacher\thanks{Universit\"at Hamburg, MIN faculty, Department of Mathematics, Germany ({\tt jens.rademacher@uni-hamburg.de}).} 
\and Jichen Yang\thanks{Sun Yat-sen University, School of Mathematics (Zhuhai), 519082 Zhuhai, China
  ({\tt yangjch36@mail.sysu.edu.cn}).}
}

\date{April 27, 2023}

\begin{document}

\maketitle

\begin{abstract}
The rotating shallow water equations with f-plane approximation and nonlinear bottom drag are a prototypical model for mid-latitude geophysical flow that experience  energy loss through simple topography.  Motivated by numerical schemes for large-scale geophysical flow, we consider this model on the whole space with horizontal kinetic energy backscatter source terms built from negative viscosity and stabilizing hyperviscosity with constant parameters.
We study its interplay with linear and non-smooth quadratic bottom drag through the existence of coherent flows. Our results highlight that backscatter can have undesired amplification and selection effects, generating obstacles to energy distribution. We find that decreasing linear bottom drag destabilizes the trivial flow and generates nonlinear flows that can be associated with geostrophic equilibria (GE) and inertia-gravity waves (IGWs). The IGWs are periodic travelling waves, while the GE are stationary and can be studied by a plane wave reduction. We show that for isotropic backscatter both bifurcate simultaneously and supercritically, while for anisotropic backscatter the primary bifurcation are GE. In all cases presence of non-smooth quadratic bottom drag implies unusual scaling laws. For the rigorous bifurcation analysis by Lyapunov-Schmidt reduction care has to be taken due to this lack of smoothness and since the hyperviscous terms yield a lack of spectral gap at large wave numbers. 
For purely smooth bottom drag, we identify a class of explicit such flows that behave linearly within the nonlinear equations: amplitudes can be steady and arbitrary, or grow exponentially and unboundedly. 
We illustrate the results by numerical computations and present extended branches in parameter space. 
\end{abstract}

\section{Introduction}

In the study of geophysical flows, spatial resolution is limited not only by missing observational data, but also due to lack of computing power for numerical simulations. For relevant and realistic simulations, this is compensated by so-called sub-grid parameterizations, which model the impact of scales below resolution on the resolved larger scale flow. 
Such parameterizations ideally also account for numerical discretization effects such as over-dissipation, and yet admit stable simulations of ocean and climate models. 
A class of practical solutions that has come to frequent use are kinetic energy backscatter schemes, cf.\ e.g.\ \cite{JH2014, ZuritaEtAl2015,JansenEtAl2019, juricke2020kinematic,Perezhogin20}; we refer to \cite{DJKO2019} for a detailed discussion and relations to other approaches. 
At the grid-scale, on the one hand horizontal hyperviscosity 
is intended to remove enstrophy, while dissipating energy at finite resolution, and on the other hand suitable negative horizontal viscosity injects energy into large scales. Simulations with backscatter have been found to provide energy `at the right place', matching the results more closely to observations and high resolution comparisons. However, this still often requires judicious choice of backscatter and numerical parameters in order to stabilize simulations, cf.\ \cite{juricke2020kinematic,Jetal20b,JDKO2019,KJC2018}. 

Motivated by this, in \cite{PRY22} we have considered several geophysical fluid models with simplified constant parameter kinetic energy backscatter and hyperviscosity.
The idealized consideration on the continuum level admits a direct analytical study of the influence of backscatter through its impact on bifurcations and explicit flows. Due to the somewhat surprising growth phenomena found in \cite{PRY22}, in this paper we include additional dissipation effects due to bottom drag for the rotating shallow water equation. 
We show that on the one hand, the dissipation by bottom drag can balance some growth induced by backscatter, but also creates coherent nonlinear flows. These can be associated with geostrophic equilibria and inertia-gravity waves in terms of the linear modes at bifurcation. Here geostrophic equilibrium refers to a steady solution which exhibits geostrophic balance, i.e., the Coriolis force equals the horizontal pressure gradient force, cf.~\cite{Zeitlin}. For inertia-gravity waves both rotation and gravity provide the restoring force, which yields a characteristic frequency relation \cite{pedlosky1987geophysical}. On the other hand, we find that moderate bottom drag does not prevent the occurrence of unboundedly growing flows that were found in \cite{PRY22}. This indicates some robustness of undesired concentration of energy due to backscatter, which is in contrast to the targeted energy redistribution. 
These results also are consistent with the observed numerical blow-up and physically unrealistic flows in some simulations with subgrid-scale models containing backscatter, including non-parametric, data-driven subgrid-scale models, cf.\ \cite{Guan2022} and the literature therein.

\medskip
The rotating shallow water equations in an f-plane approximation and augmented with backscatter $\B$ as well as bottom drag $\F$ take the form
\begin{subequations}\label{e:sw}
\begin{align}
\frac{\partial\vv}{\partial t} + (\vv\cdot\nabla)\vv 
\ & =\ -f\vv^{\perp} - g\nabla\eta - \B\vv
- \F(\vv,\eta), \label{e:swa}
\\
\frac{\partial\eta}{\partial t} + (\vv\cdot\nabla)\eta 
\ & =\ -(H_0+\eta)\nabla\cdot\vv, \label{e:swb}
\end{align}
\end{subequations}
where $\vv=(u,v)^\intercal=\vv(t,\x)\in\R^2$ is the velocity field on the space $\x = (x,y)^\intercal \in\R^2$ at time $t\geq 0$, $\eta=\eta(t,\x)\in\R$ is the deviation of the fluid layer from the mean fluid depth $H_0>0$, so $\eta$ has zero mean and the thickness of the fluid is $H = H_0+\eta$ (with flat bottom), $f\in\R$ is the Coriolis parameter, with $f=0$ the non-rotational case, $g>0$ is the gravity acceleration. 

The backscatter (and hyperviscosity) operator $\B$ is given by
\[
\B = 
\begin{pmatrix} d_1\Delta^2+b_1\Delta & 0 \\ 0 & d_2\Delta^2+b_2\Delta \end{pmatrix},
\]
where $d_1, d_2>0$ provide hyperviscosity that stabilizes small scales by dissipating energy, and $b_1, b_2>0$ create negative viscosity in order to `backscatter', i.e.\ suitably return kinetic energy into large scales~\cite{JH2014, ZuritaEtAl2015,JansenEtAl2019, juricke2020kinematic,Perezhogin20,DJKO2019}. With slight abuse of terminology, we refer to all of $b_j,d_j,j=1,2$ as backscatter parameters, cf.~\cite{PRY22}. While the isotropic case $d=d_1=d_2$, $b=b_1=b_2$, appears physically natural, the effective coefficients, in particular $b_1,b_2$ may differ as discussed in \cite{PRY22}. We refer to the latter as semi-isotropic, but include also the fully anisotropic case $d_1\neq d_2$, $b_1\neq b_2$, since this causes no additional phenomena for the aspects we study. 
As mentioned, we consider the simplified situation with constant $b_j,d_j>0,j=1,2$.

As to the bottom drag, we take the combined form 
\begin{equation}\label{e:drag}
\F(\vv,\eta) = \frac {C + Q|\vv|} {H_0+\eta}\vv, \quad |\vv| = \sqrt{u^2 + v^2}, 
\end{equation}
where $C, Q\geq 0$ are constant parameters through which the dissipation enters linearly or quadratically with respect to the velocity, respectively. 
We refer to, e.g.\  \cite{TS2008,ZLWW2011,DJKO2019,GHSV2004,GSS1995,Sat1995,KW1991,KJC2018}
for these forms of bottom drag in various contexts; \cite{BM2004} uses 
$C(H_0+\eta)^{-2}$, which does not change the leading order impact. We note that \cite{KJC2018} uses quadratic drag combined with (energy-budget based) backscatter. Linear and quadratic bottom drag are often used alternatively, and here we follow \cite{AS2008}, which includes a combination. An important feature for $Q\neq 0$ is that the quadratic velocity drag term is once continuously differentiable, but not twice (in $\vv=0$).

The trivial steady flow of \eqref{e:sw} is $(\vv,\eta)= (0,0)$ and a major goal in this paper is to understand bifurcations in terms of the linear bottom drag parameter $\lb$. It turns out that this is a non-standard problem since the spectrum approaches the origin in the complex plane as the wave number tends to infinity, see \Cref{s:spec}. This lack of spectral gap also persists on periodic domains and means that center manifold reduction cannot be used, and care has to be taken in Lyapunov-Schmidt reduction. In addition, as for the usual shallow water equation mass is conserved since \eqref{e:swb} can be written as 
\[
\frac{\partial\eta}{\partial t} + \nabla\cdot\big((H_0+\eta)\vv\big) = 0.
\]
This implies, e.g.\ on $L^2$-based spaces or periodic domains, the mass conservation law 
\begin{equation}\label{e:conslaw}
\frac{\dif}{\dif t}\int_\Omega \eta\, \dif \x =0,
\end{equation}
which justifies the assumed zero mean of $\eta$. In general the presence of a conservation law can impact the dynamics at onset of instability, which, in a fluid-type context, has been studied in the presence of a reflection symmetry or Galilean invariance in \cite{MC2000,HSZ2011,SZ2013,MC2000b,Zim2017,SU17}. However, in \eqref{e:sw} the gradient terms and $\F$ or $f\neq 0$, break these structures. 

\medskip
In the bifurcation analysis of this paper we do not pursue center manifold reduction, but study bifurcations to standing and travelling waves by Lyapunov-Schmidt reduction. We prove that for $Q\neq 0$ geostrophic equilibria (GE) and inertia-gravity-type waves (IGW) bifurcate supercritically in terms of $\lb$. For isotropic backscatter both types of solutions bifurcate simultaneously, whereas in the anisotropic case the primary bifurcations are GE. In both cases the bifurcation has the same mildly non-smooth character as in \cite{SR2020}, which yields a linear amplitude scaling law rather than the usual square root form and requires non-standard computations for normal form coefficients. 
A reduction by a certain plane wave ansatz in the non-rotational case, $f=0$, yields a scalar Swift-Hohenberg-type equation with non-smooth quadratic nonlinearity. This admits a direct bifurcation study of GE and aspects of balanced dynamics by modulation equations. Their stability (`Eckhaus'-) region has an unusual scale and is smaller than that in case of cubic nonlinearity. However, insight into full stability and modulation properties of these as solutions to \eqref{e:sw} is limited and will be a subject of future research. We present some numerical computations which suggest that bifurcating GE can be stable also for $f\neq 0$. These computations give branches of GE and IGW that connect to $\lb=0$, i.e.\ purely quadratic bottom drag.

We also analyze the case $Q=0$ of smooth bottom drag terms that results in the usual (square root form) amplitude scaling law and prove the supercriticality of the bifurcating GE in the rotational case. In contrast to $Q\neq 0$, the amplitude scales with the backscatter parameters as $\sqrt{d/b^2}$, which grows rapidly for $d=\calO(b)$ as $b\to 0$, while the wave number scales as $\sqrt{b/(2d)}$. The amplitude of normalized states also scales inversely with the Coriolis parameter $f$, and branches become `vertical' in the non-rotational case, which admits arbitrary amplitude of the GE; these also possess sinusoidal shapes. We additionally find that selected linear modes appear as explicit solutions to the nonlinear equations and feature linear dynamical behaviour analogous to the results in \cite{PRY22}. In particular, there are exponentially and unboundedly growing explicit solutions, which are an undesired consequence of the backscatter scheme as mentioned above.

Indeed, none of these phenomena occur without backscatter,
i.e.\ in the usual viscous or inviscid case without forcing. Lastly, we remark that for both $Q=0$ and $Q\neq0$ a `Squire theorem' holds, meaning that the primary instability in 2D already occurs in 1D, and in the anisotropic case a coordinate direction is selected for the bifurcation waves. 

\medskip
This paper is organized as follows. In \Cref{s:sw} we discuss the options and limitations of plane wave reductions and the resulting modulation equations. In \Cref{s:spec} we study the spectrum of the trivial flow in the full system \eqref{e:sw} and the types of linear onset of instability. \Cref{s:bif} contains the bifurcation analysis of GE and \Cref{s:igw} the IGW. The explicit flows of linear character are discussed in \Cref{s:explicit}. Finally, we present some numerical computations in \Cref{s:num} and provide a short discussion in \Cref{s:discuss}. The appendix contains some technical aspects.

\section{Balanced plane waves and modulation}\label{s:sw}

First basic insight into the impact of backscatter and bottom drag can be gained by considering specific forms of solutions that in particular lead to geostrophic equilibria. For this we consider the specific plane wave-type ansatz in \eqref{e:sw} with wave vector $\k=(k_x,k_y)^\intercal\in\R^2$ of the form 
\begin{equation}\label{e:planeansatz}
\vv(t,\x)=\psi(t,\k\cdot\x)\k^\perp, \quad\eta(t,\x)=\phi(t,\k\cdot \x)
\end{equation}
for wave shapes $\psi, \phi$ and phase variable $\xi=\k\cdot\x$, where $\k^\perp=(-k_y,k_x)^\intercal$.
We remark that this ansatz is not intended for a broader study of waves and with some abuse of terminology, we sometimes refer to any resulting solution as a wave.
 In case of isotropic backscatter $d=d_1=d_2$, $b=b_1=b_2$, from \eqref{e:sw} we then obtain 
\begin{subequations}\label{e:planek}
\begin{align}
\partial_t \psi\, \k^\perp &=  (f\psi - g\partial_\xi\phi)\k  -\left(d k^4\partial_\xi^4\psi + b k^2\partial_\xi^2\psi + \frac{C+Qk|\psi|}{H_0+\phi}\psi\right)\k^\perp,\label{e:planea}\\
\partial_t\phi &= 0,\label{e:planeb}
\end{align}
\end{subequations}
where $k:=|\k|=|\k^\perp|$ is the wave number for the nonlinear plane wave-type solutions. Here, the only remaining nonlinearity (with respect to $\psi$) stems from the bottom drag for $Q\neq 0$; the transport nonlinearity of the fluid is absent for these plane waves. 
From \eqref{e:planeb} we see that $\phi$ needs to be independent of $t$, and for non-trivial $\k$ \eqref{e:planea} gives the two equations
\begin{subequations}\label{e:plane}
\begin{align}
f\psi &= g\partial_\xi\phi,\label{e:plane2a}\\
\partial_t \psi &= -d k^4\partial_\xi^4\psi - b k^2\partial_\xi^2\psi - \frac{C+Qk|\psi|}{H_0+\phi}\psi.\label{e:plane2b}
\end{align}
\end{subequations}
From equation \eqref{e:plane2a} it follows that also $\psi$ is independent of $t$ unless $f=0$. In case $f=0$, $\phi=\phi_0\in\R$ is constant, and \eqref{e:plane2b} has the form of a non-smooth quadratic Swift-Hohenberg-type equation with parameter $\phi_0$. Here the trivial steady states appear as $\psi=0$.
If in addition $Q=f=0$, then the remaining linear equation can be solved explicitly; in particular this gives a family of wave trains that solve the original nonlinear system \eqref{e:sw}, and contains a free amplitude parameter. Such solutions can also be found in case $Q=0$, $f\neq 0$ for anisotropic backscatter as discussed in \Cref{s:explicit}. 

\begin{remark}[Anisotropic backscatter]\label{r:anisoscalar}
In the anisotropic case, $b_1\neq b_2$  or $d_1\neq d_2$, 
we can still reduce to \eqref{e:planek}, when replacing $b,d$ by $b_2,d_2$ (or $b_1,d_1$) and choosing the wave vector $\k=(1,0)^\intercal$ (or $\k=(0,1)^\intercal$) in \eqref{e:planeansatz}.
\end{remark}

In order to study bifurcations for both $f\neq0$ and the limit $f=0$, we rescale the wave shape as $\phi = \tf\tilde\phi$, where $\tf=f/g$. As for $\eta$, here $\phi$  and thus $\tilde \phi$ is assumed to have zero mean without loss; any non-zero mean of $\eta$ corresponds to changing $H_0$.
In case $f\neq0$, substitution into \eqref{e:plane}, accounting for time-independence of $\phi$ as above and dropping the tilde gives 
\begin{subequations}\label{e:rescale}
\begin{align}
\psi&=\partial_\xi\phi,\label{e:scale}\\
0 &= -d k^4\partial_\xi^5\phi - b k^2\partial_\xi^3\phi - \frac{C+Qk|\partial_\xi\phi|}{H_0 + \tf\phi}\partial_\xi\phi.\label{e:planerescale}
\end{align}
\end{subequations}
In case $f=0$, i.e.\ $\tf=0$, zero mean of (the constant) $\phi$ and \eqref{e:plane} directly gives 
\begin{equation}\label{e:scalef0}
\partial_t\psi = -d k^4\partial_\xi^4\psi - b k^2\partial_\xi^2\psi - \frac{C+Qk|\psi|}{H_0}\psi. 
\end{equation}
The latter admits temporal dynamics and thus also contains information on stability with respect to plane wave perturbations of this form. 

\begin{remark}\label{r:geostrophic}
Physically, \eqref{e:plane2a} means that the Coriolis and gradient term are in `geostrophic' balance and therefore bifurcating solutions of this form, see \Cref{s:bif}, can be viewed as (nonlinear) geostrophic equilibria, cf.\ \cite{Zeitlin}.
\end{remark}

\begin{remark}\label{r:largef}
We briefly consider the limit of fast rotation, $|f|\to \infty$. With the scaling $\phi = \tf\tilde\phi$ equation 
\eqref{e:rescale} formally limits to the linear equation with 
$C=Q=0$. More generally, let us scale $\tilde\psi:=\tf^{\gamma}\psi$, $\tilde\phi := \tf^{\gamma-1}\phi$ with $\gamma\in [0,1]$ in \eqref{e:plane}. Dropping tildes, for $\gamma\in(0,1)$ the limiting equations are the same as for $\gamma=0$.
However, for $\gamma=1$ we obtain \eqref{e:rescale} with $\tf=1$ and $Q=0$ so that the nonlinear terms remain, which results in non-trivial bifurcations.
\end{remark}

\subsection{Modulation equations}\label{s:modulation}
The evolution equation \eqref{e:scalef0}, where $f=0$, naturally admits a modulation analysis near $\psi=0$ within the specific class of solutions. This in particular foreshadows aspects of bifurcations that will be studied in later sections. A more complete modulation study would require a fully coupled system. The parameter $\lb$ translates the spectrum of the linearization, which readily implies onset of instability at some $C=C_\crit$ at some wave number $k=k_\crit$ (for details see \Cref{s:spec}). 
The onset occurs at modes $e_{\pm 1}:= \rme^{\pm \rmi x}$, which implies \eqref{e:scalef0} takes the form
\begin{equation}\label{e:scalef0b}
\partial_t\psi = -\td(\partial_x^2+1)^2\psi + \alpha \psi - q |\psi|\psi,
\end{equation}
where $\td := k_\crit^4 d$, $q:=Qk_\crit/H_0$, $\alpha := (C_\crit-C)/H_0$;
the onset at $\alpha=0$ with wave number $1$ removes $b$. 
The linear part of \eqref{e:scalef0b} is (up to a prefactor) the same as that in the classical Swift-Hohenberg equation, e.g.\ \cite{CH1993}, which implies that the parabolic scaling for modulations with scale parameter $0<\eps\ll 1$ is appropriate. Thus we set $\alpha = \eps^2 \mu$ 
and introduce the spatial scale $X=\eps x$ and temporal scale $T=\eps^2 t$. The nonlinear part, for $q\neq 0$, does not involve derivatives and scales quadratically in $\psi$. This means that the amplitude of modulations of $e_{\pm 1}$ scales as $\eps^2$ (rather than $\eps$ for the standard cubic nonlinearity), which leads to the ansatz
\[
\psi(t,x) \approx \eps^2 A(T,X)e_1 + \cc
\]
Upon substitution into \eqref{e:scalef0b}, the lowest order in $\eps$ for a non-trivial contribution is $\eps^4$ (rather than $\eps^3$ for a cubic nonlinearity). 
The projection onto $e_1$ is akin to the real Ginzburg-Landau equation \cite{CH1993}, but with non-smooth quadratic nonlinearity, 
\begin{equation}\label{e:GLf0b}
\partial_T A = 4\td \partial_X^2 A +\mu A -\frac{16}{3\pi}q |A|A.
\end{equation}
Details of these computations are reflected in the rigorous bifurcation analysis in \Cref{s:bif} and therefore omitted here. Wave trains $A(T,X) = R\rme^{ \rmi K X}$ in \eqref{e:GLf0b} satisfy the nonlinear dispersion relation $\mu = 4\left(\frac{4q}{3\pi}|R|+\td K^2\right)$ or, for $q\neq0$, equivalently the bifurcation equation for the amplitude 
\begin{equation}\label{e:modamp}
|R| = \frac{3\pi}{16 q}\left(\mu - 4\td K^2\right).
\end{equation}
For $Q=q=0$, when \eqref{e:scalef0b} is linear, the bifurcating branches are `vertical' at $\mu=4\td K^2$. For $q>0$ (i.e.\ $Q>0$) wave trains
bifurcate supercritically, but with linear (rather than square root form for a cubic nonlinearity) amplitude scaling law, as in the non-smooth bifurcations studied in \cite{SR2020}. 
For the present case of $f=0$, we can also infer some stability information of the wave trains.  A slightly non-standard computation (see \cite{MiriamThesis} for details), but analogous to the standard approach for stability of sideband-modes, also referred to as Eckhaus stability \cite{CH1993}, gives the stability condition 
\[
\mu > 5\cdot (4\td)K^2.
\]
Here the factor $5$ reflects the nature of the term $|A|A$ and is in particular bigger than the factor $3$ for the smooth cubic nonlinearity $|A|^2A$. Thus the stability (`Eckhaus'-) region for this non-smooth quadratic situation is smaller than that in the smooth case. However, \eqref{e:scalef0b} is only valid for the 
specific plane wave forms so that this consideration does not entail a full stability analysis of the wave trains as solutions to \eqref{e:sw}. 

\medskip
In case $f,Q\neq 0$, due to  \eqref{e:planeb}, the plane wave reduction \eqref{e:plane} does not yield an evolution equation so that information on stability requires consideration of the full system \eqref{e:sw}. Purely spatial modulation in \eqref{e:planerescale} can be used to study the existence of steady states, and amounts to a formal version of the bifurcation analysis that will be presented in \Cref{s:bif}. 

\medskip
We recall that \eqref{e:sw} possesses the conservation law \eqref{e:conslaw}. In general the presence of a conservation law can have a significant impact on modulation equations: in conjunction with a reflection symmetry we refer to \cite{MC2000,HSZ2011,SZ2013} and with Galilean invariance to \cite{MC2000b,Zim2017,SU17}. 
However, \eqref{e:sw} possesses neither a global reflection symmetry due to the gradient terms, nor Galilean invariance for $f\neq 0$ or in presence of $\F$, which is required for onset of instability.  For $f=0$ and $Q\neq 0$, a reflection symmetry is present within the reduction to \eqref{e:scalef0}, but since $\phi$ is required to be constant in this case, the conservation law is fully replaced by the parameter $H_0$. 
The Coriolis term for $f\neq 0$ breaks the Galilean invariance already on the linear level in \eqref{e:sw}. We do not enter into details as our main interest lies in the bifurcations.

\section{Spectrum and spectral stability}\label{s:spec}

We return to \eqref{e:sw} and study in detail the spectrum in the trivial constant flows. Since \eqref{e:sw} is a coupled parabolic-hyperbolic system it is not clear that standard center manifold reduction can be employed, even on bounded domains. 
Linearizing \eqref{e:sw} in the trivial steady flow $(\vv,\eta) = (0,0)$ gives the linear operator
\begin{equation}\label{e:linop}
\calL =  
\begin{pmatrix}
- d_1\Delta^2 - b_1\Delta - C/H_0 & f & -g\partial_x  \\
-f & - d_2\Delta^2 - b_2\Delta - C/H_0 & -g\partial_y \\
-H_0\partial_x & -H_0\partial_y & 0
\end{pmatrix}.
\end{equation}
The spectrum of $\calL$ (on e.g.\ $(\Lspace^2(\R^2))^3$) consists of the roots $\lambda$ of the dispersion relation
\begin{equation}\label{e:disp}
d(\lambda,\k) := \det(\lambda \Id - \widehat\calL) = 0,
\end{equation}
for wave vectors $\k = (\kx,\ky)^\intercal\in\R^2$ and with the Fourier transform $\widehat\calL$  of $\calL$. 
The sign of $\Re(\lambda)$ in the solutions to \eqref{e:disp} determines the spectral stability of $(\vv,\eta)=(0,0)$ with respect to perturbations with wave vector $\k$
and we write continuous selections of such solutions in the functional form $\lambda=\lambda(\k)$.
The dispersion relation can be written as
\begin{equation}\label{e:poly}
d(\lambda, \k) = \lambda^3 + a_1 \lambda^2 + a_2 \lambda + a_3 = 0,
\end{equation}
with wave vector dependent coefficients 
\begin{align*}
a_1 \ &:=\ (d_1 + d_2) |\k|^4 - (b_1 + b_2) |\k|^2 + 2C/H_0, \\
a_2 \ &:=\ (d_1 |\k|^4 - b_1 |\k|^2 + C/H_0) (d_2 |\k|^4 - b_2 |\k|^2 + C/H_0) + g H_0 |\k|^2 + f^2, \\
a_3 \ &:=\  g H_0 |\k|^2 \left((d_1|\k|^2-b_1)k_y^2+(d_2|\k|^2-b_2)k_x^2 + C/H_0 \right).
\end{align*}
We highlight two aspects concerning the structure of the spectrum. On the one hand, $\lambda=0$ is a solution at $\k=0$ for any choice of parameters, whose formal eigenfunction is $(\vv,\eta) =(0,1)$. This perturbs the total fluid mass and relates to the mass conservation law, as well as the family of trivial steady flows $(\vv,\eta)\equiv(0,h)$, $h\in\R$. 
Hence, such perturbations can be ignored by imposing that $H_0$ is the mean fluid depth, which means $\eta$ has zero mean. 

On the other hand, the presence of hyperviscous terms for $d_1d_2\neq 0$ implies that there is a continuous solution $\lambda_\infty(\k)$ with $\Re(\lambda_\infty(\k))\to 0$ as $|\k|\to \infty$. See \cref{f:specinf}. Indeed, for $\k=r^{-1}\bar\k$ with $|\bar\k|=1$ we have 
\begin{equation}\label{e:speczero}
\lim_{r\to 0} r^8d(\lambda, r^{-1}k_x,r^{-1}k_y) = d_1d_2\lambda,
\end{equation}
which gives the solution $\lambda=0$ at $r=0$, i.e.\ $|\k|=\infty$. The implicit function theorem yields $\lambda_\infty(\k)$ as claimed for sufficiently large $|\k|$. Hence, there is no spectral gap for \eqref{e:sw} posed on, e.g.\ rectangles with periodic boundaries, which is the relevant case when seeking wave trains.  As a consequence, the standard approach to center manifolds cannot be applied in order to reduce the bifurcation problem from a trivial flow to a finite dimensional ODE. 
We therefore restrict our attention to the existence of certain bifurcating solutions, and omit an analysis of well-posedness as well as stability of the bifurcating solutions. 
In practice backscatter is applied to a numerically discretized situation, which effectively cuts the available wave vectors at some large $|\k|$. Thus, a spectral gap is retained, but the properties of the present continuum equations are nevertheless relevant, in particular for large-scale settings. 

\medskip
In contrast, in the viscous case, $d_1=d_2=0$, $b_1=b_2<0$, we have 
\[
\lim_{r\to 0}r^4 d(\lambda,r^{-1}\kx,r^{-1}\ky) = b^2\lambda - gH_0 b \mbox{ , and }
\lim_{r\to 0}r^6 d(r^{-2}\lambda,r^{-1}\kx,r^{-1}\ky) = \lambda(\lambda-b)^2.
\]
This implies that the real part of the spectrum has one parabolic branch with real part that is asymptotically quadratic in $|\k|$ and one which approaches the real negative value $g H_0/b$ as $|\k|\to\infty$. Similarly, in the inviscid case, $d_1=d_2=b_1=b_2=0$,  scaling gives
\begin{align*}
\lim_{r\to 0} r^2 \Re \left(d(\lambda_\rmr+\rmi r^{-1}\lambda_\rmi, r^{-1}\kx,r^{-1}\ky)\right)\ &=\ -3\lambda_\rmr \lambda_\rmi^2 - \frac{2\lb}{H_0}\lambda_\rmi^2 + gH_0 \lambda_\rmr + g\lb,\\
\lim_{r\to 0} r^3 \Im \left(d(\lambda_\rmr+\rmi r^{-1}\lambda_\rmi, r^{-1}\kx,r^{-1}\ky)\right)\ &=\  \lambda_\rmi(\lambda_\rmi^2 - gH_0),
\end{align*}
where $\lambda_\rmr:=\Re(\lambda)$ and $\lambda_\rmi:=\Im(\lambda)$. Hence, the spectrum for large $|\k|$ approaches the negative real value $-\lb/H_0$, and a complex conjugate pair $-\lb/(2H_0) \pm \rmi\sqrt{gH_0}$, so that real parts for large $|\k|$ are bounded uniformly away from zero.

\subsection{Onset of instabilities}\label{s:Turing}

Given some $\k=(k_x,k_y)^\intercal\in\R^2$, by the Routh-Hurwitz criterion a solution $\lambda$ to \eqref{e:poly} satisfies $\Re(\lambda)<0$ if and only if 
\begin{equation}\label{e:RH}
a_1>0,\; a_3>0,\; a_1a_2-a_3>0.
\end{equation}
More specifically, \eqref{e:poly} possesses a zero root $\lambda=0$ if and only if $a_3=0$, or has purely imaginary roots $\lambda=\pm\rmi\omega\in\rmi\R\setminus\{0\}$ if and only if $a_1a_2-a_3=0$ and $a_2>0$. 
For instance, without backscatter $b_j,d_j=0,j=1,2$, the conditions in \eqref{e:RH} are satisfied for all wave vectors $\k\neq0$ and any  bottom drag coefficient $\lb>0$, which implies that $(\vv,\eta)=(0,0)$ is spectrally stable in this case (with neutral mode at $\k=0$).

\begin{figure}[t!]
\centering
\subfigure[]{\includegraphics[trim = 4cm 8.5cm 5cm 9cm, clip, height=4cm]{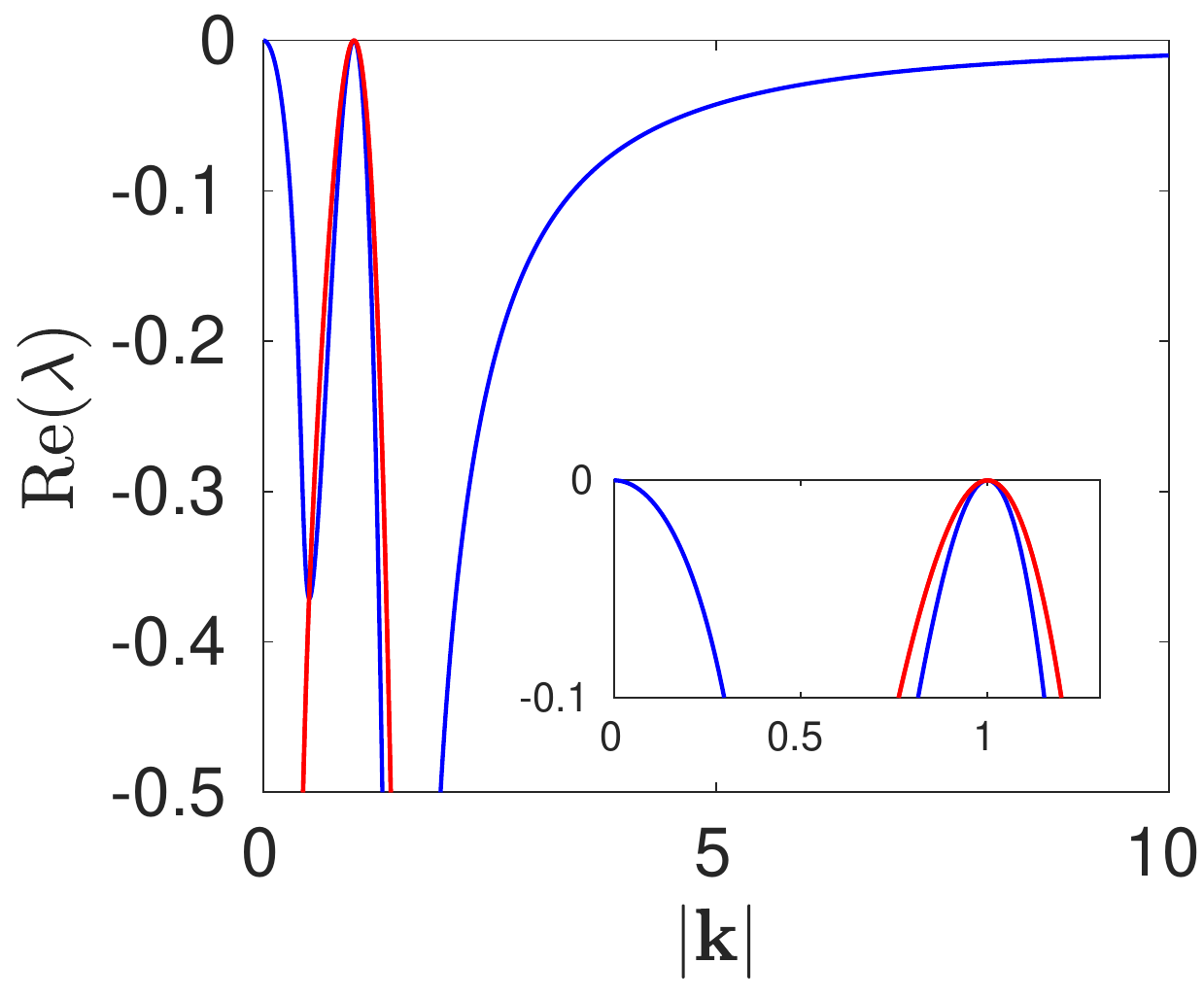}}
\hfil
\subfigure[]{\includegraphics[trim = 4cm 8.5cm 5cm 9cm, clip, height=4cm]{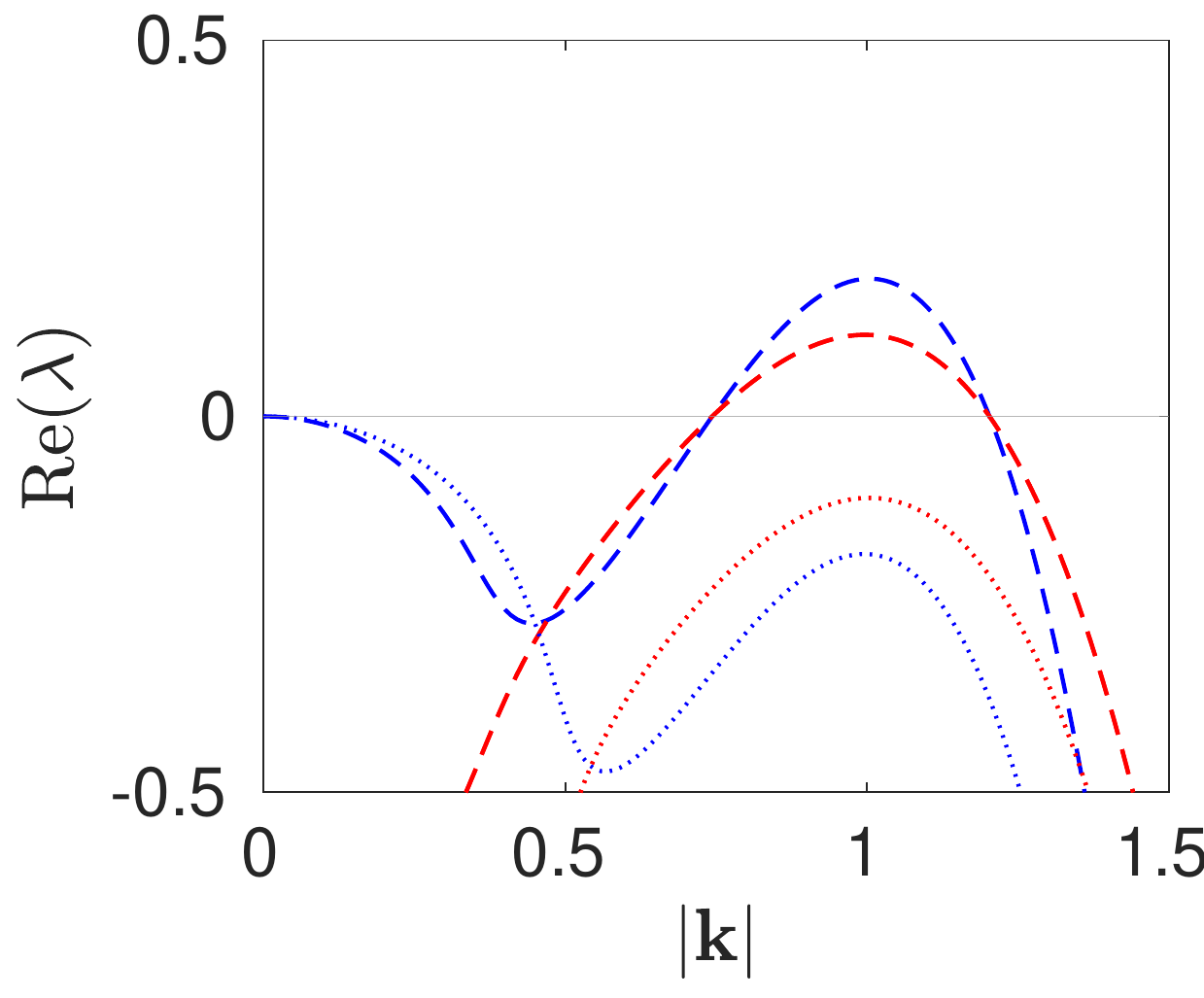}}
\caption{Samples of spectrum in the isotropic case $d_1=d_2 =1$, $b_1=b_2 =2$. Other fixed parameters: $f=0.3$, $g=9.8$, $H_0=0.1$ so that $\lb_\crit=0.1$ and $k_\crit=1$. (a) The real spectrum (blue) and the complex spectrum (red) of $\calL$ (magnification in inset) are simultaneously critical for $\lb=\lb_\crit$ at $\{\k\in\R^2:|\k|= k_\crit =1\}$. The real spectrum is zero at $\k=0$ for any choice of parameters and approaches zero as $|\k|\to \infty$. (b) Plotted are a stable case at $\lb = 0.12$ (dotted) and an unstable case at $\lb = 0.08$ (dashed); notably in the latter case real and complex spectra have zero real parts for the same wave length.}
\label{f:specinf}
\end{figure}

\medskip
\paragraph{Isotropic backscatter} In the isotropic case $d=d_1=d_2$, $b=b_1=b_2$, the dispersion relation depends on $\k$ only through $K:=|\k|^2$. The spectrum is therefore rotationally symmetric with respect to $\k\in\R^2$ in the wave vector plane. Moreover, $a_1$, $a_3$, and thus $a_1a_2-a_3$, possess the common factor $F(K;C) := d K^2 - b K + C/H_0$. In addition, $a_1$ and $a_1a_2-a_3$ have the sign of $F$, and $a_3$ that of $KF$. 
Let $\lambda_j(K;C)$, $j=1,2,3$ denote the roots of the dispersion relation in some ordering. It follows that $\lambda_j(K;C)=0$ for some $j$ occurs at $K=0$ and at roots of $F$. The latter emerge when decreasing $C$ below the threshold 
\begin{equation}\label{e:isocritC}
\lb_\crit := \frac{b^2H_0}{4d},
\end{equation}
where the global minimum of $F$ is a double root in $K$ and lies at wave vectors with
\begin{equation}\label{e:isocritk}
|\k|=k_\crit := \sqrt{\frac{b}{2d}}.
\end{equation}
At $C=\lb_\crit$ we choose indices so $\lambda_1(k_\crit;\lb_\crit)=0$, $\lambda_2(k_\crit;\lb_\crit)=-\lambda_3(k_\crit;\lb_\crit)=-\rmi \omega_\crit$, where 
\begin{equation}\label{e:omc}
\omega_\crit := \sqrt{a_2}|_{\lb=\lb_\crit, |\k|=k_\crit}=\sqrt{g H_0 k_\crit^2+f^2},
\end{equation} 
and $\Re(\lambda_j(K;\lb))<0$ for $|\k|\in\R\setminus\{0,k_\crit\}$, $j=1,2,3$. Hence, in the wave vector plane, the critical spectrum away from the origin forms a circle centered at the origin with radius $k_\crit$, cf.\ \cref{f:specinf}. And, as $C$ decreases below $\lb_\crit$, the zero state becomes unstable \emph{simultaneously} via a stationary (akin to a Turing) and an oscillatory (akin to a Turing-Hopf) instability with finite wave number;  both in presence of an additional neutrally stable mode with zero wave number. More precisely, 
for $0<\lb<\lb_\crit$, there is a positive interval $I_K$ with $a_3<0$ for $K\in I_K$. Hence, in this case \eqref{e:poly} has a positive root in the annulus $|\k|^2\in I_K$ in the wave vector plane.

Regarding the limit of small backscatter $d,b\to 0$, we note that the scalings of $\lb_\crit$ and $k_\crit$ differ so that fixed $k_\crit$ requires $\lb_\crit\to 0$, while fixed $\lb_\crit$ requires $k_\crit\to\infty$.

Any neutral mode corresponds to a mode of the inviscid case, with the usual geophysical terminology, cf.\ e.g.\  \cite{pedlosky1987geophysical}: Due to the frequency relation \eqref{e:omc}, we refer to the corresponding waves as inertia-gravity waves (IGWs) rather than Poincar\'e waves. 
The steady modes are in geostrophic balance and we therefore interpret the instability in the present isotropic case as a backscatter and bottom drag induced (simultaneous) instability of selected geostrophic equilibria and inertia-gravity waves.

\medskip
\paragraph{Anisotropic backscatter} In the anisotropic case $b_1\neq b_2$ and/or $d_1\neq d_2$ 
with $b_j,d_j>0,j=1,2$, the structure of the spectrum is more complicated. As a global constraint, we next show that the real spectrum is always more unstable than any non-real spectrum. 
In particular, the primary instability with respect to $C$ is purely stationary. 

Since the dispersion relation explicitly depends on $\kx$ and $\ky$, the spectrum is anisotropic in the wave vector plane. But, analogous to the isotropic case, $a_3$ has a factor
\[
F(\kx,\ky;\lb) := d_1|\k|^2 k_y^2 + d_2|\k|^2 k_x^2 - b_1 k_y^2 - b_2 k_x^2 + \lb/H_0,
\]
and $a_3$ has the sign of $KF$. We next show that $F(\k;\lb)=0$ is the minimal value of $F(\cdot;\lb)$ at $\lb=\lb_\crit$ and wave vectors $\k=\pm\k_\crit$ given by
\begin{subequations}\label{e:anisocritC}
\begin{align}
\lb_\crit &= \frac{b_2^2H_0}{4d_2}, \;  \k_\crit = \left( k_\crit,0\right)^\intercal, \; k_\crit=\sqrt{\frac{b_2}{2d_2}}, \quad \text{for}\ \frac{b_1^2}{d_1} \leq \frac{b_2^2}{d_2},\label{e:critx}\\
\lb_\crit &= \frac{b_1^2H_0}{4d_1}, \;  \k_\crit = \left(0, k_\crit\right)^\intercal,\; k_\crit=\sqrt{\frac{b_1}{2d_1}}, \quad \text{for}\ \frac{b_1^2}{d_1} \geq \frac{b_2^2}{d_2}.\label{e:crity}
\end{align}
\end{subequations}
The critical points of $F(\cdot;\lb)$ lie at $\k=0$, $\k=\pm\k_\crit$, and (if real) 
\[
\k=\k_{m,n}:=\left(m\sqrt{\frac{b_1(d_1+d_2)-2b_2d_1}{(d_1-d_2)^2}}, n\sqrt{\frac{b_2(d_1+d_2)-2b_1d_2}{(d_1-d_2)^2}}\right), \; m,n=\pm 1.
\]
(Note that $\k_{m,n},\,m,n=\pm 1$ are not critical points for $d_1=d_2$.) 
The (real) critical points $\k_{m,n}$ are not minima since the Hessian matrices of $F$ evaluated at $\k_{m,n}$ are indefinite. Indeed, the determinant of the Hessian is always $-16 \frac{(b_1 (d_1 + d_2)- 2 b_2 d_1)(b_2 (d_1 + d_2)-2 b_1 d_2)}{(d_1 - d_2)^2}$, which is negative for real $\k_{m,n}$. Since $F(\cdot;\lb)$ is coercive and sign symmetric, the global minima are as claimed.

In particular, it is straightforward to switch between the two cases upon moving parameters through the isotropic case. The respective scaling of $\lb_\crit$ and $k_\crit$ as $d_j,b_j\to 0$ for $j=1,2$ is the same as in \eqref{e:isocritC}, \eqref{e:isocritk}.

Due to $\partial_\lambda d(0,\k_\crit)=  gH_0 k_\crit^2+f^2\neq0$ at $\lb=\lb_\crit$, by the implicit function theorem there is a unique continuation $\lambda_1 = \lambda_1(\k;\lb)$ of roots of the dispersion relation \eqref{e:poly} for $\k$ near $\pm\k_\crit$. Moreover, $\lambda_1<0$ for $\k$ near (but not at) $\pm\k_\crit$.
The critical wave vectors $\k_\crit$ determined in \eqref{e:critx} and \eqref{e:crity} differ if $b_1^2/d_1\neq b_2^2/d_2$ so that $\k_\crit$ lies either on the $k_x$- or the $k_y$-axis in the wave vector plane. This is similar to reaction diffusion systems with uni-directional advection \cite{YangThesis}. 
In case $b_1^2/d_1= b_2^2/d_2$ the critical wave vectors lie on both axes so there are four critical wave vectors.  
By continuity, for $\lb<\lb_\crit$, but nearby, we have
$\lambda_1>0$ for wave vectors $\k$ in a disc shaped set near any of the two or four critical $\k_\crit$.
Specifically, for $b_1^2/d_1\neq b_2^2/d_2$, this corresponds to the red regions in \cref{f:growth} of \Cref{s:explicit}; for $b_1^2/d_1= b_2^2/d_2$ there are four such sets in the wave vector plane.

\begin{remark}
In any anisotropic case, the onset of instability upon decreasing $\lb$ occurs first through steady modes with wave vectors on the $\kx$- and/or $\ky$-axis. 
This is analogous to a Squire theorem, in the sense that the onset of instability in 2D coincides with the onset in 1D, which trivially holds in the isotropic case. 
For the anisotropic case, in \Cref{s:critical} we show $a_1, a_1a_2-a_3>0$ for $\lb\geq\lb_\crit$ and $\k\in\R^2$, so that the claim follows by the Routh-Hurwitz criterion. 
\end{remark}

\medskip
We summarize the results on the critical spectrum. In terms of decreasing the bottom drag parameter $\lb$, the spectrum of the linear operator $\calL$ from \eqref{e:linop} changes stability at $\lb=\lb_\crit$. The resulting operator $\calL_\crit := \calL|_{\lb=\lb_\crit}$ can have the following types of marginally stable spectral structure:
\begin{itemize}
\item[3D] A three-dimensional kernel for anisotropic backscatter satisfying $b_1^2/d_1\neq b_2^2/d_2$.\\ A sample is shown in \cref{f:spec3D}. 
\item[5D] A five-dimensional kernel for anisotropic backscatter satisfying $b_1^2/d_1= b_2^2/d_2$.\\ A sample is shown in \cref{f:spec5D}.
\item[$\infty$D] For isotropic backscatter, an infinite-dimensional kernel and simultaneously an infinite-dimensional center space, both parameterized by a circle of wave vectors. Samples are shown in \cref{f:specinf}.
\end{itemize}

\begin{remark}\label{r:smallback}
We briefly consider the regime of small backscatter parameters. Setting $(d_j, b_j) =(\eps^r \hat d_j,\eps\hat b_j)$, $0<\eps\ll1$, yields non-trivial and bounded critical wavenumbers only in case $r=1$. One may interpret $r< 1$ and $\eps\to 0$ as a regime where the backscatter cannot return the energy that is dissipated by the hyperviscosity into the same scales in terms of wavenumbers. We recall that the zeros of the dispersion relation are those of $a_3$ after \eqref{e:poly}, which is linear in $b_j, d_j$ and $C$. Hence, if all these scale in the same way, then the zeros of $a_3$, and thus the range of scales in which energy is backscattered, remain $\eps$-independent. \end{remark}

\begin{figure}[t!]
\centering
\subfigure[3D case]{\includegraphics[trim=0.5cm 9cm 1.5cm 9cm, clip, height=3.7cm]{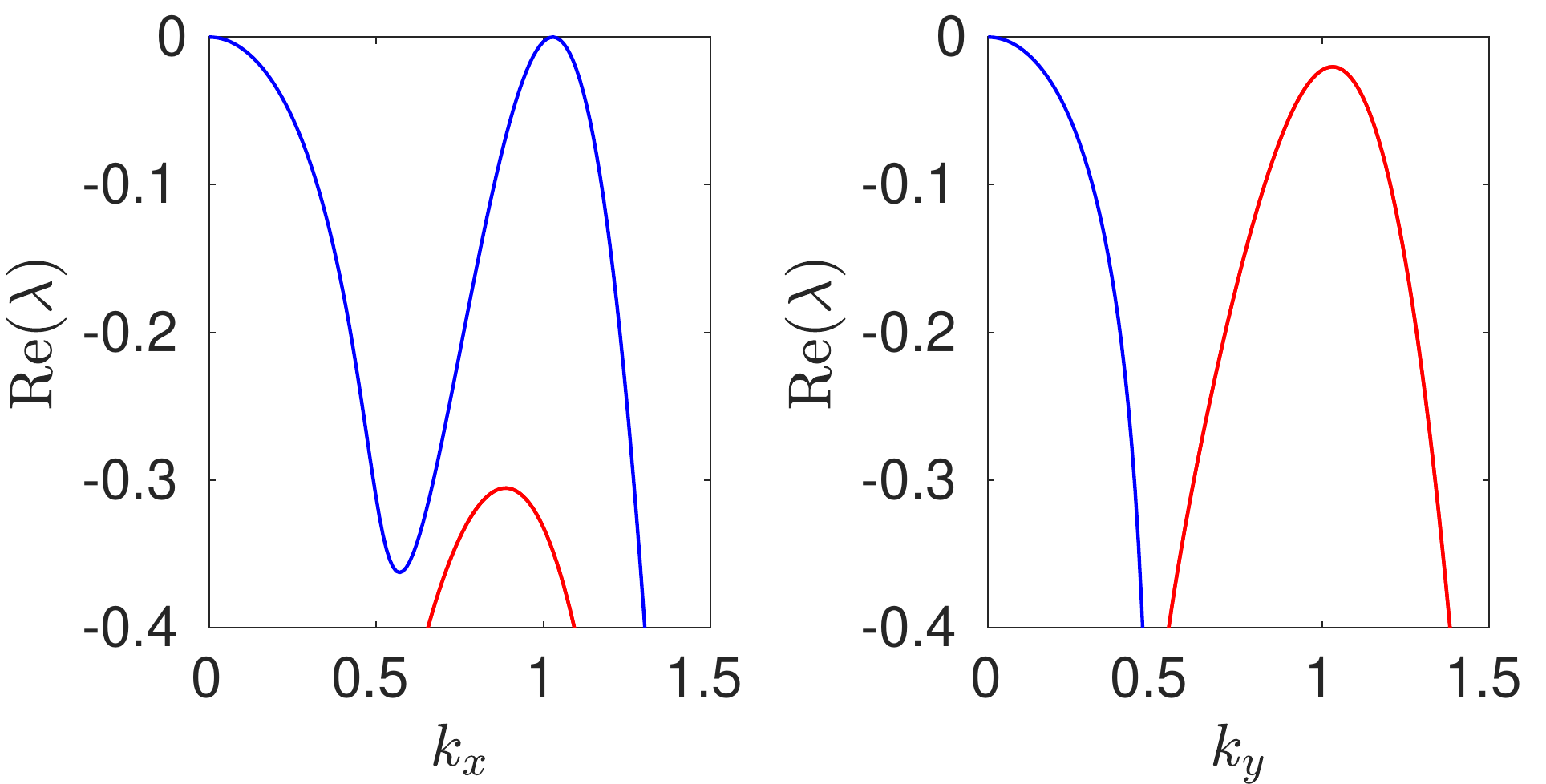}\label{f:spec3D}}
\hfill
\subfigure[5D case]{\includegraphics[trim=0.5cm 9cm 1.5cm 9cm, clip, height=3.7cm]{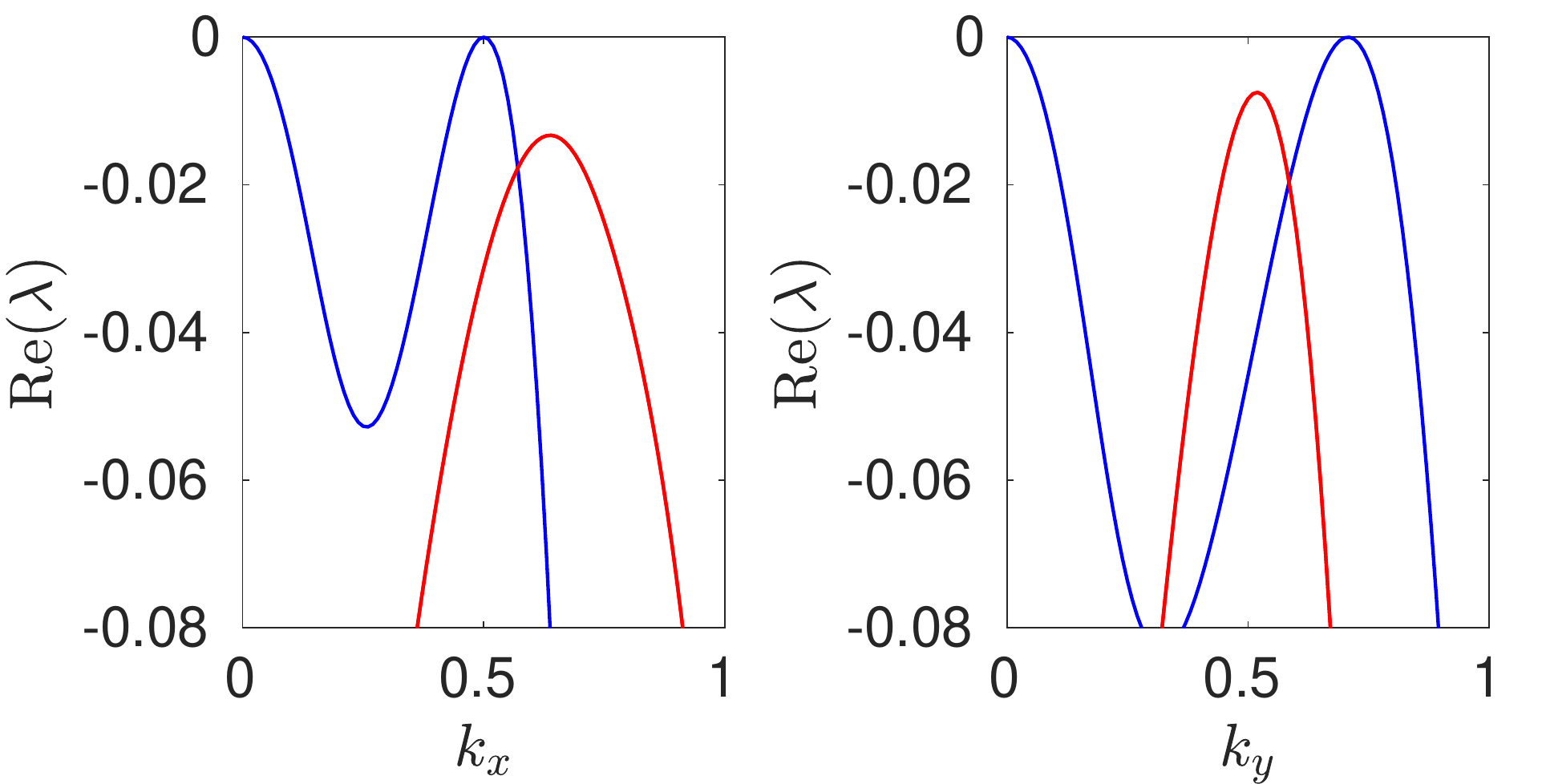}\label{f:spec5D}}
\caption{Sample of critical spectra in two anisotropic cases. The real spectrum (blue) of $\calL$ is critical for $\lb=\lb_\crit$ at (a) $\k = (\pm k_\crit,0)^\intercal$; (b) $\k=(\pm\kcx,0)^\intercal$ and $\k=(0,\pm\kcy)^\intercal$, with neutral mode at $\k=0$ for any choice of parameters, while the complex spectra (red) are strictly stable. Common parameters: $f=0.3$, $g=9.8$, $H_0=0.1$. Other parameters: (a) $d_1=1$, $d_2=1.04$, $b_1=1.5$, $b_2=2.2$ so that $\lb_\crit \approx0.116$ and $k_c\approx1.03$; (b) $d_1=1$, $d_2=4$, $b_1=1$, $b_2=2$ so that $\lb_\crit=0.025$, $\kcx=0.5$ and $\kcy\approx 0.71$.}
\label{f:specaniso}
\end{figure}

Having identified critical configurations, we discuss the form of the linear operator and its critical eigenvectors at onset of linear instability through the near zero bifurcation parameter 
\[
\alpha = (\lb_\crit-\lb) /H_0. 
\]
The linear operator $\calL$ from \eqref{e:linop} is then of the form
\[
\calL\ =\ \calL_\crit +
 \alpha\,\mathrm{diag}(1,1,0).
\]
On the one hand, at $\k=\k_\crit$ we expand $\lambda=\lambda(\alpha)$ solving \eqref{e:poly} with $\lambda(0)=0$ as
\[
\lambda(\alpha) =  \partial_\alpha\lambda(0)\alpha + \calO(\alpha^2),
\]
with $\partial_\alpha\lambda(0) = [(\partial_C a_3)H_0/a_2]_{\lb=\lb_\crit, \k=\k_\crit} = g|\k_\crit|^2 H_0/\omega_\crit^2 > 0$. 
Hence, the real part of the critical spectrum at $\k=\k_\crit$ moves linearly in $\alpha$. 
On the other hand, the spectrum near $\k =0$ 
remains stable, in fact for all $\lb>0$. To show this, we consider the continuation $\lambda_1(\k)\in\R$ with $\lambda_1(0)=0$. In polar coordinates, $\k = R(\cos\theta,\sin\theta)^\intercal$, we compute 
$\partial_R\lambda_1(0) = 0$ and $\partial_R^2\lambda_1(0)=-\frac{2 \lb g H_0^2}{\lb^2 + f^2 H_0^2}$, the latter is negative for $\lb>0$ so that $\lambda_1(\k)$ is negative near $\k=0$. In contrast, for $\lb=0$ the continuation $\lambda_1(\k)$ is positive near $\k=0$ since  its leading order expansion near $\k=0$ is $\frac{b_2 g H_0}{f^2}R^4$, cf.\ \cite{PRY22}.

\medskip
\paragraph{3D kernel}
As shown above, in this case the zero eigenmode of $\calL_\crit$ has wave vector either on the $k_x$-axis or on the $k_y$-axis, and without loss of generality we only discuss the former. Hence, in order to identify the primary bifurcation, it suffices to consider \eqref{e:sw} in 1D, i.e.\ $\x = (x,0)^\intercal$, $x\in[0,2\pi/k]$ with periodic boundary conditions, where $k$ is the wave number for the expected bifurcating nonlinear wave trains. As noted in \cref{r:anisoscalar}, somewhat surprisingly, the bifurcation problem can be reduced a priori to a scalar equation, and we will pursue this later based on the isotropic case. As to the linear structure of the full three-component system, we rescale the domain to $[0,2\pi]$ and restrict $\calL$ from \eqref{e:linop} to 1D. 
The kernel eigenvectors of $\calL$ at $(\lb,k) = (\lb_\crit,k_\crit)$ are 
$\tee_j:=\tE_j \rme^{\rmi j x}$, $\tE_j\in\C^3$ and $\tE_{-j} = \overline{\tE}_j$, for $j=0,\pm1$, and 
\begin{equation}\label{e:rossbyeigenvector}
\tE_j = 
\begin{pmatrix}
0\\ 1\\ -\rmi j f/(g k_\crit)
\end{pmatrix},
\; j=\pm 1, \quad 
\tE_0 = 
\begin{pmatrix}
0\\ 0\\ 1
\end{pmatrix}.
\end{equation}
Corresponding to mass conservation, $\tee_0$ is in the kernel of $\calL$ for all $\lb$. The first two components of $\tE_{\pm 1}$ are orthogonal to $\k=(k_x,0)^\intercal$, which is in accordance with the possibility to reduce to the scalar equation in the form of \eqref{e:planeansatz}, cf.\ \cref{r:anisoscalar}. We note that the eigenvectors depend on the bottom drag parameters only through $k_\crit$ from \eqref{e:isocritk}, i.e.\ $\lb_\crit = d_2H_0k_\crit^4$. Consistent with \cref{r:geostrophic}, the kernel eigenvectors $\tee_{\pm1}$ correspond to geostrophically balanced modes, e.g.\ \cite{pedlosky1987geophysical}.

\medskip
\paragraph{5D kernel} In this case the critical modes  in $k_x$- and $k_y$-direction of the 3D cases occur simultaneously. Hence, the eigenvectors identified in the 3D case combined provide the linear structure at criticality. The reduction to the respective scalar equations admits a partial unfolding of the bifurcation. In a more complete unfolding we expect square-type patterns of geostrophic equilibria in orthogonal directions, analogous to, e.g.\ \cite{YangThesis}. 
However, it is not clear that this can be made rigorous in the present context due to the lack of a spectral gap following from \eqref{e:speczero}. In particular, the required Fredholm properties are not immediately clear. 

\medskip
\paragraph{Infinite-dimensional kernel} Here the backscatter is isotropic and the critical eigenmodes of $\calL_\crit$ lie at the origin and on a circle with radius $k_\crit$ in the wave vector plane. This circle is doubly covered by steady and oscillatory modes. For single steady modes, in contrast to the 3D and 5D cases,  any wave vector direction can be selected in order to reduce to a 1D plane wave problem. This is always \eqref{e:planerescale} and admits an analysis of the stationary bifurcation problem independent of the spectral gap problem. By reducing to lattices, the mentioned technical issues aside, we expect various steady patterns to bifurcate, in particular squares and hexagons.

Towards bifurcations of oscillatory solutions, we consider 
a comoving frame $\x - \mathbf{c} t$ with suitable constant $\mathbf{c}$. 
This change of variables creates the additional term $-\mathbf c\cdot\nabla (\vv,\eta)^\intercal$ on the left-hand side of \eqref{e:sw} so that the dispersion relation \eqref{e:poly} turns into 
\begin{equation}\label{e:dispcomov}
d_\mathbf{c}(\lambda, \k) = d(\lambda - \rmi\mathbf{c}\cdot\k,\k). 
\end{equation}
Hence, exactly the purely imaginary spectrum at $C=\lb_\crit$ from \eqref{e:isocritC} is shifted to the origin. Indeed, if $d(\rmi \omega,\k_\omega)=0$ for $\omega\in\R$ and some $\k_\omega\neq 0$, then $d_\mathbf{c}(0, \k_\omega)=0$ for $\mathbf{c}=-\omega/|\k_\omega|^2\;  \k_\omega$. 
This applies for $\omega =\pm\omega_\crit$ 
 from \eqref{e:omc} with any $\k_\omega=\k_\crit=(\kcx,\kcy)^\intercal$ 
satisfying \eqref{e:isocritk}.  
Concerning bifurcations, we again consider the simplest case of wave trains, 
which are $2\pi$-periodic solutions that depend only on the phase variable 
\[
\zeta=(\x-\mathbf{c} t)\cdot\k_\crit = \kcx x+\kcy y + \omega t,
\]
so that $\partial_x = \kcx\partial_\zeta$, $\partial_y=\kcy\partial_\zeta$ and $\partial_t$ becomes $\partial_t + \omega\partial_\zeta$. 
Hence, in terms of $\zeta$, and using 
$k_\crit^4 d \partial_\zeta^4 + k_\crit^2 b \partial_\zeta^2 + \lb_\crit/H_0 = d k_\crit^4(\partial_\zeta^2+1)^2$, 
for $\omega=-\omega_\crit$ we obtain the linear operator 
\begin{align}\label{e:HopfOp}
\calL_\crit\ := \ \begin{pmatrix}
\omega_\crit\partial_\zeta- d k_\crit^4(\partial_\zeta^2+1)^2 & f & -\kcx g \partial_\zeta\\
-f & \omega_\crit\partial_\zeta- d k_\crit^4(\partial_\zeta^2+1)^2 & -\kcy g \partial_\zeta\\
-\kcx H_0 \partial_\zeta & -\kcy H_0 \partial_\zeta  & \omega_\crit\partial_\zeta
\end{pmatrix}.
\end{align}
Its kernel is three-dimensional, spanned by $\ee_0=\tee_0$ from \eqref{e:rossbyeigenvector} and $\ee_1,\ee_{-1}$ that have the form
$\ee_j:=\E_j \rme^{\rmi j \zeta}$, $\E_j\in\C^3$, where $\E_{-j} = \overline{\E_j}$ and $j=\pm1$. 
We choose  
\[
\E_j = \begin{pmatrix}
\omega_\crit\kcx + j\rmi f\kcy \\
\omega_\crit\kcy - j\rmi f\kcx \\
k_\crit^2 H_0
\end{pmatrix}.
\]
The case $\omega=\omega_\crit$ is analogous. As in the 3D and 5D cases, the eigenvectors $\ee_{\pm1}$ depend on the bottom drag parameters only through $\k_\crit$. 

The structure of $\E_j,\,j=\pm 1$ shows that bifurcating solutions cannot be of the plane wave form \eqref{e:planeansatz}, 
which is required for the reduction to the scalar \eqref{e:planerescale}. Hence, an analysis of the arising bifurcation requires to consider the full system. 
These kernel eigenvectors correspond to so-called inertia-gravity waves \cite{pedlosky1987geophysical}, and bifurcation from these will be studied in \Cref{s:igw}. The isotropic case thus simultaneously generates geostrophic equilibria and inertia-gravity waves, and we expect also mixed waves. Their wave vectors must have length near the critical one, but can have arbitrary directions.

\section{Bifurcation of nonlinear geostrophic equilibria}\label{s:bif}

We combine the results of the previous section with the reduction presented in \Cref{s:sw} for steady plane wave-type solutions. Such bifurcations arise from the kernel of $\calL_\crit$ and in the previous section we found that this occurs in the anisotropic case as the primary instability with axis aligned wave vectors, and in the isotropic case jointly with oscillatory modes in any direction. For convenience, we here assume isotropy, which is not a restriction for what we consider, as noted in \cref{r:anisoscalar}.

In \eqref{e:sw} we thus look for bifurcating solutions of the form 
\begin{equation}\label{e:special}
\vv = \partial_\xi \phi(\k\cdot\x)\k^\perp,\quad \eta =\tilde f\phi(\k\cdot\x),\quad |\k| =|\k^\perp| \approx k_\crit,
\end{equation}
for wave shape $\phi$ and phase variable $\xi=\k\cdot\x$ 
which leads to the reduced 
steady state equation
\[
0=d k^4\partial_{\xi}^5\phi+b k^2\partial_{\xi}^3\phi+\frac{C+Qk|\partial_\xi\phi|}{H_0+\tf\phi}\partial_{\xi}\phi, \quad  k=|\k| =|\k^\perp|.
\]
We recall that $k$ is the wave number for the nonlinear plane wave-type solutions. We remark as before that the nonlinear terms stems from the bottom drag only. 
We also recall from \cref{r:geostrophic} that the bifurcating solutions we investigate here can be viewed as (nonlinear) geostrophic equilibria (GE).

\medskip
\paragraph{Linear analysis}
We augment the stability analysis of the trivial state $(\vv,\eta)=(0,0)$ in \Cref{s:Turing} by including changes in wave number such that $|\k| = k_\crit+\kappa$ with $\kappa\approx 0$.  We recall the bifurcation parameter $\alpha$ such that $\lb = \lb_\crit-\alpha H_0$. For $\mu = (\alpha,\kappa)\approx(0,0)$ 
we expand $\lambda = \lambda(\mu)$ solving \eqref{e:poly} with $\lambda(0) = 0$ 
as
\begin{equation}\label{e:zerospec}
\lambda(\mu) =  M \big(\alpha - 2b\kappa^2 + \calO(|\kappa|^3)\big), 
\end{equation}
where 
$M:= gk_\crit^2H_0/\omega_\crit^2>0$. Hence, the trivial state $(\vv,\eta)=(0,0)$ (i.e.\ $\phi(\xi)\equiv0$) is unstable for $\alpha > 2b\kappa^2$ and the stability boundary is given by $\alpha = 2b\kappa^2$ at leading order.

\medskip
\paragraph{Steady state equation}
We consider the steady state equation in the domain $\xi\in[0, 2\pi]$ under periodic boundary conditions and nonlinear wave number $k=k_\crit+\kappa$.
With parameters $\mu=(\alpha,\kappa)\approx (0,0)$, we thus obtain 
\begin{equation}
G(\phi,\mu):=dk^4\partial_{\xi}^5\phi+bk^2\partial_{\xi}^3\phi+
\frac{\lb_\crit-\alpha H_0+ Q k|\partial_\xi\phi| }{H_0+\tf\phi}\partial_\xi\phi 
= 0, \label{e:normProb2}
\end{equation}
where $G: \Hspace_\per^5\times \R^2 \to \Lspace^2$ is Fr\'echet differentiable. In lack of a reference we include the following proof, which in particular shows the differentiability of the non-smooth term.
\begin{lemma}\label{l:nem}
Let $f:I\to\R$, $I\subset\R$ be an interval.
For the Nemitsky operator $f_N:\mathcal{U}\subset \Hspace^1_\per\to \Lspace^2$ defined by $f_N(\varphi)(x)= f(\varphi(x))$, with $\mathcal{U}$ such that $f_N$ is well-defined, the following holds: 
(a) If $f$ is Lipschitz continuous, then so is $f_N$. (b) If $f$ is $C^k$ smooth for $k\in \N$, then so is $f_N$. 
\end{lemma}

\begin{proof} The Nemitsky operator is well-defined since $\varphi\in \Hspace_\per^1$ is continuous and thus has bounded range, so that $\|f_N(\varphi)\|_2 \leq 2\pi \|f_N(\varphi)\|_\infty < \infty$ in all cases. 

(a) For $\varphi,\psi\in\Lspace^2$ and $L$ the Lipschitz-constant of $f$ we have $\|f_N(\varphi)-f_N(\psi)\|_2\leq L\|\varphi-\psi\|_2$. 

(b) We consider $k=1$ in detail. For $y\in I$, by smoothness of $f$ we have $R(z;y) := f(y+z)-f(y)-f'(y)z = o(|z|)$, i.e.\ for any $\eps>0$ there is $\delta>0$ such that $|z|\leq \delta$ implies $|R(z;y)|\leq \eps |z|$, locally uniformly in $y$. 
For the Nemitsky operator we get $(f_N(\varphi +h)-f_N(\varphi))(x)-f'(\varphi(x))h(x) = R(h(x);\varphi(x))$.  
With the Sobolev embedding $\|\cdot\|_\infty \leq C \|\cdot\|_{\Hspace^1}$, for $\varphi,h\in \Hspace_\per^1$ with $\|h\|_{\Hspace^1} \leq \delta/C$ we thus have $|R(h(x);\varphi(x))| \leq \eps |h(x)|$ for all $x\in[0,2\pi]$, which implies $\|R(h(\cdot);\varphi(\cdot))\|_2 \leq \eps \|h\|_2$. From $\|\cdot\|_2\leq \|\cdot\|_{\Hspace^1}$ we obtain $\|R(h(\cdot);\varphi(\cdot))\|_2/\|h\|_{\Hspace^1} \leq \eps$, which implies the Fr\'echet derivative of $f_N$ is $f'(\varphi(\cdot))$. 
For general $k$ the proof is the same upon replacing $R$ by the appropriate remainder term of the Taylor expansion of $f$. 
\end{proof}

This lemma implies differentiability of the full operator $G$ due to the higher order smoothness in the domain of $G$ and the differentiability of products of differentiable functions.

Due to the mass conservation, any constant $\phi$ solves \eqref{e:normProb2}, and without loss we consider the bifurcation from the zero state, 
cf.\ \Cref{s:sw}.
The derivative 
\[
\calL_0:=\partial_{\phi} G(0,0)=dk_\crit^4\partial_{\xi}^5+bk_\crit^2\partial_{\xi}^3+ \frac{\lb_\crit} {H_0}
\partial_\xi :\Hspace_\per^5\to\Lspace^2,
\]
is a Fredholm operator with index zero and with the kernel $\ker(\calL_0)=\mathrm{span}\{\e_j: j=0,\pm1\}$, $\e_j:= \rme^{\rmi j\xi}$. By Fredholm properties the domain and range of $\calL_0$ split as $\Hspace_\per^5=\ker(\calL_0)\oplus \setM$, with $\setM:=\ker(\calL_0)^\perp$, and $\Lspace^2= \ker(\calL_0^*)\oplus  \range(\calL_0)$, where $\range(\calL_0)^\perp=\ker(\calL_0^*)$ with respect to the $\Lspace^2$-inner product $\langle u,v \rangle_{\Lspace^2}=\frac{1}{2\pi}\int_0^{2\pi}u\bar v\dif\xi$. 
The kernel of the adjoint operator is $\ker(\calL_0^*)=\mathrm{span}\{\e_j^*: j=0,\pm1\}$ with $\e_j^*=\e_j$,
analogous to $\ker(\calL_0)$. When there is no ambiguity, in the remainder of \Cref{s:bif} we denote  $\langle \cdot,\cdot \rangle_{\Lspace^2}$ as $\langle \cdot,\cdot \rangle$.

\medskip
\paragraph{Lyapunov-Schmidt reduction}
The above decompositions define projections $\tilde P: H_\per^5\to \ker(\calL_0)$ along $\setM$ and $P:\Lspace^2\to \range(\calL_0)$ along $\ker(\calL_0^*)$.
The projection $\tilde P$ in the splitting $\phi = u+w$ with $u\in\ker(\calL_0)$ and $w\in \setM$ can be written as 
\[
u=\tilde P\phi:= \sum_{j=-1}^1\langle \phi,\e_j\rangle_{H^5}\e_j, \quad\langle \phi,\e_j\rangle_{H^5}=\sum_{n=0}^5\langle \partial_\xi^n  \phi,\e_j\rangle_{L^2}.
\]
Integration by parts under periodic boundary conditions gives $\langle \partial_\xi^{n+1} \phi,\e_j\rangle_{L^2} = \rmi j \langle \partial_\xi^n \phi,\e_j\rangle_{L^2}$, so that for $j=0,\pm1$ we have $\langle \phi,\e_j\rangle_{H^5}=(1+\rmi j)\langle \phi,\e_j\rangle_{L^2}$.
Thus, the inner products $\langle\cdot,\cdot\rangle_{H^5}$ and $\langle\cdot,\cdot\rangle_{L^2}$ are equivalent concerning the orthogonality between $\ker(\calL_0)$ and $\setM$. 
The projection $P$ in the decomposition of $L^2$  can be written as
\begin{equation}\label{e:projP}
P := \Id - \sum_{j=-1}^1 \langle \cdot , \e_j \rangle_{L^2}\e_j.
\end{equation}
We consider $u\in \ker (\calL_0)$, $w\in \setM$ and the projected problem
\begin{equation}\label{e:PG}
P G(u+ w,\mu)=0.
\end{equation}
Differentiating \eqref{e:PG} with respect to $w$ at zero gives $P\partial_\phi G(0,0) = P\calL_0 = \calL_0: \setM\to \range(\calL_0)$.
By the implicit function theorem, there is an open neighbourhood of $(0,0,0)$ in $\ker(\calL_0)\times \setM\times\R^2$  of the form $N_0\times M_0\times (-\eps,\eps)^2$ and a unique function  $W:N_0\times (-\eps,\eps)^2\to  M_0$ such that $W(0,0)=0$ and $w=W(u,\mu)$ solves \eqref{e:PG} for all $(u,\mu)\in  N_0\times (-\eps,\eps)^2$.
In order to solve \eqref{e:normProb2} it thus remains to determine $(u,\mu)\in N_0\times (-\eps,\eps)^2$ 
such that 
\[
(\mathrm{Id}- P) G(u+W(u,\mu),\mu)=0,
\]
which is equivalent to the bifurcation equations
\begin{equation}\label{e:bifross}
\langle G(u+W(u,\mu),\mu),\e_j\rangle=0,\quad j=0,\pm1.
\end{equation}
We write $G$ as
\[
G(\phi,\mu) = \calL_0\phi + L_\mu\phi + N_\lb(\phi,\mu) + N_Q(\phi,\mu),
\]
in terms of the operators defined next, where we denote $\phi_\xi = \partial_\xi\phi$:
\begin{subequations}\label{e:op}
\begin{align}
L_\mu\  :=&\  d(k^4-k_\crit^4)\partial_\xi^5 + b(k^2-k_\crit^2)\partial_\xi^3 - \alpha\partial_\xi, \\
N_\lb(\phi,\mu)\ :=&\ (\lb_\crit -\alpha H_0)\left(\frac{1}{H_0 + \tf\phi} - \frac{1}{H_0}\right)\phi_\xi 
\nonumber \\
=&\ -\frac{\lb_\crit - \alpha H_0}{H_0^2}\left(\tf\phi-\frac{\tf^2\phi^2}{H_0} + \tf^3\calO(|\phi|^3)\right)\phi_\xi, \label{e:NC}\\
N_Q(\phi,\mu) \ :=&\ \frac{Qk|\phi_\xi|\phi_\xi}{H_0+\tf\phi} 
\ =\ \frac{Qk}{H_0}|\phi_\xi|\phi_\xi - \frac{Qk}{H_0^2} |\phi_\xi|\phi_\xi\left( \tf\phi-\frac{\tf^2\phi^2}{H_0} +\tf^3\calO(|\phi|^3)\right).
\end{align}
\end{subequations}
Using $\calL_0 u=0$, $\langle \calL_0 w, \e_j\rangle =0$ and with $\phi = u + W(u,\mu)$, equations \eqref{e:bifross} become
\begin{equation}\label{e:bifG}
\langle G(\phi,\mu), \e_j\rangle = \langle L_\mu\phi, \e_j \rangle + \langle N_\lb(\phi,\mu),\e_j \rangle + \langle N_Q(\phi,\mu), \e_j\rangle=0 ,\quad j=0,\pm1.
\end{equation}
There are parameters $A_j\in\C,\,j=0,\pm 1$ in a neighbourhood of zero, such that $u=A_0\e_0+ A_1\e_1 + A_{-1}\e_{-1}\in N_0\subset\ker(\calL_0)$. For the real bifurcating solutions we have $A_0\in\R$ and $A_{-1} = \overline{A_1}$. Moreover, we can assume $A_1 = A_{-1}\in\R$ due to the translation symmetry. 
Due to $W\in\setM$ and $\e_0\in\ker(\calL_0)$ it follows $\langle W,\e_0\rangle=0$, so $W$ has always zero mean.
Thus, nonzero mean comes from the constant contribution $A_0\neq0$ only.
Since we consider zero mean for $\eta$, we set $A_0=0$. 
Denoting $u_\xi = \partial_\xi u$, this gives
\begin{equation}\label{e:u}
u = 2A_1\cos(\xi),\quad u_\xi = \rmi(A_1\e_1 - A_{-1}\e_{-1}) = -2A_1\sin(\xi),
\end{equation}
as well as $\langle u_\xi,\e_0\rangle=0$, $\langle u_\xi,\e_1\rangle=\rmi A_1$, and $\langle u_\xi,\e_{-1}\rangle=-\rmi A_{-1}$.

\subsection{Smooth case}

We first discuss the bifurcation and expansion of the small amplitude plane wave-type solutions near $\mu =(\alpha,\kappa)=0$ with smooth bottom drag ($Q=0$). To state this, we recall the critical wave number $k_\crit$ as in \eqref{e:isocritk}, and the bottom drag $\lb = \lb_\crit - \alpha H_0$ with $\lb_\crit$ as in \eqref{e:isocritC}.
\begin{theorem}[Bifurcation of GE for $Q=0$]\label{t:bifQ0}
Let $Q=0$, $\alpha,\kappa\in\R$ sufficiently close to zero, and $\k\in\R^2$ with $|\k| =k_\crit+\kappa$. Consider steady plane wave-type solutions to \eqref{e:sw} of the form \eqref{e:special} 
with $2\pi$-periodic mean zero $\phi$ 
and $\xi=\k\cdot\x$.
These waves are (up to spatial translations) in one-to-one correspondence with solutions $A_1\in\R$ near zero of
\begin{equation}
0\ =\ A_1\left(
\alpha  -  2b\kappa^2  - \frac{17b^2\tf^2}{72 d  H_0^2}A_1^2 + \calR_\s \right),\label{e:bifQ0}
\end{equation}
with remainder term $\calR_\s = \tf \calO(|A_1|^2 (|A_1|^2+|\alpha| +|\kappa|)) +\calO(|\kappa|^3)$. In addition, 
\begin{equation}\label{e:solQ0}
\phi(\xi) = \phi_\s(\xi;\mu)\ =\ 2A_1\cos(\xi) + \frac{\tf}{9 H_0}A_1^2 \cos(2\xi) + \tf\calO\left(|A_1|^2 (|A_1| +|\alpha|+|\kappa|)\right).
\end{equation}
\end{theorem}
For $\tf\neq0$, since the coefficient of $A_1^3$ in \eqref{e:bifQ0} is negative and the zero state is unstable for $\alpha > 2b\kappa^2$, we always obtain a supercritical pitchfork bifurcation, cf.\ \cref{f:ampalp}. After bifurcation, the leading order amplitude is given by 
\[
|A_1| = \frac{1}{|\tf|}\sqrt{\frac{72 d H_0^2}{17b^2}(\alpha - 2b \kappa^2 )}.
\]
For any fixed $\alpha>0$ the amplitude $A_1=A_1(\kappa)$ forms a semi-ellipse with maximum at $\kappa=0$ and $|A_1|>0$ for $\kappa\in(-\kappa_0,\kappa_0)$, where $\kappa_0=\sqrt{\alpha/(2b)}$ is independent of $\tf$. The amplitude is proportional to $\sqrt{d\alpha}$, while it is inversely proportional to $b\tf$. Hence, existence requires  the backscatter to be nonzero,
and  as $b\to 0$ the amplitude diverges pointwise in $\kappa\in(-\kappa_0,\kappa_0)$, where $\kappa_0$ remains $\eps$-independent, if the drag $\alpha$ scales as the backscatter term $b$. 
The amplitudes also diverge pointwise for equally small backscatter parameters $(d,b)=(\eps\hat d,\eps\hat b)$, as $\eps\to 0$, with $|A_1| = \calO(\eps^{-1/2})$, while the wave number is fixed at $k=k_\crit = (\hat b/(2\hat d))^{1/2}$. 
 See Remark~\ref{r:smallback}. 

For $\tf\to 0$ the amplitude $A_1$ also diverges, as illustrated in \cref{f:ampkap}, but that of $\eta=\tf\phi$ does not. At $\tf=0$, the leading order bifurcation equation is given by 
\[
0 = A_1(\alpha - 2b\kappa^2).
\]
Hence, the value of $A_1$ can be arbitrary at $\alpha = 2b\kappa^2$ (or at $\kappa=\pm\kappa_0$), forming `vertical' bifurcating branches, cf.\ \cref{f:amp}. For these solutions the wave shapes are sinusoidal $\phi_\s(\xi) = 2A_1\cos(\xi)\in N_0$, i.e.\ $w=W(u,\mu)\equiv0$ from the splitting $\phi = u+w$ with \eqref{e:u}. Indeed, revisiting the wave-type solution \eqref{e:special}, $\eta = \tf\phi_\s=0$ for $\tf = 0$, and the amplitude of the velocity $\vv = \partial_\xi\phi_\s(\xi)\k^\perp = -2A_1\sin(\xi)\k^\perp$ is free. 
In \Cref{s:explicit}, we study related explicit flows. See \cref{r:isobifflows} and \cref{r:anisobifflows} for details.
\begin{figure}[t!]
\centering
\subfigure[$\kappa=0$]{\includegraphics[trim = 5cm 8.5cm 5.5cm 8.5cm, clip, height=4cm]{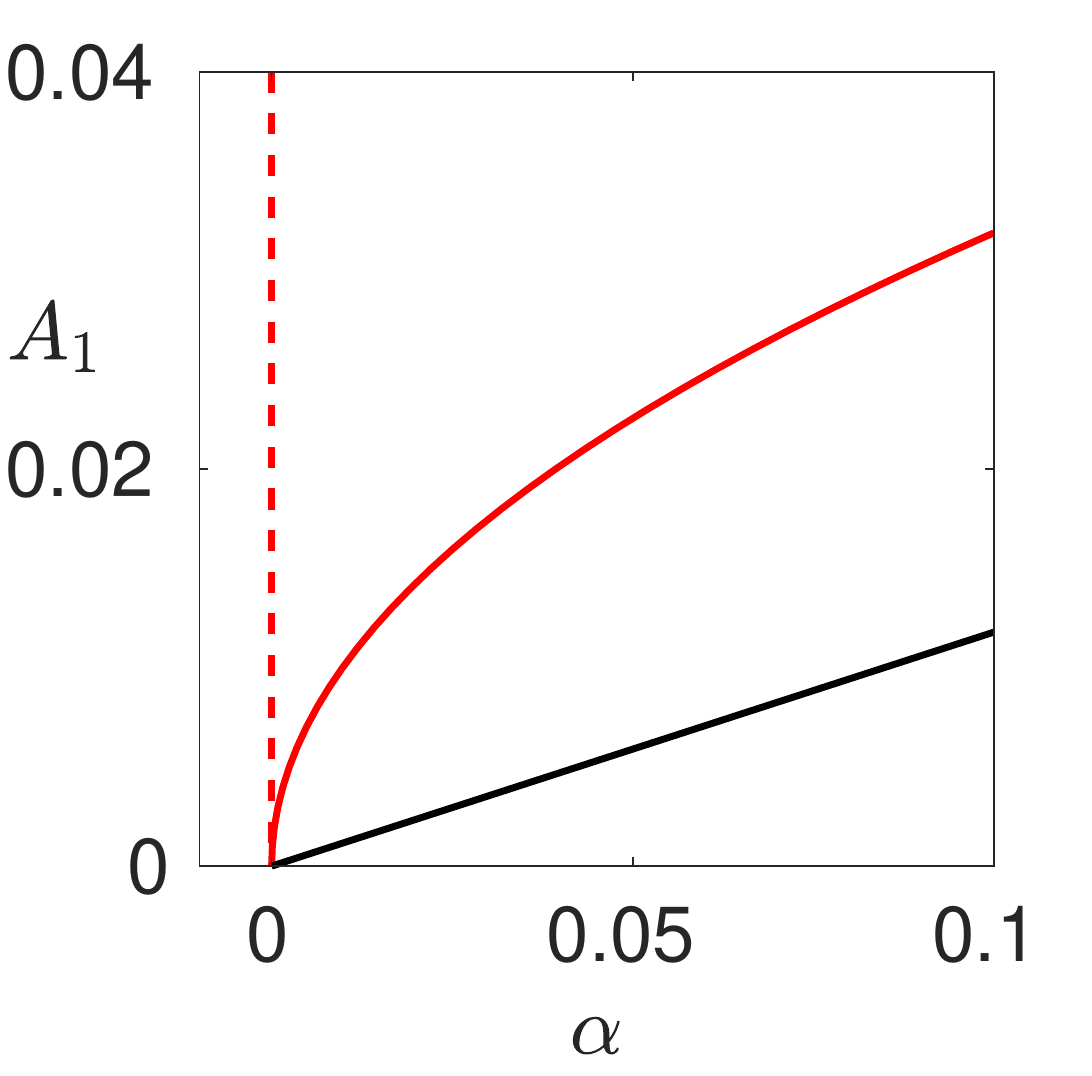}\label{f:ampalp}}
\hfil
\subfigure[$\alpha = 0.1$]{\includegraphics[trim = 5cm 8.5cm 5.5cm 8.5cm, clip, height=4cm]{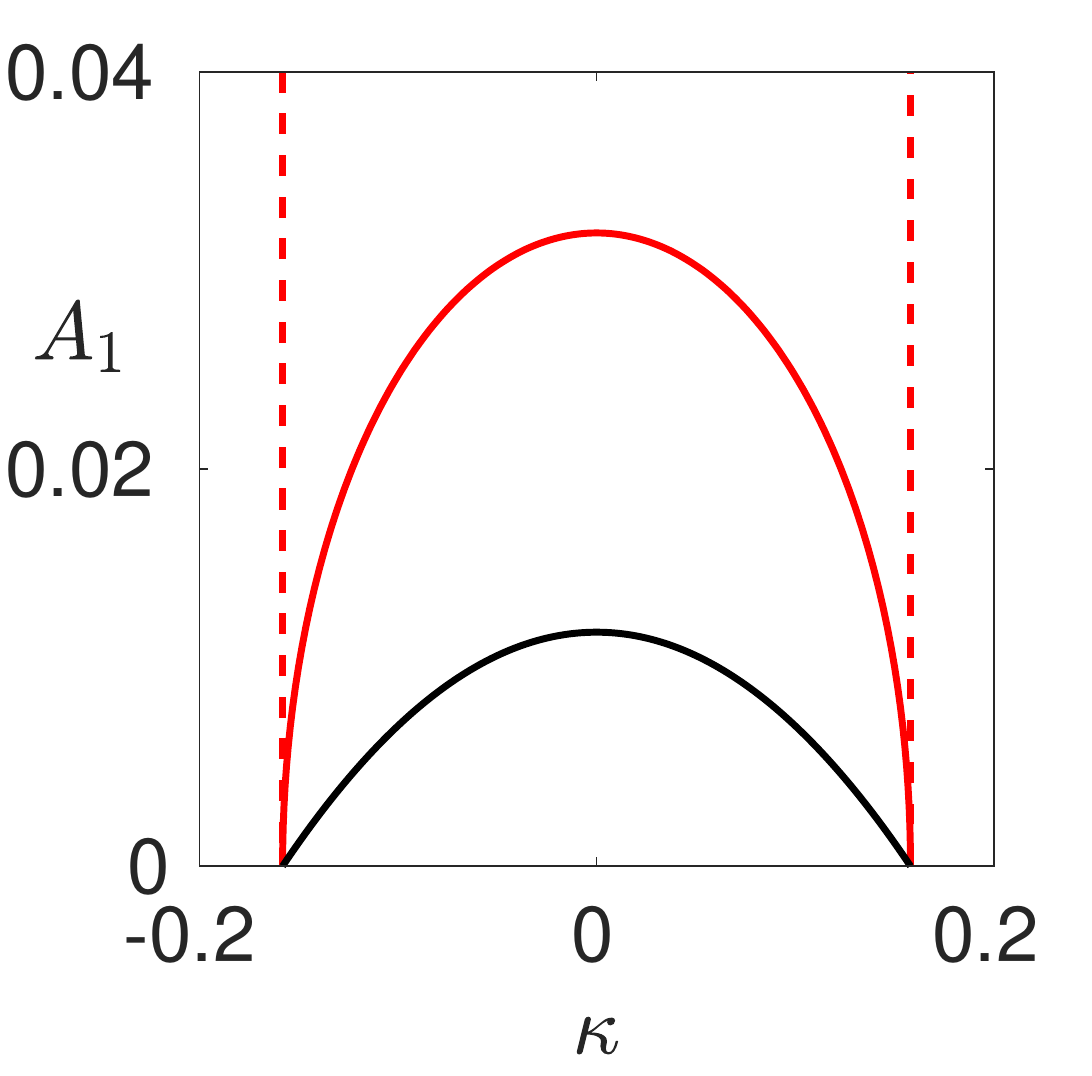}\label{f:ampkap}}
\caption{Leading order amplitudes in the isotropic case $d_1=d_2=1$, $b_1=b_2=2$ for smooth ($Q=0$, red) and non-smooth ($Q=0.5$, black) cases. Other parameters are $g=9.8$, $H_0=0.1$ so that $\lb_\crit = 0.1$. In smooth case, the amplitudes are plotted for $f=10$ (solid) and $f=0$ (dashed); in the non-smooth case, the amplitude is to leading order independent of $f$.}
\label{f:amp}
\end{figure}

\begin{proof}[Proof of \cref{t:bifQ0}]
For $Q=0$, we first solve $W$ from the equation \eqref{e:PG}. We write $PG(\phi,\mu) = P\calL_0\phi + PL_\mu\phi + PN_\lb(\phi,\mu)$. Then we rewrite \eqref{e:PG} as the fixed point equation for $W$ given by
\[
P\calL_0 W = -PL_\mu (u + W) - PN_\lb(u+W,\mu).
\]
Using implicit function theorem, and expanding $W = W(u,\mu)$ in $(A_1,\alpha,\kappa)$ 
near zero, yields
\[
W(u,\mu)\ =\ \sum_{j=\pm 1} \frac{\lb_\crit \tf}{2 H_0^2 L_2} \rme^{2\rmi j \xi} A_j^2 + \widetilde\calR_W
=\ \frac{\tf}{9 H_0} A_1^2 \cos(2\xi) + \widetilde\calR_W,
\]
where 
$L_2:=16dk_\crit^4-4bk_\crit^2+\lb_\crit/H_0 = 9b^2/(4d)>0$, 
$\widetilde\calR_W = \tf\calO\left(|A_1|^2 (|A_1| +|\alpha|+|\kappa|)\right)$.
Hence, the wave shape has the form as in \eqref{e:solQ0}.

Next we consider the projections \eqref{e:bifG}. For $j=0$, the projection is trivial, i.e.\ $\langle G,\e_0 \rangle\equiv0$, 
since the linear part $\langle\partial_\xi^j\phi,\e_0 \rangle = 0$ for $j=1,3,5$, and the integration of the full nonlinear part in \eqref{e:normProb2} is given by $\int_0^{2\pi}(H_0+\tf\phi)^{-1}\phi_\xi\dif\xi=0$.
Next we consider $j=1$. We note that the same result applies to $j=-1$ due to the assumption $A_1=A_{-1}\in\R$. For the linear part $\langle L_\mu\phi,\e_1 \rangle$ we use the integration by parts and the consideration of periodic boundary conditions to show 
$\langle W_\xi,\e_1\rangle=[W\rme^{-\rmi\xi}]_0^{2\pi}/2\pi+\rmi\langle W,\e_1\rangle=0$,
since $\e_1\in\ker(\calL_0)$ and $W\in \setM$. This yields $\langle \phi_\xi,\e_1 \rangle = \langle u_\xi,\e_1 \rangle = \rmi A_1$, and in the same way one can show $\langle\partial^3_\xi\phi,\e_1\rangle=-\rmi A_1$ and $\langle\partial^5_\xi\phi,\e_1\rangle=\rmi A_1$. This results in 
\begin{equation}\label{e:bifGlin}
\langle L_\mu\phi,\e_1 \rangle = - \rmi A_1 \left( \alpha - 2b\kappa^2 - 4dk_\crit \kappa^3 - d\kappa^4\right).
\end{equation}
We consider the nonlinear part that is involving the smooth nonlinear term \eqref{e:NC}. As to the first and second terms, we respectively compute
\begin{align*}
\langle \phi\phi_\xi,\e_1 \rangle &= \frac{\rmi\tf}{18H_0}A_1^3 + \calO(|A_1|^3(|\alpha| + |\kappa|)),\\
\langle \phi^2\phi_\xi,\e_1 \rangle &= \rmi A_1^3 + \calO(|A_1|^5).
\end{align*}
The remainder terms in $\langle N_\lb,\e_1 \rangle$ are of order $\calO(|A_1|^5)$. 
Combining this with the linear and nonlinear terms above, \eqref{e:bifG} becomes
\[
\langle G(\phi,\mu),\e_1 \rangle = -\rmi A_1\left( \alpha - 2b\kappa^2 
- \frac{17\lb_\crit\tf^2}{18H_0^3}A_1^2 + \calR_\s\right),
\]
where $\calR_\s$ as given in the statement of the theorem. 
Substituting the value \eqref{e:isocritC} and dividing out the factor $-\rmi$, it follows the bifurcation equation \eqref{e:bifQ0}. 
\end{proof}

\subsection{Non-smooth case}\label{s:nonsm}
For the case $Q\neq0$ we note that $ G$ is not smooth in $\phi$, but at least once continuously differentiable in $\phi$, so that we can still use the same method to derive the bifurcation equations. 
We first discuss the bifurcation in the following theorem. 
\begin{theorem}[Bifurcation of GE for $Q\neq 0$]\label{t:bifQn0}
Let $Q\neq0$, $\alpha,\kappa\in\R$ sufficiently close to zero, and $\k\in\R^2$ arbitrary with $|\k| =k_\crit+\kappa$. Consider steady plane wave-type solutions to \eqref{e:sw} of the form \eqref{e:special} with 
$2\pi$-periodic mean zero $\phi$ and $\xi=\k\cdot\x$.
These waves are (up to spatial translations) in one-to-one correspondence with solutions $A_1\in\R$ near zero of
\begin{equation}\label{e:bifQn0}
0 = A_1\left( 
\alpha - 2b\kappa^2 - \frac{16 Q \sqrt b}{3\pi H_0 \sqrt{2d}}|A_1| +  \calR_\ns \right),
\end{equation}
with remainder term $\calR_\ns = Q\tf\calO(|A_1|^2)+(Q+\tf)\calO(|A_1|(|A_1|+|\alpha|+|\kappa|))+\calO(|\kappa|^3)$. In addition,
\[
\phi(\xi) = \phi_\ns(\xi;\mu) = 
2A_1\cos(\xi) + \calO(|A_1|(|A_1| + |\alpha| +|\kappa|)).
\]
\end{theorem}
Since the coefficient of $A_1|A_1|$ is negative and the zero state is unstable for $\alpha>2b\kappa^2$ (at leading order, cf.\ \eqref{e:zerospec}), the bifurcation is always supercritical. It is a degenerate pitchfork bifurcation as in \cite{SR2020}, where the bifurcating branch behaves linearly in $\alpha$ near zero,  rather than the square root form for the smooth case, cf.\ \cref{f:ampalp}. In contrast to the smooth case, the coefficient of $A_1|A_1|$ in \eqref{e:bifQn0} is independent of Coriolis parameter $\tf$. This implies that the balance between linear bottom drag and Coriolis force is invisible at leading order, and thus the `vertical branch' does not occur for any $\tf$. After bifurcation, the leading order amplitude is given by
\[
|A_1| = \frac{3\pi H_0 \sqrt{2d}}{16Q \sqrt b} \left( \alpha - 2b\kappa^2 \right). 
\]
We readily see that it behaves parabolic in $\kappa$, cf.\ \cref{f:ampkap}; $Q$ enters through only a prefactor so that the value of amplitude is monotonically decreasing in $Q$. In contrast to the smooth case, for small backscatter parameter $(d,b)=(\eps\hat d, \eps\hat b)$ as $\eps\to 0$, the prefactor is independent of $\eps$ so that the `vertical branch' does not occur. 
We note, with the rescaling $\kappa = \eps k_\crit K$, that the expression of $|A_1|$ is the same as that of $|R|$ from \eqref{e:modamp}, where $\eps$ and $K$ are defined in \Cref{s:modulation}. 

\begin{remark}
For $Q\rightarrow0$ in \eqref{e:bifQn0} the quadratic term in $|A_1|$ goes to zero and the remainder term limits to $\calR_\ns = \tf\calO(|A_1|(|A_1|+|\alpha|+|\kappa|))+\calO(|\kappa|^3)$. 
Compared with the bifurcation equation \eqref{e:bifQ0} at $Q=0$ there may be additional terms of order $\tf\calO(|A_1|(|\alpha|+|\kappa|))$. 
We expect that it can be shown with a refined analysis that such terms do not occur, which would prove that the bifurcation equation \eqref{e:bifQn0} is continuous with respect to $Q$ at $Q=0$. 
\end{remark}

\begin{remark}[Stability of steady solutions]
For $f=0$ we consider the evolution equation \eqref{e:scalef0}, which provides temporal dynamics and therefore we are able to determine stability aspects of bifurcating steady solutions. Due to above proof of supercriticality, these are stable under perturbations of the same form \eqref{e:special} and with the same wave vector. The Eckhaus region discussed in \Cref{s:modulation} gives additional stable bifurcating solutions whose wave vectors are close enough to the critical one. For $f\neq0$ one has to transform the system of equations \eqref{e:rescale} into an evolution equation first, which yields an integro-differential equation. However, its stability analysis is beyond the scope of this paper.
\end{remark}

In the remainder of the present subsection, we give a proof of \cref{t:bifQn0}. In contrast to the smooth case, we need to estimate the function $W$ of the decomposition $\phi = u+W(u,\mu)$ instead of expanding it, since $G$ does not have enough smoothness in $\phi$. Since $W$ inherits the differentiability of $G$,
we can approximate near the zero state via
\[
W(u,\mu)=\partial_u W(0,0)u+\partial_\alpha W(0,0)\alpha+\partial_\kappa W(0,0)\kappa+o(\|u\|_{H^5}+|\mu|).
\]
Since $W$ solves \eqref{e:PG}, it follows from differentiation that the first derivatives of $W$ are zero, which means $W(u,\mu)=o(\|u\|_{H^5}+|\mu|)$. 
We can refine this estimate as follows: since $W(u,\mu)= \calO((|\mu| + \|u\|_{H^5})\|u\|_{H^5})$, which is shown in \Cref{s:proofWestX}, and $\|u\|_{\Hspace^5(0,2\pi)} = \calO(|A_1|)$ from \eqref{e:u}, we obtain the estimate 
\begin{equation}\label{e:estW}
W(A_1,\alpha,\kappa) = \calO(|A_1|(|A_1|+|\alpha|+|\kappa|)).
\end{equation}

\begin{remark}\label{r:sobolev}
Since we consider $u,W\in\Hspace^5(0,2\pi)$, it follows by the Sobolev embedding theorem that $u,W\in C^{4,1/2}([0,2\pi])$. In addition, there is a constant $a>0$ so that $\|v\|_{C^{4,1/2}}\leq a\|v\|_{H^5}$ for all $v\in H^5$.
\end{remark}

\begin{proof}[Proof of \cref{t:bifQn0}]
With the given approximations we can estimate the smooth and non-smooth nonlinear parts in the projection \eqref{e:bifG}. We first consider $j=0$. As noted we make the splitting $\phi = u+w$, where $\phi\in H_\per^5$, and $u$ from \eqref{e:u} is even in $\xi$. As proved in \Cref{s:bif}, by the implicit function theorem there exists a unique function $W$ with values in $M_0$ such that $w=W(u,\mu)$ solves the projected problem \eqref{e:PG} near zero. We prove in \Cref{s:proofproj0} that in the space of even function in $H_\per^5$, i.e.\ for even $\phi\in H_\per^5$, such a function $W$ is even in $\xi$. Combining these we conclude that the solution $\phi = u+W$ to the steady state equation \eqref{e:normProb2} is always even. The operators $L_\mu, N_\lb, N_Q$ from \eqref{e:op} map even functions to odd functions, thus for any even and $2\pi$-periodic function $\phi$ the functions $L_\mu\phi$, $N_\lb(\phi)$, $N_Q(\phi)$ have zero mean on any interval of length $2\pi$. It follows that for such $\phi$ the right-hand side of \eqref{e:bifG} vanishes for $j=0$, i.e.,
\[
\langle G(\phi,\mu),\e_0\rangle \equiv 0.
\]

We next consider $j=1$; for $j=-1$ we get the same result due to the choice $A_1=A_{-1}\in\R$. 
Defining $\calR_W:=\calO(|A_1|(|A_1|+|\alpha|+|\kappa|))$ and using \cref{r:sobolev}, 
we approximate the absolute value 
\[
\bigl||u_\xi+W_\xi| - |u_\xi|\bigr| \leq |u_\xi + W_\xi - u_\xi| \leq a\|W\|_{H^5} = \calR_W,
\]
which in particular means $|u_\xi+W_\xi| = |u_\xi| + \calR_W$. 
This gives
\[
|u_\xi+W_\xi|(u_\xi+W_\xi) = |u_\xi|u_\xi + |u_\xi|W_\xi + (u_\xi+W_\xi)\calR_W.
\]
Thus, the leading order part in $\langle N_Q,\e_1 \rangle$ can be approximated by
\[
\langle |\phi_\xi|\phi_\xi, \e_1 \rangle = \langle |u_\xi|u_\xi ,\e_1 \rangle + \langle |u_\xi|W_\xi + (u_\xi+W_\xi)\calR_W, \e_1\rangle 
= \frac{16\rmi}{3\pi} |A_1|A_1 + A_1\calR_W, 
\]
where we used the form of $u_\xi$ given in \eqref{e:u}. 
Note that the coefficient $\frac{16}{3\pi}\rmi$ results from the nature of the absolute value and reflects the non-smooth character, similar to \cite{SR2020}.
The higher order parts of $\langle N_Q,\e_1 \rangle$ are 
\[
\langle |\phi_\xi|\phi_\xi(\tf\phi-\frac{\tf^2\phi^2}{H_0} + \tf^3\calO(|\phi|^3)),\e_1 \rangle = \tf A_1\calO(|A_1|^2).
\]
The smooth nonlinear term $\langle N_\lb,\e_1 \rangle$ can be estimated by
\[
\langle (\tf\phi-\frac{\tf^2\phi^2}{H_0} + \tf^3\calO(|\phi|^3))\phi_\xi,\e_1 \rangle = \tf A_1\calR_W.
\]
Analogous to the proof of \cref{t:bifQ0} we have $\langle \phi_\xi,\e_1 \rangle = \langle\partial^5_\xi\phi,\e_1\rangle =\rmi A_1$ and $\langle\partial^3_\xi\phi,\e_1\rangle=-\rmi A_1$. 
Combining these estimates with the linear part \eqref{e:bifGlin}, \eqref{e:bifG} becomes
\[
\langle G(\phi,\mu),\e_1 \rangle = 
-\rmi A_1\left( 
\alpha  - 2b\kappa^2 - \frac{16Q k_\crit}{3\pi H_0}|A_1| + \calR_\ns \right),
\]
where $\calR_\ns = (Q+\tf)\calR_W + Q\tf\calO(|A_1|^2) +\calO(|\kappa|^3)$.
Substituting the value \eqref{e:isocritk} 
and dividing out the factor $-\rmi$, it follows the bifurcation equation \eqref{e:bifQn0}.
\end{proof}

\section{Bifurcation of nonlinear inertia-gravity waves}\label{s:igw}

In this section we consider one-dimen\-si\-onal bifurcations due to the purely imaginary spectrum in the case of marginal stability for isotropic backscatter. 
We recall from \Cref{s:spec} that in the isotropic case pattern forming stationary and oscillatory modes destabilize simultaneously. For anisotropic perturbations from isotropy the oscillatory modes destabilize slightly after the steady ones, and the bifurcations of IGW studied next are perturbed as shown numerically in \Cref{s:num}.
We focus entirely on the non-smooth case $Q\neq 0$, which gives non-standard bifurcation equations and is more subtle than that in \Cref{s:nonsm}.

To set up Lyapunov-Schmidt reduction, we cast \eqref{e:sw} as prepared by \eqref{e:HopfOp}, but with deviations $\kappa$ from the critical wave number $k_\crit$ and $s$ from the critical wave speed $-\omega_\crit$. That is, we choose $\k_\crit=(\kcx,\kcy)^\intercal$ satisfying \eqref{e:isocritk}, 
change variables to $\zz,\tilde \zz\in \R$ via  
\begin{equation}\label{e:xphase}
\x=\big((1+\kappa)^{-1}\zz+st\big)/|\k_\crit|^2 \;\k_\crit + \tilde \zz \k_\crit^\perp,
\end{equation}
and seek solutions that are independent of $\tilde \zz$.  
Then, with $U=(\vv,\eta)^\intercal$ and parameters $\mu=(\alpha,s,\kappa)$, the existence of such travelling waves of \eqref{e:sw} near the critical modes of \eqref{e:HopfOp} can be cast in the form
\begin{equation}\label{e:igwprob}
G(U;\mu) = \calL_\mu U - B_Q(U) - B(U;\mu) - N(U;\mu) =0,
\end{equation} 
with suitable linear $\calL_\mu$ and nonlinear $B_Q, B, N$ as defined in the following. 
The term $B(U;\mu)$ will contain all smooth quadratic terms with respect to $U$, $N(U;\mu) = \calO(|U|^3)$ the higher order terms, and 
\[
B_Q(U) = \begin{pmatrix}\frac{Q}{H_0}|\vv|\vv\\0\end{pmatrix}
\]
is the quadratic non-smooth term. 
The linear part $\calL_\mu$ at $\alpha=s=0$ is given by $\calL_\crit$ from \eqref{e:HopfOp} with $\partial_\zeta$ replaced by $(1+\kappa)\partial_\zz$, i.e.\ $\calL_\crit$ depends on $\kappa$. 
We split $\calL_\mu$ into $\alpha$-, $s$- and $\kappa$-dependent parts
\begin{equation}
\calL_\mu := \calL_\crit(\kappa) + 
 \alpha\,\mathrm{diag}(1,1,0) - (1+\kappa )s\,\mathrm{diag}(1,1,1) \partial_\zz.  \label{e:calLmu}
\end{equation} 
The linear operator $\calL_\crit(\kappa)$ is then given by
\begin{align}
\calL_\crit(\kappa) &:= \calL_\crit+ \kappa\calK,\label{e:calLkappa}\\[2mm]
\calK &:=
\begin{pmatrix}
\omega_\crit\partial_\zz- dk_\crit^4 \calS & 0 & -\kcx g\partial_\zz\\
0& \omega_\crit\partial_\zz- dk_\crit^4 \calS & -\kcy g\partial_\zz\\
-\kcx H_0\partial_\zz & -\kcy H_0 \partial_\zz & \omega_\crit\partial_\zz
\end{pmatrix},\nonumber\\[2mm]
\calS &:=(2+\kappa)\left(\kappa(2+\kappa) \partial_\zz^2 + 2(\partial_\zz^2+1)\right)\partial_\zz^2,\nonumber
\end{align}
where $\calL_\crit$ is the operator \eqref{e:HopfOp} with $\partial_\zeta$ replaced by $\partial_\zz$; note that $\calL_\crit(0)=\calL_\crit$. 
We consider $G:X\times \R^3\to Y$ with $X:=\Hspace_\per^4([0,2\pi])\times \Hspace_\per^4([0,2\pi])\times \Hspace_\per^1([0,2\pi])$ and $Y:=(\Lspace^2([0,2\pi]))^3$. Due to \cref{l:nem} $G$ is then continuously differentiable.
The following lemma admits to apply Lyapunov-Schmidt reduction in order to solve \eqref{e:igwprob} near $(U,\mu)=(0,0)$. 

\begin{lemma}\label{l:calLmuFredholm}
For $\mu=0$ the linear operator $\calL_\mu=\calL_\crit:X\to Y$, defined in \eqref{e:calLmu} and \eqref{e:HopfOp}, is a Fredholm operator with index zero. Its generalized kernel and that of its adjoint are three-dimensional.
\end{lemma}
\begin{proof}
We will show that for $\mu=0$ the operator $\calL_\mu=\calL_\crit(0)=\calL_\crit$ is a compact perturbation of a Fredholm operator. 
To see this, let us define the following diagonal operator and matrix Nemitsky operators. 
\begin{align*}
D&: X\to Y, &D&:=\mathrm{diag}\big( \widetilde{D}(\kcx), \widetilde{D}(\kcy),\omega_\crit\partial_\zz\big), &\\
&&&\widetilde{D}(k):=\big(\omega_\crit-gH_0k^2/\omega_\crit\big)\partial_\zz- d k_\crit^4(\partial_\zz^2+1)^2&\\
R&:Y\to Y,&R&:= \begin{pmatrix}1& 0&  -\kcx g /\omega_\crit\\0&1& -\kcy g /\omega_\crit\\0&0&1\end{pmatrix}&\\
S&:X\to X,&S&:=\begin{pmatrix}1& 0& 0\\0&1&0\\ -\kcx H_0 /\omega_\crit&-\kcy H_0 /\omega_\crit&1\end{pmatrix}&
\end{align*}
Each of the diagonal elements in $D$ is a Fredholm operator, which implies $D$ is. The matrix $R$ is invertible and since on $Y$ all components come from the same space, also $R$ is boundedly invertible. For the operator $S:X\to X$ the same argument applies to the upper left block since the first two components of $X$ are the same. The entire operator $S$ is well defined since the last row maps into $H^1([0,2\pi])$ due to the inclusion $H^4([0,2\pi])\subset H^1([0,2\pi])$. Using the inverse matrix, which is of the same form, this also implies that $S$ is boundedly invertible.

The product $\check\calL_\crit:=RDS$ of boundedly invertible and Fredholm operators is a Fredholm operator and takes the form
\[
\check\calL_\crit = 
\begin{pmatrix}
\omega_\crit\partial_\zz- d k_\crit^4(\partial_\zz^2+1)^2 & \frac{g H_0 \kcx\kcy}{\omega_\crit}\partial_\zz & -\kcx g \partial_\zz\\
 \frac{g H_0 \kcx\kcy}{\omega_\crit}\partial_\zz & \omega_\crit\partial_\zz- d k_\crit^4(\partial_\zz^2+1)^2 & -\kcy g \partial_\zz\\
-\kcx H_0 \partial_\zz & -\kcy H_0 \partial_\zz  & \omega_\crit\partial_\zz
\end{pmatrix}.
\]
The difference to $\calL_\crit$ reads
\[
\check\calL_\crit-\calL_\crit  = 
\begin{pmatrix}
0 & \frac{g H_0 \kcx\kcy}{\omega_\crit}\partial_\zz-f& 0\\
f+ \frac{g H_0 \kcx\kcy}{\omega_\crit}\partial_\zz & 0 & 0\\
0 & 0  &0
\end{pmatrix},
\]
which is a compact perturbation of $\calL_\crit$ since the range of $\check\calL_\crit-\calL_\crit$ lies in the compact subset $H^3([0,2\pi])\times H^3([0,2\pi])\times\{0\}$ of $Y$. 
Hence, the unperturbed operator $\calL_\crit$ shares the Fredholm property of $\check\calL_\crit$. The index is zero since kernel and kernel of adjoint have the same dimension.

The kernel of the adjoint operator $\calL_\crit^*$ to \eqref{e:HopfOp} is spanned by $\ee_0^*=\ee_0$ and $\ee_1^*,\ee_{-1}^*$ of the form $\ee_j^*:=\E_j^*\rme^{\rmi j \zz},\ \E_j^*\in\C^3$. We recall the eigenvectors $\ee_j=\E_j \rme^{\rmi j \zz}$ of $\calL_\crit$ and choose $\E_j^*$ as follows,
\begin{equation}\label{e:eigenvector2}
\E_j = \begin{pmatrix}
\omega_\crit \kcx + j\rmi f\kcy \\
\omega_\crit \kcy - j\rmi f\kcx \\
k_\crit^2 H_0
\end{pmatrix},
\quad 
\E_j^* =   \frac{1}{m}
\begin{pmatrix}
\omega_\crit \kcx + j\rmi f\kcy \\
\omega_\crit \kcy - j\rmi f\kcx \\
k_\crit^2 g
\end{pmatrix},
\end{equation}
with $m=2 \omega_\crit^2 k_\crit^2$ so that $\langle \E_j , \E_j^*\rangle = 1$ for $j=0,-1,1$. 

Consider the dispersion relation \eqref{e:dispcomov} expressed via \eqref{e:poly}. At the critical bottom drag and wave number the constant term with respect to the eigenvalue parameter $\lambda$ vanishes, 
and the linear coefficient is $-2\omega_\crit^2$. 
This is non-zero so that there is no double root and thus no generalized eigenfunction for $\calL_\crit$ and its adjoint. 
\end{proof}

Due to \cref{l:calLmuFredholm} we can split domain and range analogous to \Cref{s:bif} as $X=\ker(\calL_\crit)\oplus \setM$ where $\setM=\ker(\calL_\crit)^\perp$, 
and $Y= \ker(\calL_\crit^*)\oplus  \range(\calL_\crit)$, where $\range(\calL_\crit)^\perp = \ker(\calL_\crit^*)$ with respect to the inner product $\langle U,V \rangle_Y=\frac{1}{2\pi}\int_0^{2\pi}\langle U(\zz), V(\zz)\rangle_{\C^3}\dif \zz$, and the kernel of the adjoint operator $\ker(\calL_\crit^*)=\mathrm{span}\{\ee_j^*: j=0,\pm1\}$ as in the proof of \cref{l:calLmuFredholm}. 
With the inner product for $U,V\in X$ given by $\langle U,V\rangle_X= \langle  U_1,V_1\rangle_{H^4} + \langle  U_2, V_2\rangle_{H^4} + \langle  U_3, V_3\rangle_{H^1}$, we split $U = u+w$ with $u\in \ker (\calL_\crit)$ and $w\in \setM$ by the projection $\tilde P: X\to \ker(\calL_\crit)$, which can be written as
\[
u=\tilde P U:= \sum_{j=-1}^1\langle U, \ee_j^* \rangle_X\ee_j.
\]
The projection $P:Y\to \range(\calL_\crit)$ along $\ker(\calL_\crit^*)$ can be written as 
\[
P:= \Id - \sum_{j=-1}^1 \langle \cdot, \ee_j^*\rangle_Y \ee_j,
\]
and this gives the reduced problem of \eqref{e:igwprob}
\begin{equation}\label{e:PG2}
P G(u+ w;\mu)=0.
\end{equation}
As in \Cref{s:bif}, 
by the implicit function theorem, there is an open neighbourhood  of $(0,0,0)$ in $\ker(\calL_\crit)\times \setM\times \R^3$ of the form $N_0\times M_0\times (-\eps,\eps)^3$ and a unique function $W:N_0 \times(-\varepsilon,\varepsilon)^3\to  M_0$ such that $W(0,0)=0$ and $w=W(u;\mu)$ solves \eqref{e:PG2} for all $(u,\mu)\in N_0\times(-\varepsilon,\varepsilon)^3$. 
We will make use of the following estimate, which holds for possibly smaller $\eps$ and $N_0$: 
\begin{equation}\label{e:WestX}
W(u;\mu)= \calO((|\mu| + \|u\|_X)\|u\|_X).
\end{equation}
Since this estimate is in some sense standard and can be shown in the same way as in \Cref{s:nonsm}, we give the proof in \Cref{s:proofWestX} again.

In order to solve \eqref{e:igwprob} 
it remains to determine $(u,\mu)\in N_0\times(-\varepsilon,\varepsilon)^3$ such that 
\[
(\mathrm{Id}- P) G(u+W(u;\mu);\mu)=0,
\]
which is equivalent to the bifurcation equations 
\begin{equation}\label{e:bifigw} 
\langle G(u+W(u;\mu);\mu), \ee_j^*\rangle_Y=0,\quad j=0,\pm1.
\end{equation}
The next theorem gives the bifurcation near $\mu=(\alpha,s,\kappa)=0$. We recall $\lb = \lb_\crit -\alpha H_0$ with critical bottom drag $\lb_\crit$ as in \eqref{e:isocritC}, the critical wave length $k_\crit$ as in \eqref{e:isocritk} and the critical wave speed $\omega_\crit$ as in \eqref{e:omc}. 

\begin{theorem}[Bifurcation of inertia-gravity waves for $Q\neq 0$]\label{t:bifIGWQn0}
Let $Q\neq0$, $\alpha,\kappa\in \R$ sufficiently close to zero and $\k_\crit=(\kcx,\kcy)^\intercal$ arbitrary with $|\k_\crit|=k_\crit$. 
Consider $2\pi$-periodic steady travelling wave-type solutions $(\vv,\eta)^\intercal$ to \eqref{e:sw} with mean zero $\eta$ and $\zeta=(1+\kappa)^{-1}\zz-(\omega_\crit-s)t$. These waves are (up to spatial translations) in one-to-one correspondence with solutions $s, A_1$ near zero of 
\begin{subequations}\label{e:gwBifEqu}
\begin{align}
0 &= A_1\left(-s + \kappa \frac{f^2}{\omega_\crit} 
+ \calO(
|A_1|^2+|\mu|^2) 
\right), \label{e:gwBifEqua}\\
0&= A_1\left(\alpha 
- \frac{2 Q  k_\crit}{ H_0} \left( I_1 + \frac{k_\crit^2 g H_0}{2f^2+ k_\crit^2 g H_0} I_2\right)|A_1|+ 
\calO(
|A_1|^2+|\mu|^2) 
\right),\label{e:gwBifEqub}
\end{align}
\end{subequations}
with positive quantities
\begin{subequations}\label{e:gwCoefficient}
\begin{align}
I_1 &= \frac{1}{2\pi}\int_0^{2\pi}\sqrt{f^2 + k_\crit^2 g H_0 \cos(\zz)^2}\dif \zz,\\
I_2 &=  \frac{1}{2\pi}\int_0^{2\pi}\sqrt{f^2 + k_\crit^2 g H_0 \cos(\zz)^2}\cos(2\zz)\dif \zz.
\end{align}
\end{subequations}
With $\x$ as in \eqref{e:xphase}, these waves have the form
\[
\begin{pmatrix}\vv\\\eta\end{pmatrix}(t,{\x})=
2A_1\begin{pmatrix}
\omega_\crit \kcx\cos\zeta-f\kcy\sin\zeta\\
\omega_\crit \kcy\cos\zeta + f\kcx\sin\zeta\\
k_\crit^2H_0\cos\zeta
\end{pmatrix}
+\calO(|A_1|(
|A_1|+|\mu|)).
\]
\end{theorem}
Note that \eqref{e:gwBifEqua} specifies the deviation through $s$ from the travelling wave velocity and \eqref{e:gwBifEqub} the amplitude $|A_1|$. Since the coefficient of $A_1|A_1|$ in \eqref{e:gwBifEqub} is negative and the zero state is unstable for $\alpha>0$, the bifurcation is always supercritical. 
In contrast to the bifurcation of GE, a linear term in $\kappa$ appears, which balances the deviation $s$ in \eqref{e:gwBifEqua}. In \eqref{e:gwBifEqub} the dependence on $\kappa$ enters through the remainder term, which we do not further resolve here since this form suffices to infer supercriticality of the  bifurcation for sufficiently small $\kappa$.

\begin{proof}[Proof of \cref{t:bifIGWQn0}]
We derive 
the bifurcation equations \eqref{e:gwBifEqu} from \eqref{e:bifigw}.
For $j=0$ the equation is trivial, since $\ee_0^*=(0,0,1)^\intercal$ and with $U=(\vv,\eta)^\intercal$ the third component of $G$ can be written as the derivative
\[
-\partial_\zz (1+\kappa)  (  \k_\crit \cdot \vv )( H_0+ \eta )-(s-\omega_\crit)(1+\kappa)\partial_\zz\eta,
\]
whose integral vanishes on the space of periodic functions. 

Next we consider $j=\pm 1$. 
Since $\range(\calL_\crit)$ is orthogonal to $\ker(\calL_\crit^*)$, 
the linear part \eqref{e:calLmu} of $G$ can be replaced by 
\[
L_\mu:= \calL_\mu-\calL_\crit = \calO(|\mu|).
\]
With \eqref{e:igwprob} the non-trivial bifurcation equations \eqref{e:bifigw} can be written as
\begin{equation}\label{e:bifigw2}
\langle L_\mu u,\ee_j^*\rangle_Y - \langle B_Q (u),\ee_j^*\rangle_Y - \langle B (u;\mu),\ee_j^*\rangle_Y = \calR_j,
\quad j=\pm1,
\end{equation}
where $\calR_j$ is a remainder term, which includes $L_\mu W = \calO(|\mu| \|W\|_X)$ and 
$B(u+W;\mu)-B(u;\mu) = \calO(\|W\|_X(\|u\|_X + \|W\|_X))$, as well as
$B_Q (u+W)-B_Q(u) = (\|W\|_X(\|u\|_X + \|W\|_X))$. 
The latter follows from 
the reverse triangle inequality for the Euclidean norm of $\vv$ within $B_Q$ and we obtain 
\[
\calR_j = \calO(\|W\|_X(\|u\|_X + \|W\|_X+|\mu|)).
\]
Analogous to \Cref{s:bif}, $u$ in \eqref{e:bifigw2} can be written in terms of the amplitudes $A_j\in\C$,  $j=-1,0,1$, for the kernel modes $\ee_j$ of $\calL_\crit$ as
\[
u = A_0\ee_0 + A_1\ee_1 + A_{-1}\ee_{-1}.
\]
Here $A_{-1} = \overline{A_1}$, and since $\ee_0 = (0,0,1)^\intercal$ is real we have $A_0\in\R$; $\ee_1, \ee_{-1}$ are given in \eqref{e:eigenvector2}. As in \Cref{s:bif}, by translation symmetry we can assume that $A_1=A_{-1}\in \R$. Since \eqref{e:bifigw2} for $j=1$ is the complex conjugate of that for $j=-1$, it suffices to consider $j=1$.
Additionally, analogous to \Cref{s:bif} we can set $A_0=0$ without loss. Indeed, the orthogonality between $\ker(\calL_\crit)$ and $\setM$ is equivalent for the inner products $\langle\cdot,\cdot\rangle_X$ and $\langle\cdot,\cdot\rangle_Y$. 
Thus, $\langle W,\ee_0\rangle_Y=0$ and it follows that the third component of $W$ has always zero mean.
In particular, nonzero mean of the solutions $\eta$ comes from the constant contribution $A_0\neq0$ only.
With \eqref{e:WestX} and the form of $u$, 
this implies $W= \calO(|A_1|(|A_1|+|\mu|))$. 
Hence, 
\begin{equation}\label{e:Rest}
\calR_1 = \calO(|A_1|(
|A_1|^2+|\mu|^2)).
\end{equation}
We also write $u=(\tvv,\teta)^\intercal$ in the following so that 
\begin{align*}
\tvv &= A_1 
\begin{pmatrix}
\omega_\crit \kcx + \rmi f\kcy \\
\omega_\crit \kcy - \rmi f\kcx 
\end{pmatrix} 
\rme^{\rmi \zz}
+ A_{1} 
\begin{pmatrix}
\omega_\crit \kcx - \rmi f\kcy \\
\omega_\crit \kcy + \rmi f\kcx 
\end{pmatrix}
\rme^{-\rmi \zz},
\\
\teta &= 
A_1 k_\crit^2 H_0(\rme^{\rmi\zz} + \rme^{-\rmi\zz}) . 
\end{align*}
Next, we derive the leading order part in \eqref{e:bifigw2}.
We first consider the term involving the smooth quadratic terms of \eqref{e:sw} given by
\[
B(u;\mu) = 
(1+\kappa)(\tvv\cdot \k_\crit)\begin{pmatrix} \partial_\zz \tvv\\ \partial_\zz \teta\end{pmatrix}
+(1+\kappa)\begin{pmatrix}0\\ \teta \, \k_\crit\cdot \partial_\zz \tvv\end{pmatrix} 
- \frac{C_\crit-\alpha H_0}{H_0^2}\begin{pmatrix}\teta\tvv\\ 0\end{pmatrix}. 
\]
Here $\nabla$ from \eqref{e:sw} is replaced by $(1+\kappa)\k_\crit\partial_\zz$ due to the choice of variables. It is straightforward that all the terms are orthogonal to $\ker(\calL_\crit^*)$ due to the quadratic combinations of $\rme^{\pm \rmi\zz}$, so
\begin{equation}\label{e:Best}
\langle B(u;\mu),\ee_1^* \rangle = 0.
\end{equation}
This simplifies \eqref{e:bifigw2} for $j=1$, and we next consider the term in \eqref{e:bifigw2}, which involves the non-smooth quadratic term
\[
\langle B_Q (u),\ee_1^*\rangle_Y =  \langle \begin{pmatrix}\frac{Q}{H_0}|\tvv|\tvv\\0\end{pmatrix},\ee_1^*\rangle_Y,
\]
where
\[
|\tvv| = 2|A_1| k_\crit \sqrt{\omega_\crit^2  \cos(\zz)^2 + f^2\sin(\zz)^2} = 2|A_1| k_\crit \sqrt{f^2 + k_\crit^2 g H_0 \cos(\zz)^2}.
\]
We thus compute that 
\begin{equation}\label{e:BQest}
\langle B_Q (u),\ee_1^*\rangle_Y = \frac{2k_\crit^3 Q}{H_0 m} |A_1| A_1
\big( (\omega_\crit^2 + f^2) I_1 + (\omega_\crit^2 - f^2) I_2\big), 
\end{equation}
where the decisive coefficients $I_1, I_2$, are given in \eqref{e:gwCoefficient}, which is characterized by the non-smooth nonlinearity.
In contrast to the analogous integrals in \Cref{s:nonsm}, here we cannot find explicit expressions.  
However, for the qualitative result it suffices to determine the signs. Clearly, $I_1>0$ and to show that $I_2>0$ we abbreviate
$f(\zz):=\sqrt{f^2 + k_\crit^2 g H_0 \cos(\zz)^2}$ and compute 
\[
\pi I_2  = \int_0^{\pi/4} \underbrace{\Bigl(f(\zz)-f(\zz+\frac\pi 2)\Bigr)}_{>0 \hspace{2mm} (*)}\underbrace{\cos(2\zz)}_{>0}\dif \zz 
+ \int_{\pi/4}^{\pi/2} \underbrace{\Bigl(f(\zz)-f(\zz+\frac\pi 2)\Bigr)}_{<0 \hspace{2mm}(**)}\underbrace{\cos(2\zz)}_{<0}\dif \zz >0,
\]
where $(*)$ holds since $\cos(\zz)^2 > 1/2 > \cos(\zz+\pi/2)^2$ for $\zz\in[0,\pi/4)$, and $(**)$ holds since $\cos(\zz)^2 < 1/2 < \cos(\zz+\pi/2)^2$ for $\zz\in(\pi/4,\pi/2]$. 

Concerning the linear part of \eqref{e:bifigw2}, we first note that due to $\calL_\mu$ in \eqref{e:calLmu}
\[
L_\mu =  \alpha\,\mathrm{diag}(1,1,0) -(1+\kappa) s\,\mathrm{diag}(1,1,1) \partial_\zz + \kappa\, \calK. \label{e:Lmu}
\]
We next compute $\langle L_\mu  u, \ee_1^*\rangle_Y$. The first term yields 
\[
\langle \mathrm{diag}(1,1,0) u, \ee_1^*\rangle_Y = \langle \begin{pmatrix}\tvv\\ 0\end{pmatrix}, \ee_1^*\rangle_Y = A_1 k_\crit^2 \frac{\omega_\crit^2 + f^2}{m}.
\]
Further, for the comoving frame term with factor $s$ we obtain
\begin{equation}\label{e:uzej}
\langle \partial_\zz u, \ee_1^*\rangle_Y = \rmi  A_1 \frac{2k_\crit^2\omega_\crit^2}{m}.
\end{equation}
Regarding $\calK$ as in \eqref{e:calLkappa}, we are only interested in the leading order terms and thus consider $\calK|_{\kappa=0}$. Here $\calS|_{\kappa=0}$ from \eqref{e:calLkappa} simplifies and for the projection onto 
$\ee_1^*$ 
we note that $(\partial_\zz^2+1)\rme^{\rmi j\zz}=0$, so that the relevant diagonal part of $\calK|_{\kappa=0}$ is $\mathrm{diag}(1,1,1) \omega_\crit \partial_\zz$, for which we can use \eqref{e:uzej}. The remaining terms of the third column in $\calK$ give 
\[
- g\langle (\k_\crit,0)^\intercal \partial_\zz \tilde\eta ,\ee_1^*\rangle_Y = - \rmi A_1 g H_0 k_\crit^2 \langle (\k_\crit,0)^\intercal, \E_1^* \rangle_Y
= -\rmi A_1 g H_0 k_\crit^4 \frac{\omega_\crit}{m},
\]
which is the same as the term created by the third row, so it is doubled.
Gathering terms,
\begin{equation}\label{e:Lmuest}
\langle L_\mu  u, \ee_1^*\rangle_Y = \left( \alpha (\omega_\crit^2+ f^2) 
+ 2 \rmi (-s\omega_\crit +\kappa f^2) \omega_\crit\right) \frac{k_\crit^2}{m}  A_1
+\calO(|\mu|^2 |A_1|),
\end{equation}
where we have used $\omega_\crit^2  - g H_0 k_\crit^2=f^2$. 
Concerning \eqref{e:bifigw2} we observe that the order of the remainder in \eqref{e:Rest} includes the error term in \eqref{e:Lmuest}. 
Using the results \eqref{e:Best}, \eqref{e:BQest} and \eqref{e:Lmuest}, equation \eqref{e:bifigw2} for $j=1$ becomes
\begin{equation}\label{e:igweq1}
\begin{split}
&\left( \alpha (\omega_\crit^2+ f^2) + 2 \rmi(-s\omega_\crit +\kappa f^2) \omega_\crit\right) \frac{k_\crit^2}{m} A_1\\
& - \frac{2k_\crit^3 Q}{H_0 m} |A_1|A_1 \big((\omega_\crit^2 + f^2) I_1 + (\omega_\crit^2 - f^2) I_2\big)
=  \calO(|A_1|(
|A_1|^2+|\mu|^2)).
\end{split}
\end{equation}
Upon dividing by 
$k_\crit^2/m$,  
\eqref{e:igweq1} can be split into real and imaginary parts as 
\begin{align*}
0&= A_1\left(2 (-s\omega_\crit +\kappa f^2) \omega_\crit 
+ \calO(|A_1|^2+|\mu|^2)\right) \\
0&=  A_1\left(\alpha (\omega_\crit^2+ f^2) 
- \frac{2 Qk_\crit}{H_0} |A_1| \big((\omega_\crit^2 + f^2) I_1 + (\omega_\crit^2 - f^2) I_2\big)+ \calO(|A_1|^2+|\mu|^2)\right).
\end{align*}
Using $\omega_\crit^2 - f^2 = k_\crit^2 g H_0$, $\omega_\crit^2 +f^2 = 2f^2+ k_\crit^2 g H_0$ 
and rearranging terms, we obtain the bifurcation equations \eqref{e:gwBifEqu}.
\end{proof}

\section{Explicit nonlinear flows with arbitrary amplitudes}\label{s:explicit}

The shallow water equations \eqref{e:sw} admit explicit solutions with linear dynamics in the nonlinear setting, cf.\ e.g.\ \cite{PR2020}. These come in linear spaces and for $C=Q=0$, backscatter can create solutions of this kind that grow unboundedly and exponentially, which is one indication of undesired concentration of energy due to backscatter \cite{PRY22}.  
In this section we study the impact of bottom drag on these solutions, in particular on the robustness of the degenerate growth.  
Since such  explicit solutions can only be found for smooth bottom drag, we restrict to $Q=0$ in this section. 
In contrast to the solutions determined in \Cref{s:bif}, we find explicit solutions that have arbitrary amplitude and can depend on time. 

Our starting point is the specific plane wave-type ansatz with wave vector $\k=(\kx,\ky)^\intercal$ of the form 
\begin{equation}\label{e:sol}
\vv = \ampA \rme^{\beta t}\cos(\k\cdot\x + \tau)\k^{\perp}, \quad   \eta = 0,
\end{equation}
where $\k^\perp = (-\ky,\kx)^\intercal$. 
Compared with \cite{PRY22} $\eta$ is constant due to its presence in the denominator of the bottom drag term \eqref{e:drag}. 
This is also a special case of \eqref{e:planeansatz}, although the solutions will not be geostrophically balanced.
Substituting \eqref{e:sol} into \eqref{e:sw} gives the existence conditions 
\begin{subequations}\label{e:cond}
\begin{align}
\beta\ &= \ (b_1-d_1|\k|^2)k_y^2+(b_2-d_2|\k|^2)k_x^2 -  \lb/H_0, \label{e:conda} \\
0\ &=\ p(\k):=\kx\ky \left((d_1-d_2)|\k|^2+b_2-b_1\right) + f. \label{e:condb}
\end{align}
\end{subequations}
Since the conditions \eqref{e:cond} are independent of $\ampA$, each solution yields a `vertical branch' of explicit flows \eqref{e:sol}, parameterized by $\ampA\in\R$. In the steady case $\beta=0$ this can be viewed as bifurcating from the trivial state for fixed parameters, cf.\ \cref{f:ampalp} (dashed). 
The wave vectors $\k$, for which explicit flows of the form \eqref{e:sol} exist, lie on a union of curves determined by the condition \eqref{e:condb}. Their growth rates $\beta=\beta(\k)$ are given by the dispersion relation \eqref{e:conda}, so that 
the wave vectors of steady flows lie on the intersections with curves determined by \eqref{e:conda} for $\beta=0$. 
For $\sigma\in\{0,\pm1\}$ let us define 
\[
\Gamma_\sigma:=\{\k\in\R^2: p(\k) = 0,\, \sgn(\beta(\k))=\sigma\}.
\]
Steady flows of the form \eqref{e:sol} exist if and only if $\Gamma_0\neq \emptyset$ (the intersection of black and blue curves in \cref{f:growth}), exponentially and unboundedly growing such flows if and only if $\Gamma_1\neq\emptyset$ (the intersection of red regions and blue curves in \cref{f:growth}), and exponentially decaying ones have wave vectors $\k\in\Gamma_{-1}$ (the intersection of white regions and blue curves in \cref{f:growth}). The latter is non-empty except in the isotropic case with rotation, where solutions \eqref{e:sol} do not exist.

\begin{figure}[t!]
\centering
\subfigure[$\lb=0.11$]{\includegraphics[trim = 5.5cm 8.5cm 6cm 8.5cm, clip, width=0.24\linewidth]{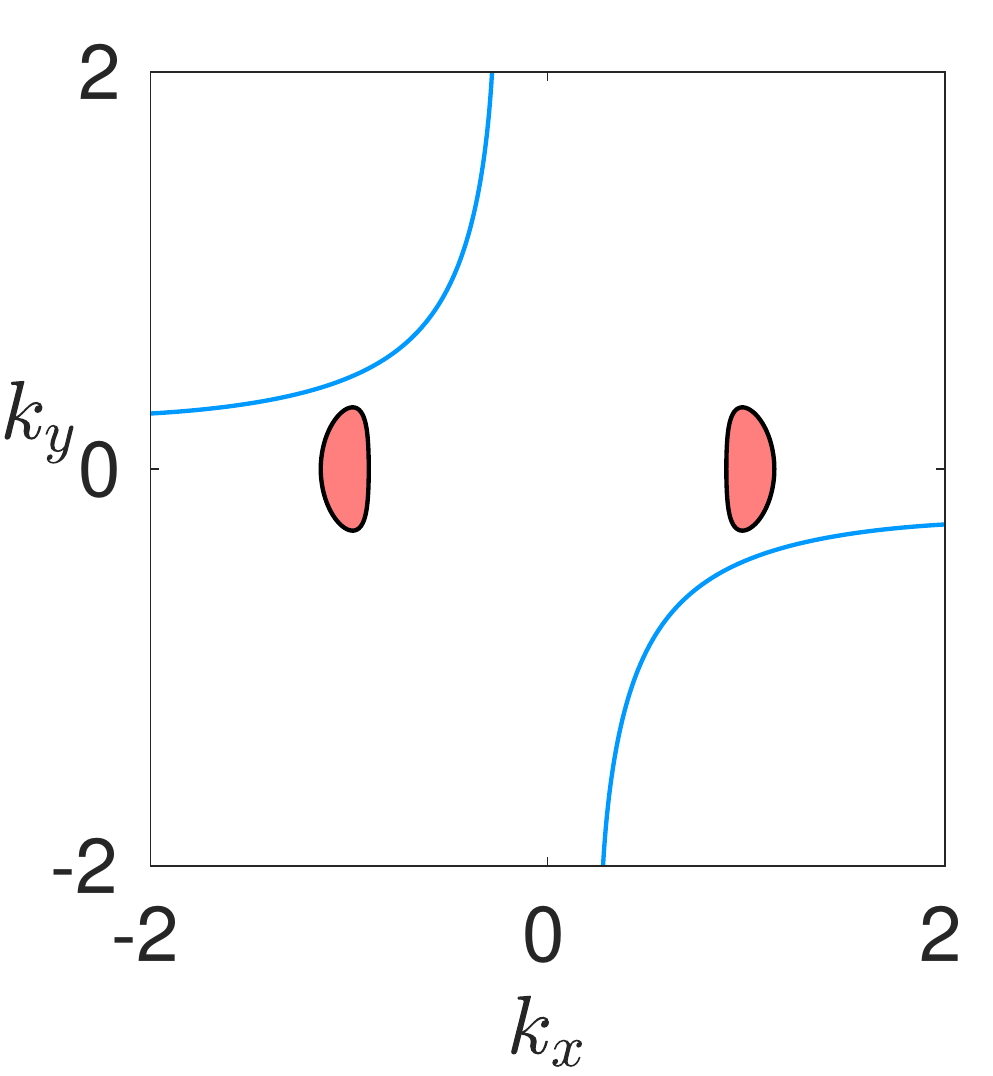}\label{f:anisoC0p11}}
\hfil
\subfigure[$\lb=0.08$]{\includegraphics[trim = 5.5cm 8.5cm 6cm 8.5cm, clip, width=0.24\linewidth]{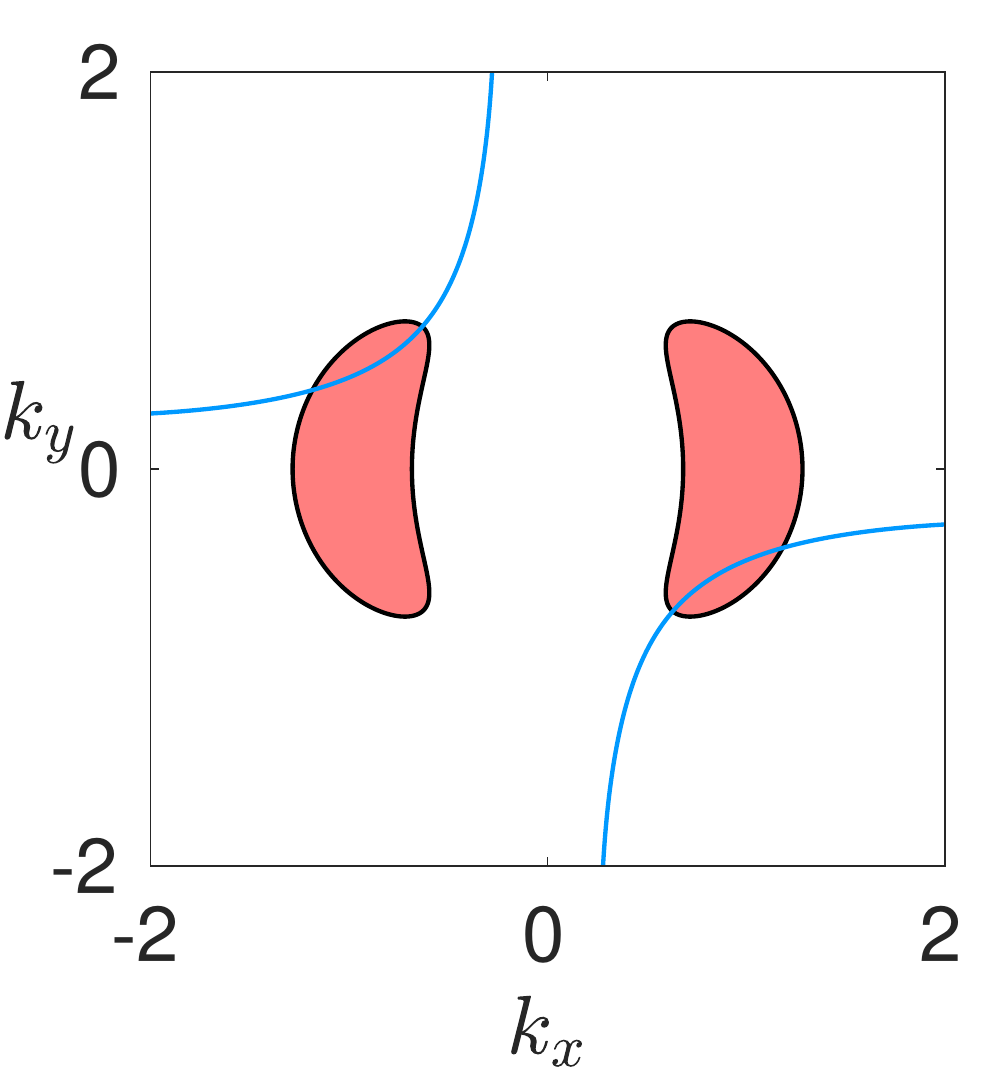}\label{f:anisoC0p08}}
\hfil
\subfigure[$\lb=0.05$]{\includegraphics[trim = 5.5cm 8.5cm 6cm 8.5cm, clip, width=0.24\linewidth]{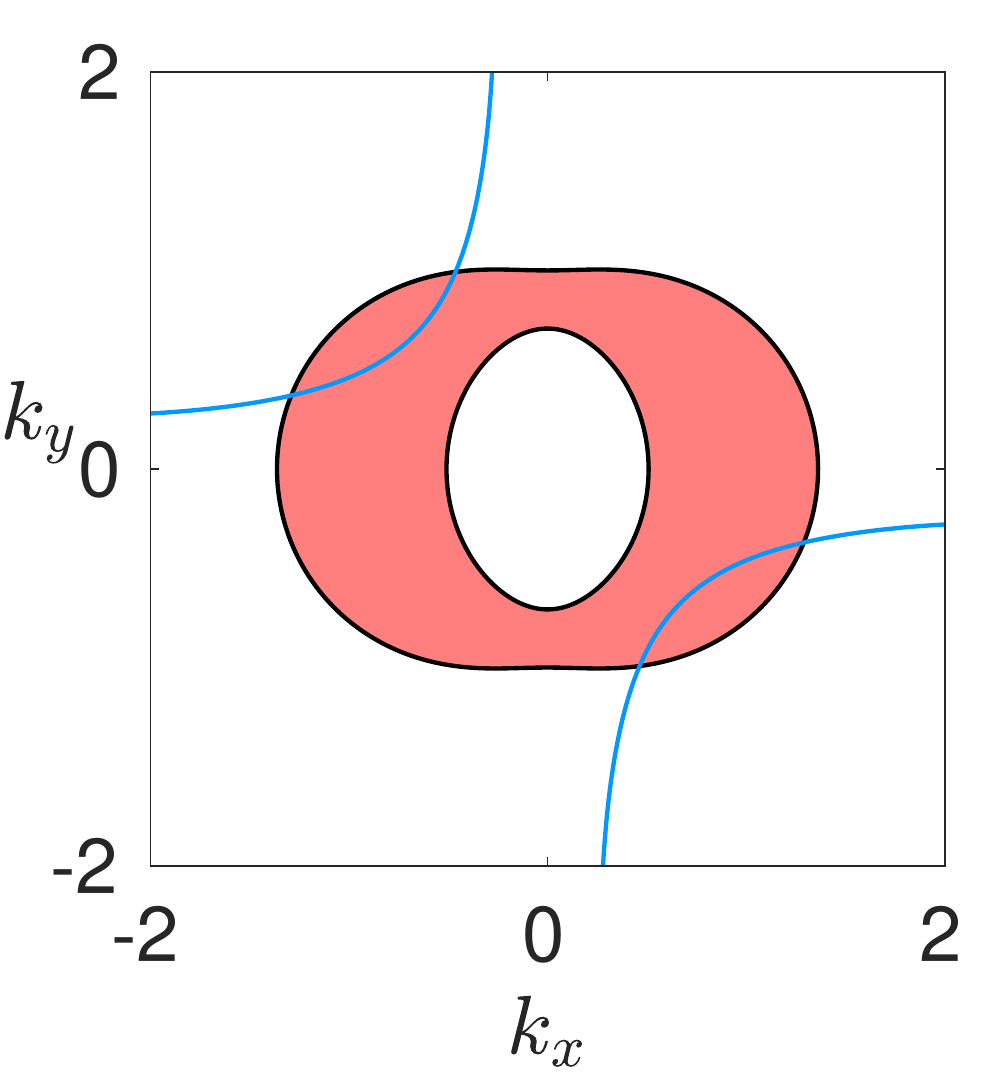}\label{f:anisoC0p05}}
\hfil
\subfigure[$\lb=0$]{\includegraphics[trim = 5.5cm 8.5cm 6cm 8.5cm, clip, width=0.24\linewidth]{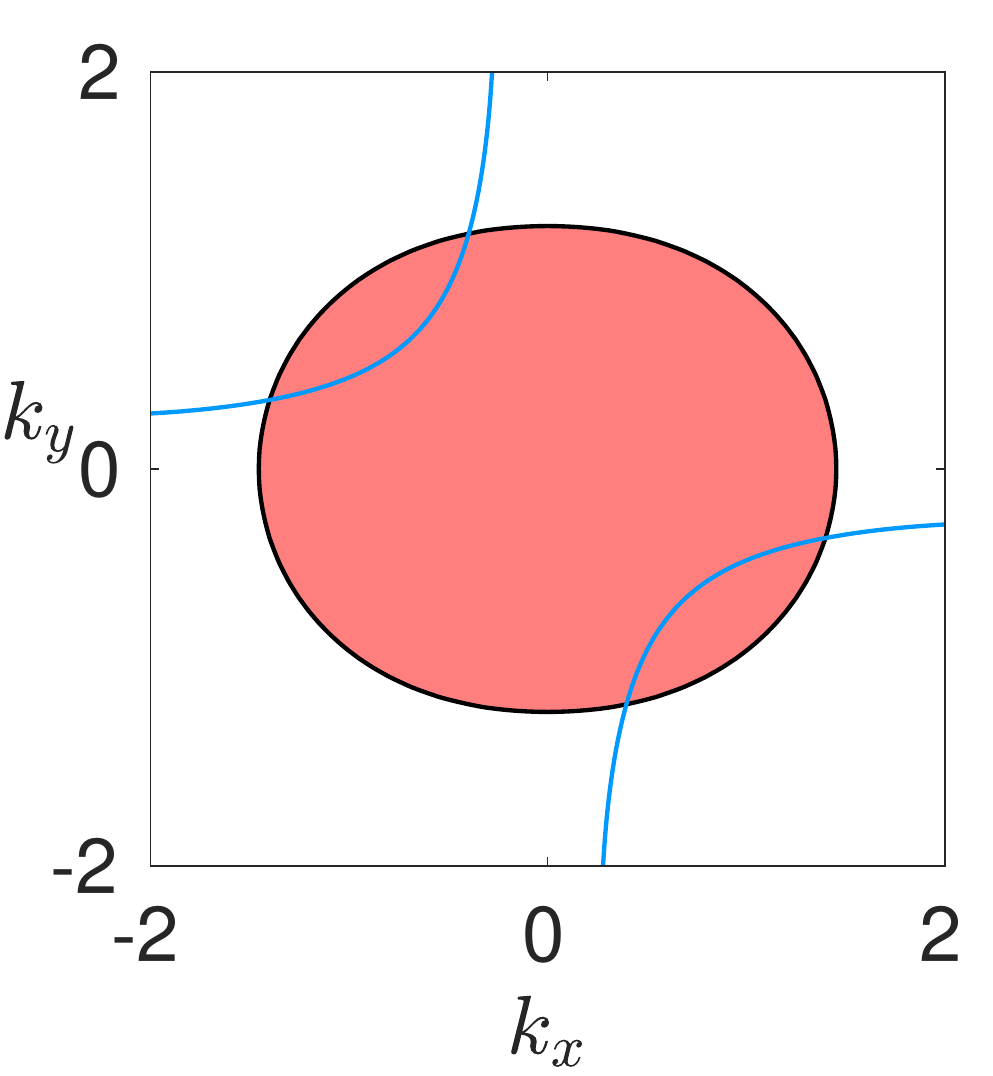}\label{f:anisoC0}}
\caption{Sample of the loci of explicit flows \eqref{e:sol} in the wave vector plane. The parameters are $d_1=1$, $d_2=1.04$, $b_1=1.5$, $b_2=2.2$, $f=0.3$, $g=9.8$, $H_0=0.1$, $Q=0$ so that $\lb_\crit\approx0.116$. Red: $\beta>0$; white: $\beta<0$; black: $\beta=0$; blue: the wave vectors satisfying \eqref{e:condb}. The intersections of blue curves with black curves are steady flows, those within red regions unboundedly growing flows; $\beta = 0$ at the origin is excluded in each figure.}
\label{f:growth}
\end{figure}
 
The explicit flows \eqref{e:sol} with $\k\in\Gamma_1$ grow unboundedly and exponentially in the nonlinear systems, and thus form a linear and unbounded part of the unstable manifold of the trivial state due to the arbitrary choice of the amplitude $\ampA$; analogous for $\k\in\Gamma_{-1}$ and the stable manifold.
Since explicit flows \eqref{e:sol} with \eqref{e:cond} also satisfy \eqref{e:sw} without the nonlinear terms, they are selected real eigenmodes of the linearization in the trivial state. 
In particular, the existence of an explicit flow \eqref{e:sol} requires a solution to the eigenvalue problem of \eqref{e:linop} studied in \Cref{s:spec} with $\lambda=\beta$. 
In fact, for $\beta$ from \eqref{e:conda}, the term $-|\k|\beta$ has the same sign as $a_3$ in the dispersion relation \eqref{e:poly}. Thus, for any $\k\neq0$ with $\beta>0$ it follows $a_3<0$ and a positive real eigenvalue exists for these $\k$. This means that regions with positive $\beta$ in the wave vector plane (the red regions in \cref{f:growth}) give subspaces of unstable real eigenmodes of trivial states. We refer to \cite{PRY22} for a broader discussion in the case $\lb=0$.  

The bottom drag parameter $C$ only affects the growth rate $\beta$, and has a monotonically stabilizing effect on the explicit flows. The previous  paragraph means that steady flows \eqref{e:sol} cannot exist if there are no steady eigenmodes. 
Hence, for  $\lb>\lb_\crit$ from \eqref{e:isocritC} or \eqref{e:anisocritC}  we have $\Gamma_0=\Gamma_1=\emptyset$.
At $\lb=0$, for sufficiently small $\kx$ and $\ky$ the leading order term of $\beta$ is given by the positive quantity $b_1k_y^2+b_2k_x^2$, which means any explicit flow with $\k\approx 0$ is growing. 
For $\lb>0$, the leading order term of $\beta$ for small $\kx$ and $\ky$ is given by $-C/H_0$, and for large $(\kx,\ky)$ by $-(d_1k_y^2+d_2k_x^2)|\k|^2$. Both quantities are negative, thus the explicit flows are decaying for $\k$ close to and sufficiently far from the origin in the wave vector plane. 
We plot examples in \cref{f:growth} that illustrate these situations.

\begin{remark}\label{r:stabexplicit}
We briefly comment on linear stability properties of the explicit flows \eqref{e:sol} for $A\approx 0$ and refer to \cite{PRY22} for further details. 
For $A\approx 0$, a part of the spectrum of the linearization in the explicit flows is close to the spectrum of the underlying trivial flow. Since this is unstable precisely for $\lb<\lb_\crit$, the explicit flows with $A\approx 0$ are also unstable. From the results in \cite{PRY22} we expect that the explicit flows are unstable for all $A$, although we do not have a proof. For $|A|\gg 1$ it might be possible to exploit the scaling argument in  \cite{PRY22}.
\end{remark}

\begin{remark}
We briefly consider fast rotation $|f|\gg1$, which requires $|\k|\gg 1$ to solve \eqref{e:condb}. Since the right-hand side of \eqref{e:conda} is negative for $|\k|\gg 1$ (and fixed $\lb$), there are no steady explicit flows of the form \eqref{e:sol} in this regime. However, as noted in \cref{r:largef}, $|f|\to\infty$ with $\phi=\tf\tilde\phi$ gives the formal limit $\psi=\partial_\xi\tilde\phi$ and $0=(d k^2\partial^2_\xi+b)\partial_\xi^2 \psi$. These have the steady solutions of the different form $\psi(\xi) = A\cos(\xi)$, $\tilde\phi(\xi)= A\sin(\xi)$ with $k=\sqrt{b/d}$. 
\end{remark}

\subsection{Number of steady explicit flows}
In order to gain insight into the structure of the explicit flows, we next investigate the number of different steady explicit flows \eqref{e:sol} that occur for fixed parameters. 
This is equivalent to the different intersections of the curves defined by \eqref{e:conda} and \eqref{e:condb} in the wave vector plane with $\beta=0$.
We order this analysis by the types of isotropy of backscatter terms and consider both, the rotational and non-rotational case.

\medskip
\paragraph{Isotropic backscatter}
For $d:=d_1=d_2$, $b:=b_1=b_2$, explicit flows \eqref{e:sol} exist if and only if $f=0$ due to \eqref{e:condb}. The condition \eqref{e:conda} reduces to $\beta = \beta(K) = (b-dK)K - C/H_0$, with $K:=|\k|^2$, and the solutions correspond to those of \eqref{e:scalef0}, which is linear in this case. Steady explicit flows \eqref{e:sol} exist (i.e.\ $\Gamma_0\neq \emptyset$) for $0<\lb\leq \lb_\crit$, but not for $\lb>\lb_\crit$; for $\lb=\lb_\crit$ the parabola $\beta(K)$ has a double root at $K=k_\crit^2$ from \eqref{e:isocritk} and the wave vectors of steady explicit flows form a circle with radius $|\k|=k_\crit$. Hence, in terms of decreasing $\lb$, the primary bifurcation of the steady explicit flows  occurs at $\lb=\lb_\crit$ with a vertical branch as in \cref{f:ampalp} (dashed), 
parameterized by the amplitude $\ampA$ analogous to \eqref{e:solQ0} for $f=0$. 
For $0<\lb<\lb_\crit$ the negative parabola $\beta(K)$ has two different roots and is monotonically shifted upwards upon decreasing $\lb$, so that the wave vectors of the steady explicit flows \eqref{e:sol} form two concentric circles as in \cref{f:iso}. Their radii depend monotonically on $\lb$ and the region in between defines $\Gamma_1$, i.e.\ unboundedly and exponentially growing flows. 

\begin{figure}[t!]
\centering
\subfigure[$f=0$, $\lb=0.08$]{\includegraphics[trim = 5.5cm 8cm 5.5cm 7.5cm, clip, width=0.24\linewidth]{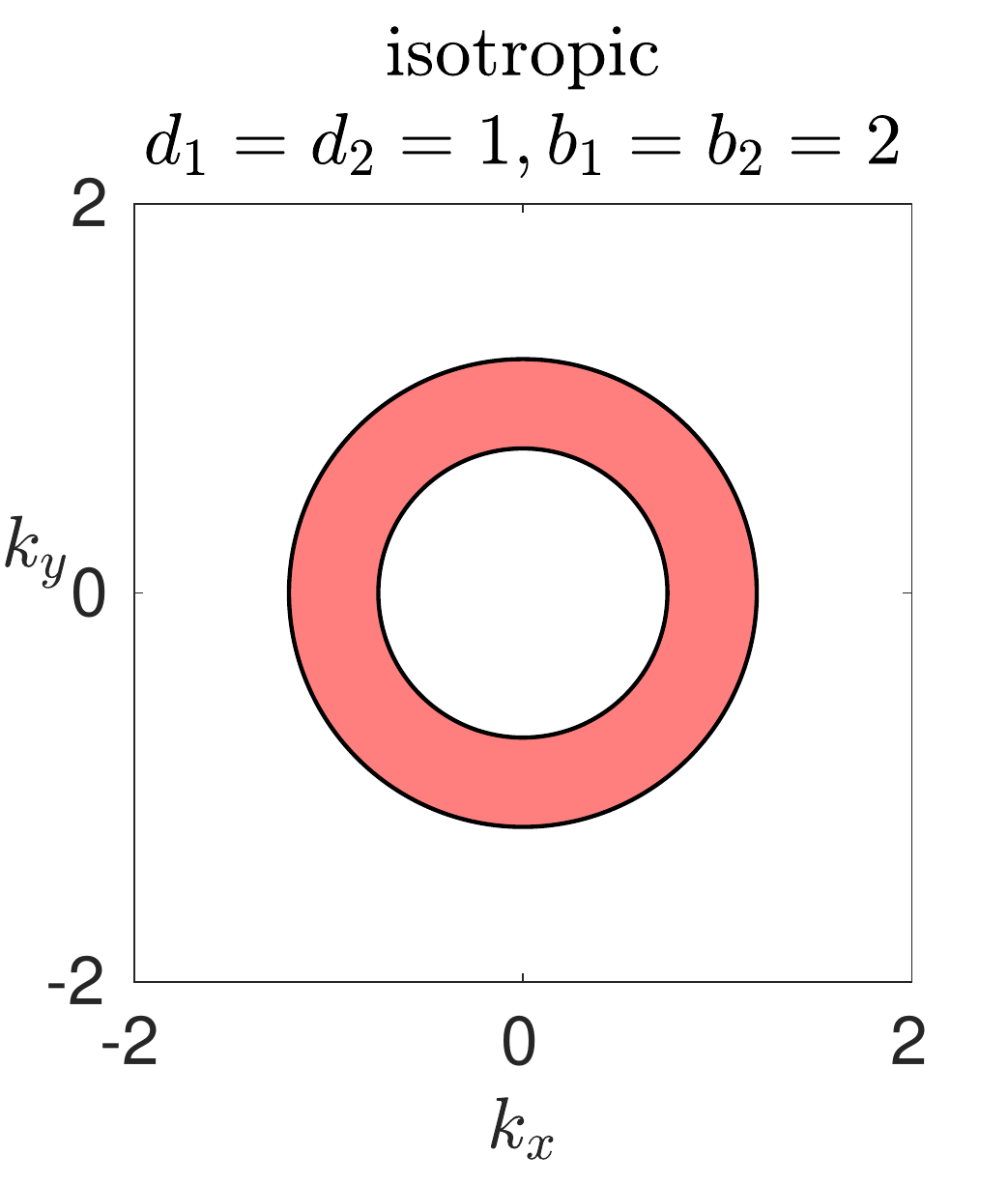}\label{f:iso}}
\hfil
\subfigure[$f=0$, $\lb=0.04$]{\includegraphics[trim = 5.5cm 8cm 5.5cm 7.5cm, clip, width=0.24\linewidth]{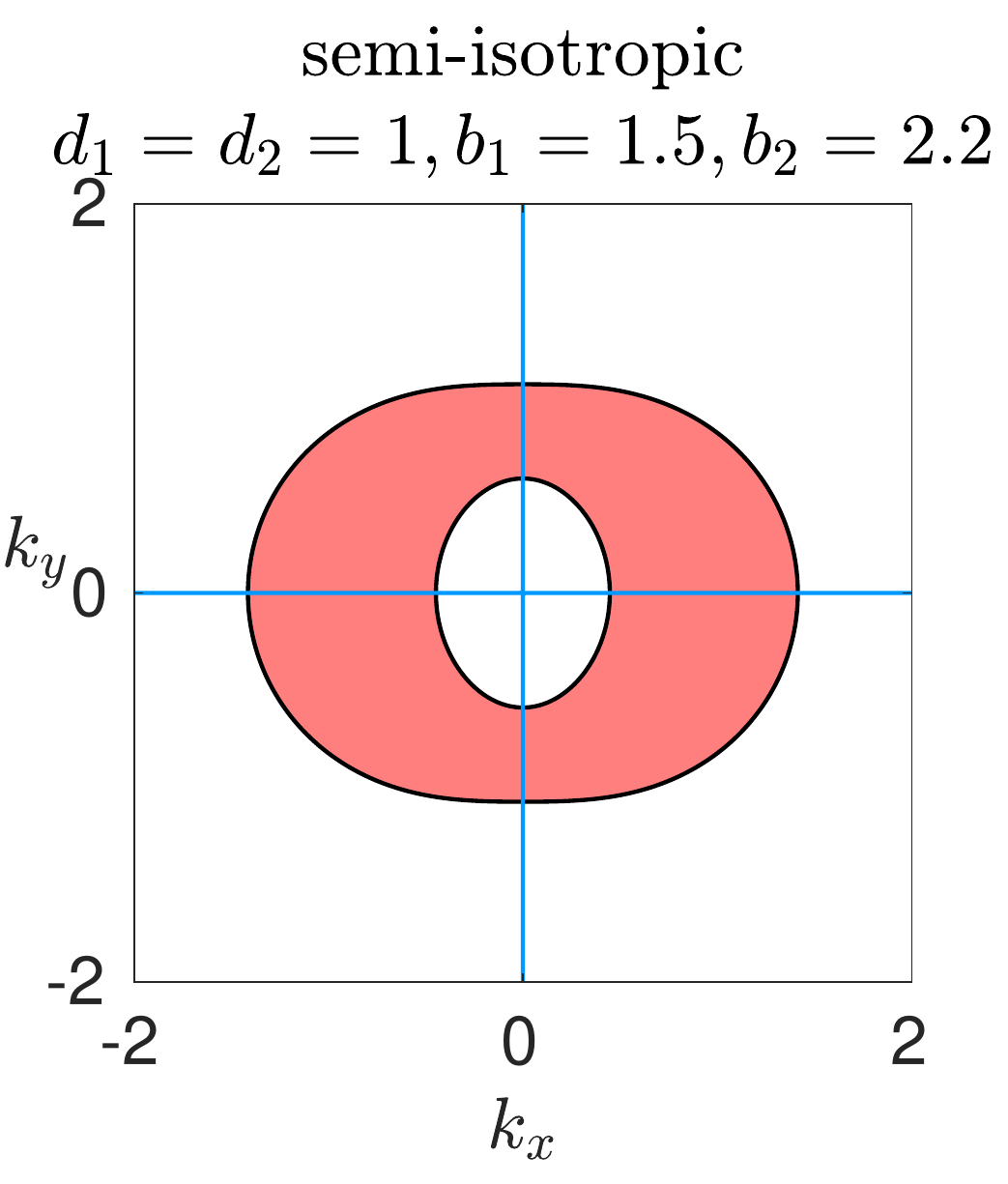}\label{f:semi-iso1}}
\hfil
\subfigure[$f=0.05$, $\lb=0.04$]{\includegraphics[trim = 5.5cm 8cm 5.5cm 7.5cm, clip, width=0.24\linewidth]{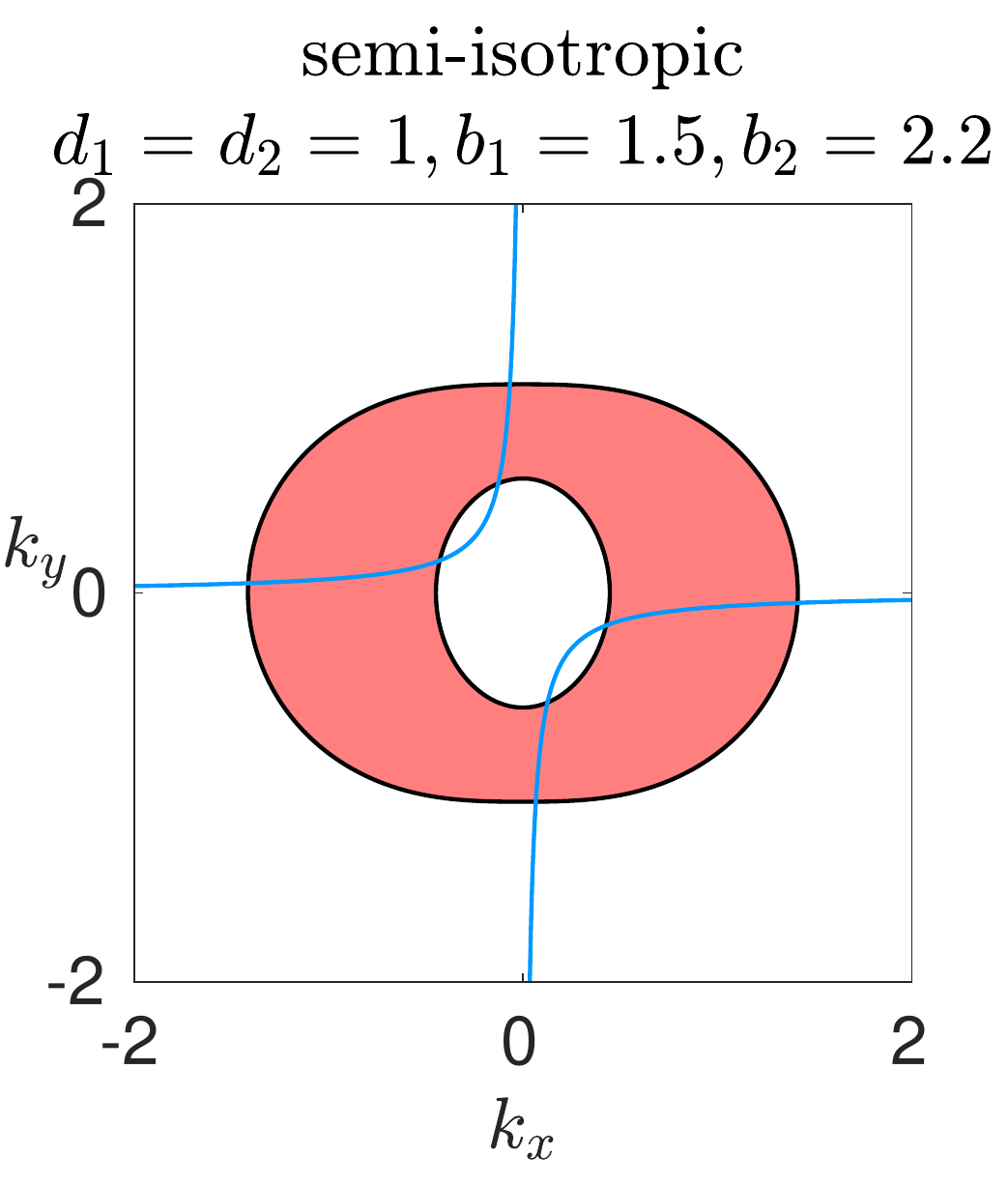}\label{f:semi-iso2}}
\hfil
\subfigure[$f=0.2$, $\lb=0.1$]{\includegraphics[trim = 5.5cm 8cm 5.5cm 7.5cm, clip, width=0.24\linewidth]{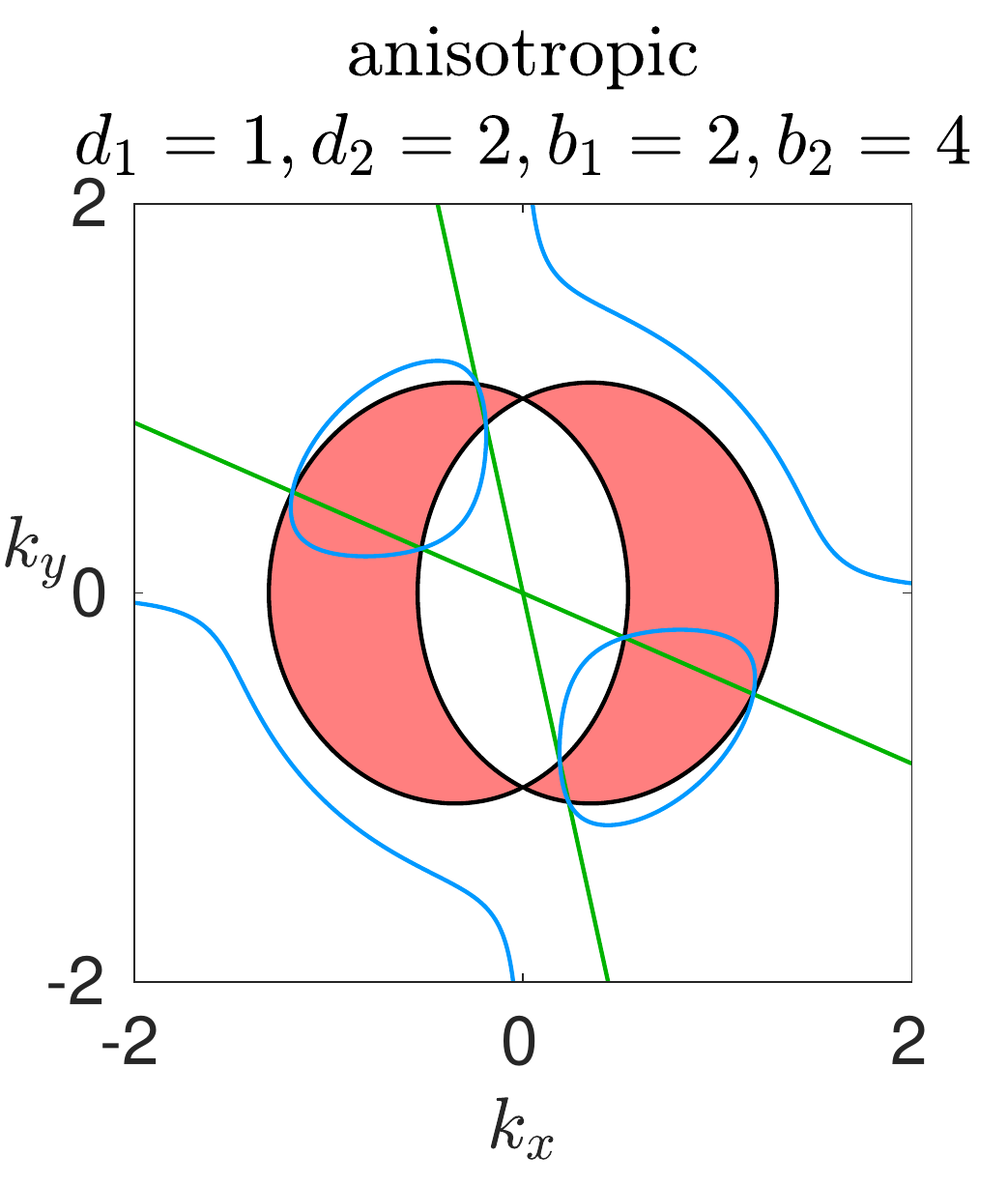}\label{f:fourroots}}
\caption{
Sample of the loci of explicit flows \eqref{e:sol} in the wave vector plane; (a) existence for any wave vector, (b-d) existence on blue curves only. 
Red: $\beta>0$; white: $\beta<0$; black: $\beta=0$; blue: wave vectors satisfying \eqref{e:condb}. Common parameters: $g=9.8$, $H_0=0.1$, $Q=0$. In (d), there are two intersections of blue and black curves on each ray (green) in the second and fourth quadrant.
}
\label{f:intersections}
\end{figure}

\begin{remark}[Relation to bifurcation analysis]\label{r:isobifflows}
With isotropic backscatter and $f=Q=0$, the bifurcating  equilibria 
of \cref{t:bifQ0} coincide with a part of the unbounded branches of steady explicit flows \eqref{e:sol}. 
This can be seen by relating \eqref{e:special}, \eqref{e:solQ0}, \eqref{e:bifGlin} with \eqref{e:conda}. 
We also recall that $w=W(u,\mu)\equiv 0$ for $f=Q=0$ in \cref{t:bifQ0}. 
Whereas \eqref{e:sol} can have arbitrary $A\in\R$ and $\alpha\in [0,\lb_\crit /H_0]$, the bifurcation analysis requires $|\alpha|$ and $|A_1|$ to be sufficiently small in \cref{t:bifQ0} due to the use of the implicit function theorem. \newline
While for $f\neq0$ steady explicit flows of the form \eqref{e:sol} do not exist, 
 the bifurcation analysis gives plane wave-type solutions 
\eqref{e:special} with \eqref{e:solQ0} for wave vectors in an open neighbourhood of the circle with radius $k_\crit$. \end{remark}

\medskip
\paragraph{Semi-isotropic backscatter}
For $b_1\neq b_2$, $d:=d_1=d_2$ and $f=0$, \eqref{e:condb} requires $k_x=0$ or $k_y=0$ so that steady solutions have wave vectors on the axes in the wave vector plane.
Due to the remaining fourth order polynomial \eqref{e:conda} with $\beta=0$ and even exponents, there are up to two different steady solutions with axis aligned wave vectors, i.e.\ up to four in total. 
Moreover, the interval between two such wave vectors on the same axis belongs to wave vectors of exponentially growing explicit flows. An example of the occurrence of steady, growing and decaying explicit flows \eqref{e:sol} for this situation is depicted in \cref{f:semi-iso1}. \newline
In the rotational case $f\neq0$ the equations \eqref{e:cond} reduce to
\begin{align*}
0&=-d|\k|^4+b_1k_y^2+b_2k_x^2-C/H_0,\\
0&=(b_2-b_1)k_xk_y+f,
\end{align*}
 where the second equation implies $k_x,k_y\neq0$ and can be reformulated to $k_y=f/\bigl((b_1-b_2)k_x\bigr)$.
Inserting this form of $k_y$ into the first equation and defining $K_x:=k_x^2>0$ leads to
\[
0=-dK_x^4+b_2K_x^3-\left(\frac{2df^2}{(b_1-b_2)^2}+\frac{C}{H_0}\right)K_x^2+\frac{b_1f^2}{(b_1-b_2)^2}K_x-\frac{df^4}{(b_1-b_2)^4}.
\]
This is a fourth order polynomial with always four sign changes of the coefficients, so that by Descartes' rule of signs it has zero, two or four different positive roots. 
The wave vectors on the curve defined by $k_y=f/\bigl((b_1-b_2)k_x\bigr)$ with $\beta(\k)>0$ belong to exponentially growing explicit flows. We plot an example of this situation in \cref{f:semi-iso2}.

\medskip
\paragraph{Anisotropic backscatter}
For $b_1\neq b_2$, $d_1\neq d_2$ and $f=0$, both $\kx=0$ and $\ky=0$ solve \eqref{e:condb}. This leads to the same situation as for the semi-isotropic case with $f=0$. 
Additionally to these solutions, if $(b_1-b_2)/(d_1-d_2)>0$, wave vectors on the circle with radius $|\k| = k_0:= \sqrt{(b_1-b_2)/(d_1-d_2)}$ also solve \eqref{e:condb}. 
The remaining condition \eqref{e:conda} with $\beta=0$ and $|\k|=k_0$ implies $\lb = \lb_0:=\frac{(b_1-b_2)(b_2d_1-b_1d_2)H_0}{(d_1-d_2)^2}$. The 
case $\lb_0=0$ for $b_2d_1=b_1d_2$ also occurs in \cite{PRY22}. 
We claim that $\lb_0<\lb_\crit$. For $\lb_\crit$ given by \eqref{e:critx} (the case \eqref{e:crity} is analogous) we compute $\lb_\crit-\lb_0 = \frac{(b_2 (d_1 + d_2)-2 b_1 d_2)^2 H_0}{4d_2 (d_1 - d_2)^2}$, which is strictly positive 
since 
\[
b_2 (d_1 + d_2)-2 b_1 d_2 
\geq (b_1^2d_2 + b_2^2d_2-2 b_1 b_2 d_2)/b_2 = d_2(b_1-b_2)^2/b_2>0,
\]
where the first inequality is derived from $b_2^2d_1\geq b_1^2d_2$ in \eqref{e:critx}. 
 If $(b_1-b_2)/(d_1-d_2)>0$ and $(b_1-b_2)(b_2d_1-b_1d_2)\geq 0$, 
then $0\leq \lb_0<\lb_\crit$, and there exists a $k_0>0$ such that for $\lb=\lb_0$ the steady explicit flows \eqref{e:sol} exist on the entire circle with radius $k_0$ in the wave vector plane. 
If $(b_1-b_2)(b_2d_1-b_1d_2)<0$, then $\lb_0<0$ 
 is outside the range $\lb\geq 0$.
See samples for these two cases in \cref{f:aniso}. 
We note that the explicit flows \eqref{e:sol} with $|\k|=k_0$ are all exponentially decaying for $\lb>\lb_0$ and exponentially growing for $0\leq\lb<\lb_0$.

\begin{figure}[t!]
\centering
\subfigure[$0<\lb_0<\lb_\crit$]{\includegraphics[trim = 2.5cm 8.5cm 3.5cm 9cm, clip, height=4cm]{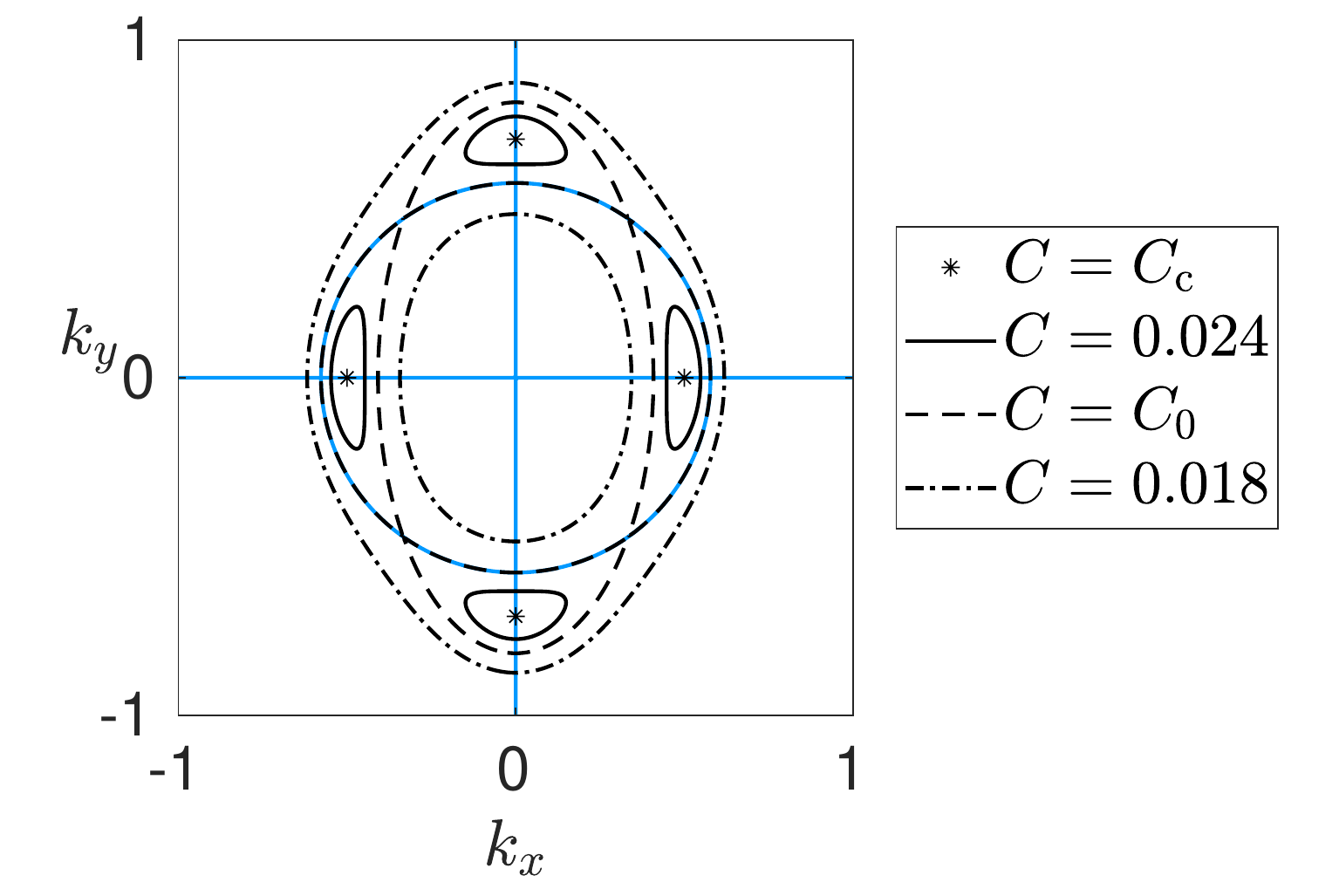}\label{f:anisosteadyloci4}}
\hfil
\subfigure[$\lb_0<0$]{\includegraphics[trim = 2.5cm 8.5cm 3.5cm 9cm, clip, height=4cm]{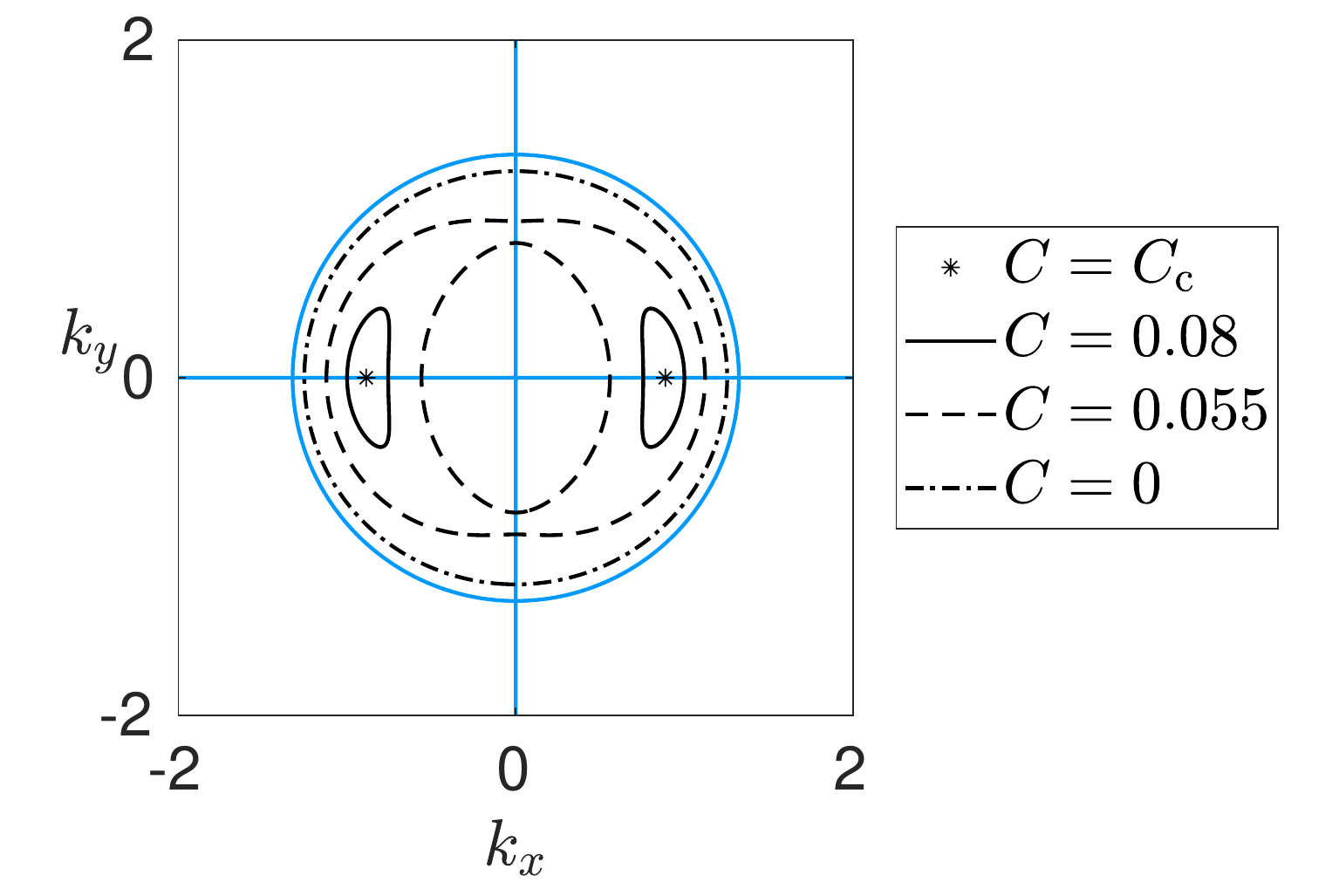}\label{f:anisosteadyloci}}
\caption{Sample of loci of the wave vectors of steady explicit flows \eqref{e:sol} for anisotropic backscatter on the intersections of blue curves and the different types of black curves. Black: $\beta=0$ with different types for different values of $\lb$; blue: the wave vectors satisfying \eqref{e:condb}. Common parameters: $f=0$, $g=9.8$, $H_0=0.1$. Other parameters: (a) $d_1=1$, $d_2=4$, $b_1=1$, $b_2=2$ so that $\lb_\crit=0.025$ and $\lb_0\approx0.022$; (b) $d_1=1$, $d_2=1.4$, $b_1=1.5$, $b_2=2.2$ so that $\lb_\crit\approx0.086$ and $\lb_0\approx-0.044$.}
\label{f:aniso}
\end{figure}
  
\begin{remark}
Since $\lb_0<\lb_\crit$, the primary bifurcation of the steady explicit flows upon decreasing $\lb$ does not occur for wave vectors on the circle with radius $k_0$. 
 This means, that these solutions do not correspond to a bifurcation analyzed in \Cref{s:bif}.
\end{remark}

In the rotational case $f\neq0$ it is more involved to determine the number of steady solutions or intersections of the corresponding curves.
If we consider the wave vectors in polar coordinates $\k=r\bigl(\cos\varphi,\sin\varphi\bigr)^\intercal$, then \eqref{e:cond} with $\beta=0$ becomes
\begin{subequations}\label{e:polarpoly}
\begin{align}
0&=a_1r^4-a_2r^2+\lb/H_0,\label{e:polarpolya}\\
0&=a_3r^4-a_4r^2+f,\label{e:polarpolyb}
\end{align}
\end{subequations}
with $a_1(\varphi):=d_1\sin^2\varphi+d_2\cos^2\varphi$, $a_2(\varphi):=b_1\sin^2\varphi+b_2\cos^2\varphi$, $a_3(\varphi):=\frac{d_1-d_2}{2}\sin(2\varphi)$ and $a_4(\varphi):=\frac{b_1-b_2}{2}\sin(2\varphi)$.
Concerning \eqref{e:polarpolya}, for $r>0$ and fixed angle $\varphi\in[0,2\pi)$ this fourth order polynomial has always two sign changes of the coefficients (or one for $C=0$). 
Due to the Descartes' rule of signs it thus has zero or two different positive roots (or one for $C=0$).
Since both equations \eqref{e:polarpoly} need to be satisfied, this means there are not more than two different steady solutions for a fixed direction $\varphi$ of their wave vector $\k$ (or one for $C=0$). 
 See \cref{r:2steady1direc}. 
We omit an analytical investigation of the total number of intersections for all wave vector directions.
Numerically, we have found up to four different steady solutions \eqref{e:sol} for fixed parameters, cf.\ \cref{f:fourroots}, which also shows the non-empty sets $\Gamma_1$ and $\Gamma_{-1}$.

\begin{remark}
\label{r:anisobifflows}
The explicit flows \eqref{e:sol} are not contained in the reduction of \cref{r:anisoscalar}, since wave vectors  of steady flows are not constrained to the axes (cf.\ \cref{f:anisosteadyloci4}) and have unconstrained amplitude. 
The same applies for the exponentially growing flows, which do not exist for wave vectors on the axes and have unconstrained growth, cf.\ \cref{f:growth}. 
\end{remark}

\begin{remark}
\label{r:2steady1direc}
Finding parameters for \cref{f:fourroots} is based on solving both equations in \eqref{e:polarpoly} 
simultaneously with two different values $r>0$.
Due to \eqref{e:polarpolyb} and $f\neq0$, this requires $\sin(2\varphi)\neq0$.
The graphs of the polynomials \eqref{e:polarpoly} share the same symmetry axis for $\frac{a_2}{2 a_1} = \frac{a_4}{2 a_3}$, which is equivalent to $b_2d_1 = b_1d_2$. 
Then, in order to get the same real roots we may vary $\lb$, which appears in \eqref{e:polarpolya} only, or $f$, which appears in  \eqref{e:polarpolyb} only. 
An example of this situation is depicted in \cref{f:fourroots}. \newline
\end{remark}

\section{Numerical bifurcation analysis}\label{s:num}

In order to illustrate and corroborate the analytical results, we present some numerical computations. In particular, we confirm the analytically predicted branches of nonlinear  GE and IGWs by numerical continuation. For this we have implemented \eqref{e:sw} for $y$-independent solutions in the matlab package pde2path \cite{p2p,p2pbook}. We thus consider solutions that depend on the $x$-variable only, which is scaled so that the onset of instability occurs on the normalized domain $[0,2\pi]$ with periodic boundary conditions. We plot some of the results in \cref{f:numbifiso} for an isotropic case, and in \cref{f:numbifaniso} for an anisotropic case. As analytically predicted, supercritical bifurcations of GE and IGWs occur near $\lb=\lb_\crit$, i.e.\ the bifurcating branches emerge towards decreasing $\lb$. In all cases, we found that the branches extend (after two folds in the isotropic case) to $\lb=0$, i.e.\ purely nonlinear bottom drag. We do not show the various bifurcation points that are numerically detected along the branches.

Since the numerical discretization has a minimal resolution, there is a spectral gap for large wave numbers, and we can directly consider stability of the bifurcating waves. We recall that for $f=0$ the modulation equations for GE  \eqref{e:GLf0b} allow to analytically predict stability with respect to perturbations of plane wave type, and only the sideband modes are relevant. In the more relevant case $f\neq0$ this reduction is not available and we present numerical results for $f=0.3$, including perturbations that are not of plane wave type. 

For this setting and with isotropic backscatter, where steady and oscillatory modes are simultaneously critical, we find that these bifurcating solutions are all unstable (see \cref{f:numbifiso}). In fact, the instability occurs already for purely $x$-dependent perturbations of the same wave number as the solutions, i.e.\ for the PDE posed on $[0,2\pi]$ directly. For both the GE and IGWs, the unstable eigenvalues near bifurcation are a complex conjugate pair. The unstable eigenfunction for GE has the shape of an IGWs, and vice versa, suggesting that the instability stems from the interaction between GE and IGWs. 

For anisotropic backscatter only steady modes are critical at the primary instability, and GE bifurcate supercritically as predicted near $\lb=\lb_\crit$. As expected, we find that IGWs bifurcate at some smaller value of $\lb$ 
(see \cref{f:numbifaniso}). 
Interestingly, the bifurcating nonlinear GE appear to be spectrally stable against general 2D perturbations.
This spectral stability numerically persists until $\lb\approx 0.01$ (where $\lb_\crit=0.1$), and unstable spectrum occurs when further deacreasing $\lb$ towards $\lb=0$. In order to determine spectral stability, we have implemented a Floquet-Bloch transform in the $x$-direction by replacing $\partial_x$ with $\partial_x-\rmi\gamma_x$, with Floquet-Bloch wave number parameter $\gamma_x\in[-\pi,\pi]$ in the linearization. Since the waves are constant in the $y$-direction, for perturbations in this direction we use a direct Fourier transform with wave number $\gamma_y\in\R$. By checking a grid of $(\gamma_x,\gamma_y)$ values, we found that the most unstable growth rate is zero and stems from the translation eigenmode at $\gamma_x=\gamma_y=0$. In particular, the sidebands are stable as plotted in \cref{f:numbifaniso}(c,d). This suggests that the combined backscatter and bottom drag can stabilize GE, meaning that backscatter not only induces the bifurcation of waves, but also promotes a dynamic selection of balanced states. This would further negatively impact a neutral energy redistribution by backscatter. 

\begin{figure}[t!]
\centering
\subfigure[]{\includegraphics[width=0.24\textwidth]{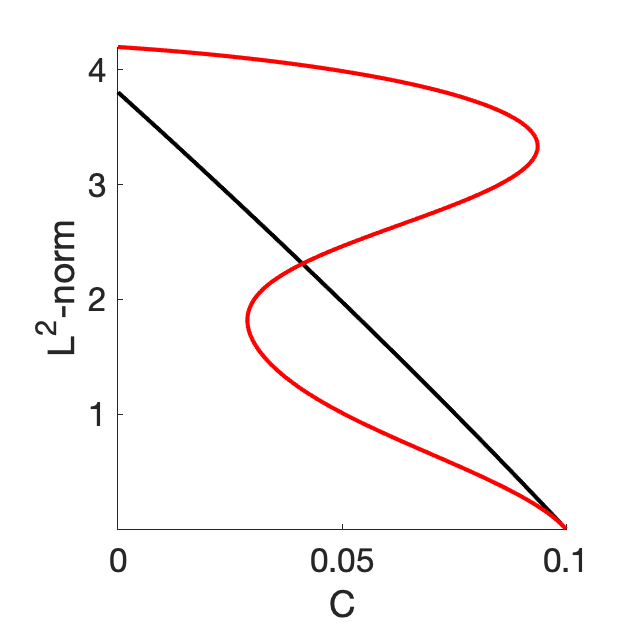}}
\hfil
\subfigure[]{\includegraphics[width=0.24\textwidth]{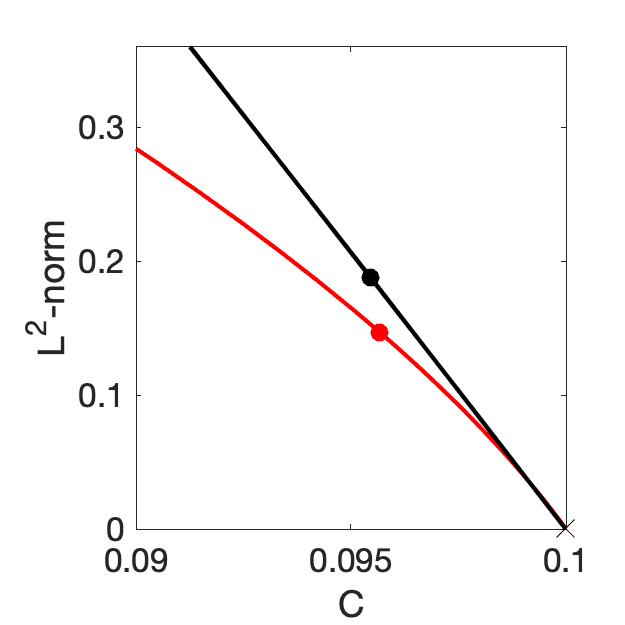}}
\hfil
\subfigure[]{\includegraphics[width=0.24\textwidth]{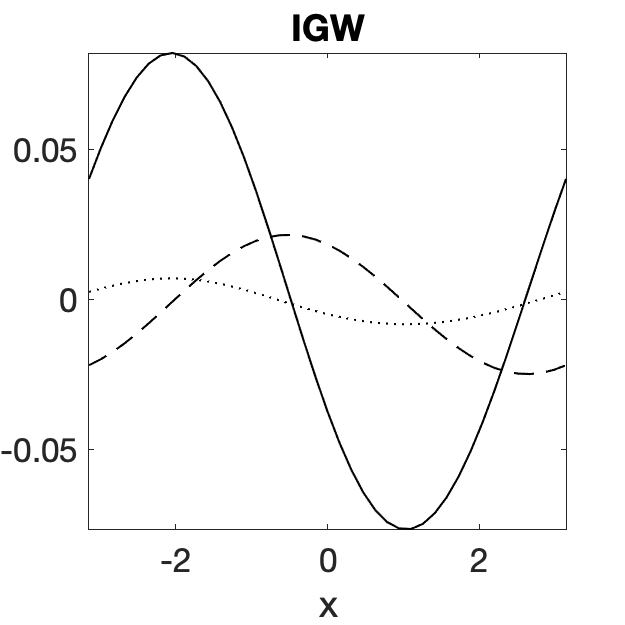}}
\hfil
\subfigure[]{\includegraphics[width=0.24\textwidth]{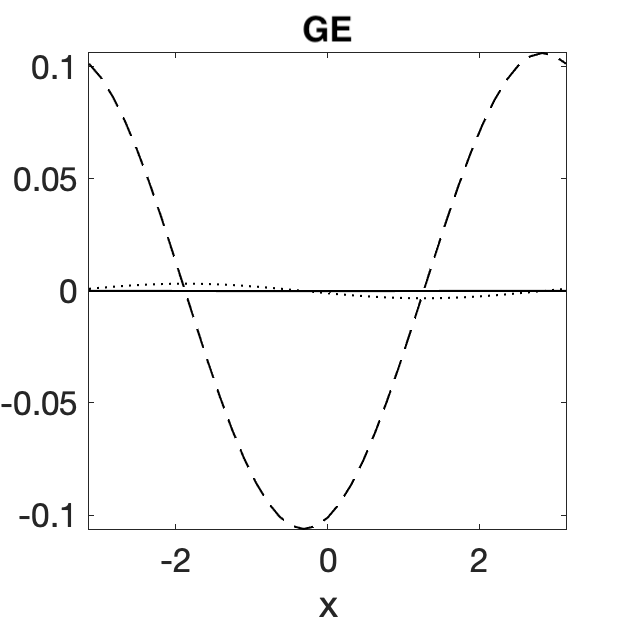}}
\caption{
(a) Numerically computed branches of nonlinear GE (black) and IGWs (red) in the isotropic case $d_1=d_2=1$, $b_1=b_2=2$; (b) magnification of (a). 
Marked solutions are plotted in panels (c) and (d) with $v_1$ solid, $v_2$ dashed, $\eta$ dotted. 
Other parameters are $f=0.3$, $g=9.8$, $H_0=0.1$ so that $\lb_\crit=0.1$, and $Q=0.05$. 
}
\label{f:numbifiso}
\end{figure}

\begin{figure}[t!]
\centering
\subfigure[]{\includegraphics[width=0.24\textwidth]{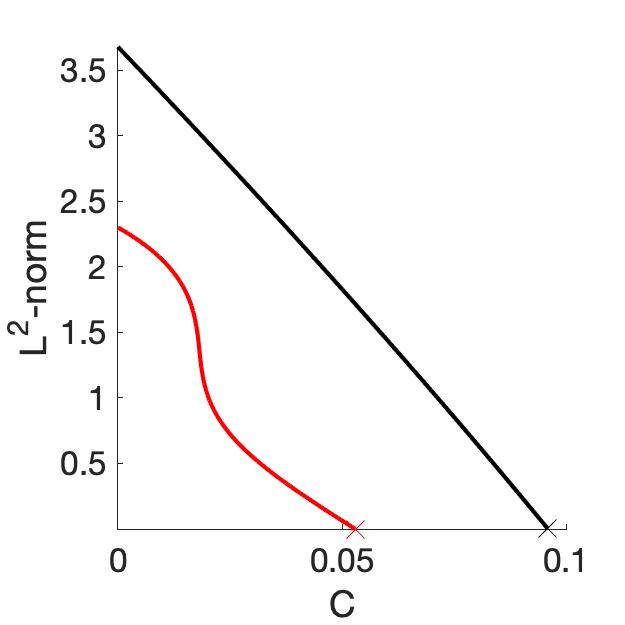}}
\hfil
\subfigure[]{\includegraphics[width=0.24\textwidth]{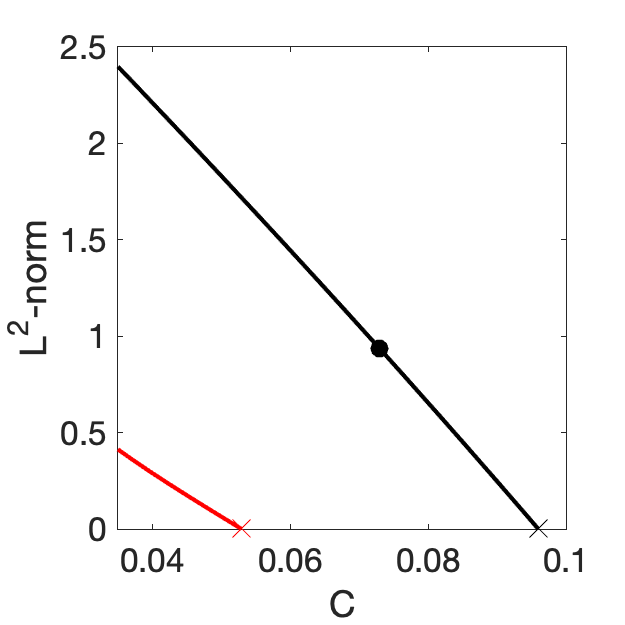}}
\hfil
\subfigure[]{\includegraphics[width=0.24\textwidth]{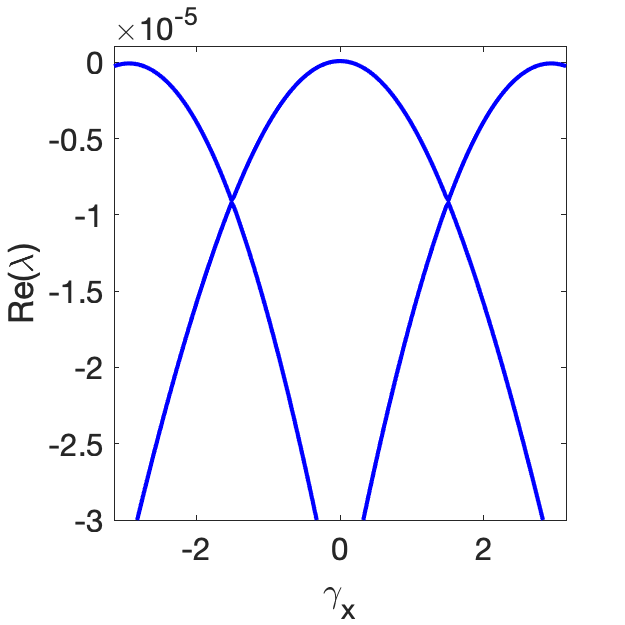}}
\hfil
\subfigure[]{\includegraphics[width=0.24\textwidth]{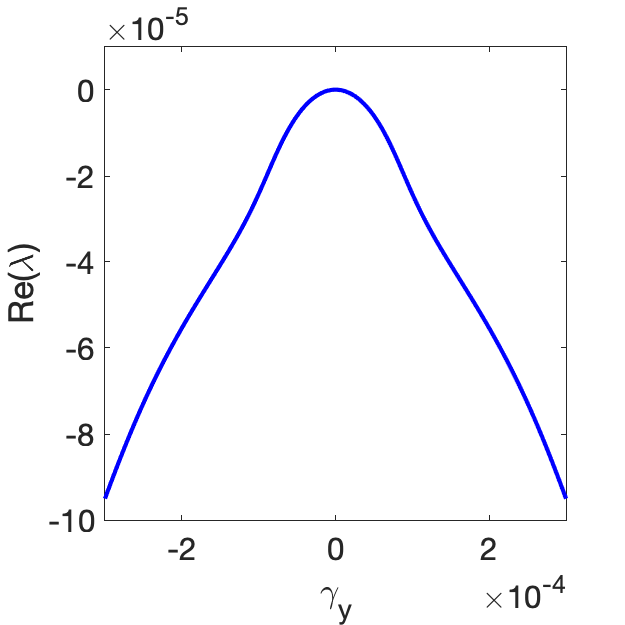}}
\caption{
(a) Numerically computed branches of nonlinear GE (black) and IGWs (red) in the anisotropic case $d_1=1$, $d_2=1.04$, $b_1=1.5$, $b_2=2$; (b) magnification of (a). 
Other parameters are $f=0.3$, $g=9.8$, $H_0=0.1$ so that $\lb_\crit=0.1$, and $Q=0.05$. 
(c) Floquet-Bloch spectrum near zero for $\gamma_x\in[-\pi,\pi]$ and (d) Fourier spectrum for $\gamma_y\approx 0$ of the marked solution in (b), suggesting spectral stabilty.
}
\label{f:numbifaniso}
\end{figure}

\section{Discussion}\label{s:discuss} 
We have mathematically studied the impact of kinetic energy backscatter on certain waves, flows and equilibria of \eqref{e:sw}. Our results show how, in an idealized setting, backscatter generates and selects geostrophic equilibria and inertia-gravity-type waves. We found that effective anisotropy of backscatter plays a special role and can further stabilize balanced states. In addition, we have shown that moderate bottom drag does not suppress the phenomenon of linear dynamics and associated unbounded growth found without drag in \cite{PRY22}. 

While the unbounded growth likely is not robust under numerical discretization, it still implies that for fine discretizations large, but possibly finite, growth can occur. In contrast, we expect the bifurcation results are robust under discretization, including stability on finite domains. However, the robustness of stability of the bifurcating solutions for arbitrarily large domains is subtle and it would be interesting to pursue this further.  We recall that the PDE analysis is hampered by the lack of spectral gap, which already raises questions of well-posedness. This could be alleviated by cutting off the backscatter at some large wave number, so that a center manifold reduction should be possible. Concerning the dynamics on such a center manifold, the most interesting case is the unfolding of the isotropic case with its simultaneous pattern forming stationary and oscillatory modes. This is reminiscent of `Turing-Hopf' instabilities in reaction-diffusion systems, where a Turing instability simultaneously occurs with an oscillatory instability \cite{DeWitEtAl,Ledesma-Duran2019-sd}. In this context complicated, `chaotic', dynamics can occur. We are not aware of studies that include a conservation law or that are in a fluid context. It would be interesting to see whether a connection to the dynamics of \eqref{e:sw} can be made. 

Another interesting direction would be consider multiple coupled layers, each modelled by a shallow water equation, as is customary in geophysical model studies. This would be one step closer to draw a connection to practical use of backscatter. 

\appendix

\section{Signs of $a_1, a_2, a_1 a_2-a_3$}\label{s:critical}

Since $a_1$ is a quadratic polynomial in $K:=|\k|^2$ with a positive quadratic coefficient, it possesses a global minimum which is positive for $\lb > \lb_1:=(b_1+b_2)^2H_0/(8(d_1+d_2))$. Since $\lb_\crit>\lb_1$, we have $a_1>0$ for $\lb\geq\lb_\crit$ and $\k\in\R^2$. We also find that $a_2>0$ for $\lb\geq\lb_\crit$: Without loss of generality, assume $\lb_\crit = b_1^2H_0/(4d_1)$; the global minimum of $d_jK^2-b_jK+\lb/H_0$ equals to $4d_j\lb/H_0-b_j^2$ which is non-negative for $\lb\geq\lb_\crit$, $j=1,2$. Hence $a_2>0$ for $\lb\geq\lb_\crit$ and $\k\in\R^2$. 

We next show that $a_1a_2-a_3>0$ for all $\k\in\R^2$ and $\lb\geq\lb_\crit$.  The previous implies that $a_1,a_2>0$ for all $\k\in\R^2$ and $\lb\geq\lb_\crit$. Concerning $a_3$, for $\lb>\lb_\crit$ we have $a_3>0$ for all $\k\in\R^2$, and for $\lb=\lb_\crit$ it holds that $a_3>0$ for all $\k\neq\k_\crit$ and $a_3=0$ for $\k=\k_\crit$. We make the dependence on $\k$ of $a_3=a_3(k_x,k_y)$ explicit in the following. Since $a_1>0$ for all $\k\in\R^2$ and $\lb\geq\lb_\crit$ we may compute 
\begin{align*}
\frac{a_3(k_x,k_y)}{a_1}=&\ \frac{g H_0 |\k|^2 \left((d_1|\k|^2-b_1)|\k|^2+(d_2|\k|^2-b_2)|\k|^2 + 2C/H_0 \right)}{(d_1 + d_2) |\k|^4 - (b_1 + b_2) |\k|^2 + 2C/H_0}\\[2mm]
&\ -\frac{g H_0 |\k|^2 \left((d_1|\k|^2-b_1)k_x^2+(d_2|\k|^2-b_2)k_y^2 + C/H_0 \right)}{(d_1 + d_2) |\k|^4 - (b_1 + b_2) |\k|^2 + 2C/H_0}\\[2mm]
=&\ g H_0 |\k|^2-\frac{a_3(k_y,k_x)}{a_1};
\end{align*}
note that the wave vector components in $a_3$ on the right-hand side are swapped. 
It follows for all $\k\in\R^2$ and $\lb\geq\lb_\crit$ that
\begin{align*}
a_2-\frac{a_3(k_x,k_y)}{a_1}&=(d_1 |\k|^4 - b_1 |\k|^2 + C/H_0) (d_2 |\k|^4 - b_2 |\k|^2 + C/H_0) + \frac{a_3(k_y,k_x)}{a_1} + f^2\\
& \geq f^2 > 0,
\end{align*}
since the polynomials in $|\k|$ in the brackets are non-negative for $\lb\geq\lb_\crit$. In particular, $a_1a_2-a_3>0$ for all $\k\in\R^2$ and $\lb\geq\lb_\crit$, since $a_1>0$ for these parameters. 

\section{Proof of even parity for $W$ in proof of \cref{t:bifQn0}}\label{s:proofproj0}

Analogous to the proof of \cref{t:bifQ0}, we rewrite \eqref{e:PG} as the fixed point equation for $W$ given by
\begin{equation}\label{e:fixGQn0}
P\calL_0 W = -P(L_\mu \phi + N_\lb(\phi,\mu) + N_Q(\phi,\mu)).
\end{equation}
We consider the even function $\phi$. Since the operators $L_\mu, N_\lb, N_Q$ map even functions to odd functions, and the projection $P$ \eqref{e:projP} maps odd functions to odd functions. 
Hence the term on the right-hand side of \eqref{e:fixGQn0} is odd. We write the periodic function $W$ as a Fourier series $W = \sum_{\ell\in\Z}w_\ell \rme^{\rmi\ell\xi}$ with $w_\ell\in\C$, and write the odd periodic function on the right-hand side of \eqref{e:fixGQn0} as $\sum_{m=1}^{\infty} R_m \sin(m\xi)$ with $R_m\in\R$, then \eqref{e:fixGQn0} becomes
\[
\sum_{\ell\in\Z} \widehat P\widehat\calL_0(\ell) w_\ell\rme^{\rmi\ell\xi} = \sum_{m=1}^{\infty} R_m \sin(m\xi),
\]
where $\widehat\calL_0(\ell) = \rmi\ell(dk_\crit^4\ell^4 - bk_\crit^2\ell^2 + \lb/H_0)$ which is an odd function in $\ell$, and $\widehat P = \Id$ since $W\in \setM$. Hence, $\widehat P\widehat\calL_0(\ell) = \widehat\calL_0(\ell)$, and its inverse is also odd. We project the both sides of the above equation onto $\rme^{\rmi\ell\xi}$ and $\rme^{-\rmi\ell\xi}$ (with the inner product $\langle \cdot,\cdot \rangle = \langle \cdot,\cdot \rangle_{\Lspace^2}$), respectively
\begin{align*}
\widehat\calL_0(\ell) w_\ell = \langle R_\ell \sin(\ell\xi),\rme^{\rmi\ell\xi}\rangle\quad &\Leftrightarrow \quad w_\ell = (\widehat\calL_0(\ell))^{-1}\langle R_\ell \sin(\ell\xi),\rme^{\rmi\ell\xi}\rangle, \\
\widehat\calL_0(-\ell) w_{-\ell} = \langle R_\ell \sin(\ell\xi),\rme^{-\rmi\ell\xi}\rangle \quad &\Leftrightarrow \quad w_{-\ell} = (\widehat\calL_0(-\ell))^{-1}\langle R_\ell \sin(\ell\xi),\rme^{-\rmi\ell\xi}\rangle.
\end{align*}
Since $\langle R_\ell \sin(\ell\xi),\rme^{\rmi\ell\xi}\rangle = -\langle R_\ell \sin(\ell\xi),\rme^{-\rmi\ell\xi}\rangle \in\rmi\R$, we have $w_{\ell} = w_{-\ell}\in\R$. It follows that $W$ is even and real.

\section{Proof of estimate \eqref{e:estW} and \eqref{e:WestX}}\label{s:proofWestX}

The proof relies on rewriting \eqref{e:PG} or \eqref{e:PG2} as a fixed point equation for $W$. We first write $PG(V;\mu) = L_0 V  + L_\mu V + \calN_\mu(V)$ with nonlinear part $\calN_\mu(V) = \calO(\|V\|_X^2)$ and  corresponding solution space $X$, as well as linear parts $L_0$, 
which is $\mu$-independent, 
and $L_\mu$ as the perturbation by parameters $\mu$, i.e.\ $\|L_\mu\| = \calO(|\mu|)$. 
More precisely, for \Cref{s:bif} these are $V=\phi$, $L_0=P\calL_0$, $L_\mu=PL_\mu$ and $\calN_\mu=P(N_\lb + N_Q)$,  
while for \Cref{s:igw} $V=U$, $L_0=P\calL_\crit$, $L_\mu=P(\calL_\mu-\calL_\crit)$ and $\calN_\mu=P(- B_Q - B - N)$ from \eqref{e:igwprob}. 
Using $V = u +W$, $L_0u=0$ and that $L_0$ is boundedly invertible from $\setM$ to $\range(\calL_0)$ or $\range(\calL_\crit)$, we rewrite $PG=0$ as the fixed point equation for $W$ given by
\begin{equation}\label{e:fixW}
-(L_0+L_\mu)W = L_\mu u + \calN_\mu(u+W) \Leftrightarrow W = -(L_0+L_\mu)^{-1} (L_\mu u + \calN_\mu(u+W)).
\end{equation}
Since a priori $\|W\| = \calO(|\mu| + \|u\|_X)$, and $L_0+L_\mu$ is boundedly invertible for sufficiently small $|\mu|$, $\|u\|_X$, we find constants $C_1, C_2>0$ such that 
\begin{align*}
 \|(L_0+L_\mu)^{-1}L_\mu u\|_X &\leq C_1 |\mu|\|u\|_X,\\
 \|(L_0+L_\mu)^{-1}\calN_\mu(u+W)\| &\leq C_2 (\|u\|_X^2 + \|W\|_X^2).
\end{align*}
From \eqref{e:fixW} we then obtain
\begin{align}
\|W\|_X &\leq C_1|\mu|\|u\|_X + C_2 (\|u\|^2_X + \|W\|_X^2)\nonumber \\
\Rightarrow \|W\|_X (1-C_2 \|W\|_X) &\leq C_1|\mu|\|u\| _X+ C_2 \|u\|_X^2.\label{e:West2}
\end{align}
Choosing $\mu, u$ sufficiently small gives $C_2 \|W\|_X\leq \frac 1 2$, and then \eqref{e:West2} implies
\[
\|W\|_X \leq 2C_1|\mu|\|u\|_X + 2C_2 \|u\|_X^2,
\]
which proves \eqref{e:estW} and \eqref{e:WestX}.

\section*{Acknowledgments}
The authors A. Prugger and J. Yang are grateful for the financial support from the Collaborative Research Centre TRR 181 ``Energy Transfers in Atmosphere and Ocean'' and the hospitality from Jacobs University Bremen during part of the creation of this paper. We thank the anonymous reviewers for comments that helped to improve the manuscript.


\end{document}